\pdfoutput=1
\documentclass[acmsmall,nonacm,10pt]{acmart}
\usepackage[mathletters]{ucs} %
\usepackage[utf8x]{inputenc}

\def\BibTeX{{\rm B\kern-.05em{\sc i\kern-.025em b}\kern-.08emT\kern-.1667em\lower.7ex\hbox{E}\kern-.125emX}}

\setcopyright{acmcopyright}
\acmJournal{TOCL}
\usepackage{mathtools}%
\usepackage{amsmath}
\usepackage{stmaryrd} %

\usepackage{amssymb}  %
\usepackage{galois}   %
\usepackage{paralist} %
\usepackage{multicol}
\usepackage{changepage} %
\usepackage{scalerel}

\usepackage{tikz}
\usetikzlibrary{arrows}
\usetikzlibrary{positioning}
\usetikzlibrary{automata}

\usepackage[ruled,noend,noline,linesnumbered]{algorithm2e}
\makeatletter
\newcommand{\RemoveAlgoNumber}{\renewcommand{\fnum@algocf}{\AlCapSty{\AlCapFnt\algorithmcfname}}}
\makeatother
\newcommand*{\AlgRegularW}{\ensuremath{\mathtt{FAIncW}}\xspace}
\newcommand*{\AlgRegularWr}{\ensuremath{\mathtt{FAIncWr}}\xspace}
\newcommand*{\AlgRegularA}{\ensuremath{\mathtt{FAIncS}}\xspace}
\newcommand*{\AlgGrammarW}{\ensuremath{\mathtt{CFGIncW}}\xspace}
\newcommand*{\AlgGrammarA}{\ensuremath{\mathtt{CFGIncS}}\xspace}
\newcommand*{\AlgRegularGfp}{\ensuremath{\mathtt{FAIncGfp}}\xspace}

\newtheorem{theorem}{Theorem}[section]
\newtheorem{remark}[theorem]{Remark}

\newcommand{\commentall}[1]{}

\usepackage{esvect} 
\usepackage{bm}
\newcommand{\vect}[1]{\vv{\bm{#1}}}
\newcommand{\vectarg}[2]{\vv{\bm{#1}}^{{\!{\scriptstyle {#2}}}}}
\newcommand{\minor}[1]{\lfloor {#1}\rfloor}

\newcommand{\tuple}[1]{\langle {#1}\rangle}

\newcommand{\lang}[1]{{\mathcal{L}(#1)}} %
\DeclareMathOperator{\Pre}{{Pre}}
\DeclareMathOperator{\ctx}{{ctx}}
\DeclareMathOperator{\Post}{{Post}}
\DeclareMathOperator{\CPre}{{CPre}}
\DeclareMathOperator{\cpre}{{cpre}}
\DeclareMathOperator{\absincl}{{Incl^{\sharp}}}

\DeclareMathOperator{\Conv}{{Conv}}
\DeclareMathOperator{\Eq}{{Eq}}
\DeclareMathOperator{\Incl}{{Incl}}

\DeclareMathOperator{\base}{{\textit{b}}}
\DeclareMathOperator{\minim}{{min}}

\DeclareMathOperator{\AC}{{AC}}
\DeclareMathOperator{\Kleene}{{\normalfont \textsc{Kleene}}}

\DeclareMathOperator{\wsqsubseteq}{{\widetilde{\sqsubseteq}}}
\DeclareMathOperator{\wsqcup}{{\widetilde{\sqcup}}}
\DeclareMathOperator{\wsqcap}{{\widetilde{\sqcap}}}
\DeclareMathOperator{\wbigsqcup}{{\widetilde{\bigsqcup}}}

\newcommand{\nullable}[3]{{\ensuremath{\psi^{#2}_{#3}(#1)}}}

\newcommand{\fp}{\mathcal{FP}}

\newcommand{\bN}{\mathbb{N}}

\newcommand{\cA}{\mathcal{A}}
\newcommand{\cB}{\mathcal{B}}
\newcommand{\cG}{\mathcal{G}}

\newcommand{\cO}{\mathcal{O}}

\newcommand{\cV}{\mathcal{V}}
\newcommand{\cL}{\mathcal{L}}

\newcommand{\id}{\mathrm{id}}
\newcommand{\ud}{\triangleq}

\newcommand{\Lra}{\Leftrightarrow}
\newcommand{\Ra}{\Rightarrow}

\newcommand{\ra}{\rightarrow}

\DeclareMathOperator{\uco}{uco}

\DeclareMathOperator{\pre}{pre}
\DeclareMathOperator{\post}{post}

\DeclareMathOperator{\lfp}{lfp}
\DeclareMathOperator{\gfp}{gfp}

\DeclareMathOperator{\Eqn}{{Eqn}}
\DeclareMathOperator{\Eqnr}{{Eqn}^{r}}
\DeclareMathOperator{\Fn}{{Fn}}
\newcommand{\eox}{\hfill{\ensuremath{\Diamond}}}
\newcommand{\inqm}{\in^{\scaleto{?}{3.5pt}}}
\newcommand{\eqqm}{=^{\scaleto{?}{3.5pt}}}
\newcommand{\subseteqm}{\subseteq^{\scaleto{?}{3.5pt}}}

\DeclareFontFamily{U}{mathx}{\hyphenchar\font45}
\DeclareFontShape{U}{mathx}{m}{n}{<-> mathx10}{}

\newcommand{\cd}[1]{\text{\lstinline\(#1\)}}

\newcommand{\udr}{\stackrel{{\mbox{\tiny\ensuremath{\triangle}}}}{\Longleftrightarrow}}
\newcommand{\udrshort}{\stackrel{{\mbox{\tiny\ensuremath{\triangle}}}}{\Leftrightarrow}}

\def\ok#1{\mbox{\raisebox{0ex}[1ex][1ex]{$#1$}}}

\newcommand{\goes}[1]{\stackrel{#1}{\leadsto}}
\newcommand{\ggoes}[1]{\stackrel{#1}{\rightarrow}}

\makeatletter
\newcommand{\specialcell}[1]{\ifmeasuring@#1\else\omit$\displaystyle#1$\ignorespaces\fi}
\makeatother

\allowdisplaybreaks

\begin{document}
\sloppy
\title{Complete Abstractions for Checking Language Inclusion}

\author{Pierre Ganty}
\email{pierre.ganty@imdea.org} 
\orcid{0000-0002-3625-6003}
\affiliation{%
  \institution{IMDEA Software Institute}
   \city{Madrid}
   \country{Spain}
}

\author{Francesco Ranzato}
\email{francesco.ranzato@unipd.it} 
\orcid{0000-0003-0159-0068}
\affiliation{%
  \institution{Dipartimento di Matematica, University of Padova} 
  \city{Padova} 
  \country{Italy}
}
             
\author{Pedro Valero}
\email{pedro.valero@imdea.org} 
\orcid{0000-0001-7531-6374}
\affiliation{%
  \institution{IMDEA Software Institute}
  \city{Madrid}
  \country{Spain}
}

\begin{abstract}
We study the language inclusion problem \(L_1 \subseteq L_2\) where \(L_1\) is regular or context-free.
Our approach relies on abstract interpretation and checks whether an overapproximating abstraction of \(L_1\), obtained by overapproximating the Kleene iterates of its least fixpoint characterization, is included in \(L_2\). 
We show that a language inclusion problem is decidable whenever this overapproximating abstraction satisfies a completeness condition (i.e., its loss of precision causes no false alarm) and prevents infinite ascending chains (i.e., it guarantees 
termination of least fixpoint computations).
This overapproximating abstraction of languages can be defined using quasiorder relations on words, where the abstraction gives the language of all the words ``greater than or equal to'' a given input word for that quasiorder. 
We put forward a range of such quasiorders that allow us to systematically design decision procedures for different language inclusion problems such as regular languages into regular languages or into trace sets of one-counter nets, and  context-free languages into regular languages.
In the case of inclusion between regular languages, some of the induced inclusion checking procedures correspond to well-known state-of-the-art algorithms like the so-called antichain algorithms.
Finally, we provide an equivalent language inclusion checking algorithm based on 
a greatest fixpoint computation 
that relies on quotients of languages and, to the best of our knowledge, was not previously known.  
\end{abstract}



\maketitle     

\section{Introduction}%
\label{sec:introduction}

Language inclusion is a fundamental and classical problem \cite[Chapter~11]{HU79} which consists in deciding, given two languages \(L_1\) and \(L_2\), whether \(L_1\subseteq L_2\) holds. Language inclusion problems are found in diverse fields ranging 
from compiler construction \cite{bauer76,waite84} to model checking \cite{baier08,clarke18}.
We consider languages of finite words over a finite alphabet $\Sigma$.
For regular and context-free languages, the inclusion problem is well known to be PSPACE-complete (see \cite{hunt}). 

The basic idea of our approach for solving a language inclusion problem $L_1\subseteq L_2$ is to leverage
Cousot and Cousot's abstract interpretation~\cite{CC77,CC79} for checking the inclusion of an overapproximation (i.e., a superset) of \(L_1\) into \(L_2\). 
This idea draws inspiration from the work of Hofmann and Chen~\shortcite{Hofmann2014}, who used abstract interpretation to decide language inclusion between languages of infinite words.

Let us assume that \(L_1\) is specified as least fixpoint of an equation system $X=F_{L_1}(X)$ on sets of words in $\wp(\Sigma^*)$, that is, $L_1 = \lfp(F_{L_1})$
is viewed as limit of the possibly infinite sequence of Kleene iterates $\{F^n_{L_1}(\varnothing)\}_{n\in \bN}$ of the transformer $F_{L_1}$. 
An approximation of $L_1$ is obtained by applying an overapproximation 
for sets of words as modeled by a closure operator 
\(\rho:\wp(\Sigma^*)\ra \wp(\Sigma^*)\). In abstract interpretation
one such closure $\rho$ logically defines an \emph{abstract domain}, which is here 
used  for 
overapproximating a language by adding 
words to it, possibly none in case of no approximation. 
The language abstraction $\rho$ is then used 
for defining an abstract check of convergence  
for the Kleene iterates of $F_{L_1}$ whose limit is 
$L_1$, i.e., 
the convergence of the sequence $\{F^n_{L_1}(\varnothing)\}_{n\in \bN}$
is checked on the abstraction $\rho$ by the condition $\rho(F^{n+1}_{L_1}(\varnothing))
\subseteq \rho(F^{n}_{L_1}(\varnothing))$. If the abstraction $\rho$ does not contain infinite ascending chains then we obtain finite convergence w.r.t.\ this abstract check for some 
$F_{L_1}^N(\varnothing)$.  
 
Therefore, this abstract interpretation-based 
approach  finitely computes
an abstraction $L_1^\rho = \rho(F_{L_1}^N(\varnothing))$ such
that  the abstract language inclusion check 
$L_1^\rho \subseteq L_2$ is 
\emph{sound} because 
$L_1 \subseteq L_1^\rho$ always holds. 
We then give conditions on \(\rho\) which ensure a \emph{complete} abstract inclusion 
check, namely, the answer to $L_1^\rho \subseteq L_2$ is always exact (no ``false alarm'' in abstract interpretation terminology):
\begin{enumerate}[\upshape(i\upshape)]
\item $L_2$ is \emph{exactly represented} by the abstraction $\rho$, i.e., \(\rho(L_2)=L_2\);
\item \(\rho\) is a \emph{complete abstraction} for symbol concatenation $\lambda X\in \wp(\Sigma^*).\,aX$, for all $a\in \Sigma$, 
according to the standard notion of completeness in abstract interpretation~\cite{CC77}; this entails that 
$\rho(L_1) = L_1^\rho$ holds, so that $L_1^\rho \not\subseteq L_2$ implies $L_1 \not\subseteq L_2$.
\end{enumerate}
This approach leads us to design a general algorithmic framework for language inclusion problems which is parameterized by an underlying language abstraction. %

We then focus on language abstractions $\rho$ which are 
induced by a quasiorder relation on words $\mathord{\leqslant}\subseteq \Sigma^*\times \Sigma^*$. Here, a language \(L\) is overapproximated by adding all the words which are ``greater than or equal to'' some word of \(L\) for $\mathord{\leqslant}$. This allows us to 
instantiate the above conditions (i) and (ii) 
for achieving a complete abstract inclusion check in terms of the quasiorder relation $\mathord{\leqslant}$.
Termination, which corresponds to having finitely many Kleene iterates, 
is guaranteed by requiring that 
the relation $\mathord{\leqslant}$ is a \emph{well-quasiorder}.

We define well-quasiorders satisfying the conditions (i) and (ii) which are directly derived from the standard Nerode equivalence relations on words.
These quasiorders have been first investigated by Ehrenfeucht et al.~\shortcite{ehrenfeucht_regularity_1983} and have been later generalized and
extended by de Luca and Varricchio~\shortcite{deLuca1994,deluca2011}.
In particular, 
drawing from a result by de Luca and Varricchio~\shortcite{deLuca1994}, we show that the language abstractions induced by the Nerode quasiorders are the most general ones (intuitively, optimal) which fit in our algorithmic framework for checking 
language inclusion.  
While these quasiorder abstractions do not depend on some finite 
representation of languages (e.g., some class of 
automata), 
we provide quasiorders which instead exploit an underlying language representation given by a finite automaton.
In particular, by selecting suitable well-quasiorders for the class of language inclusion problems at hand we are able to systematically derive 
decision procedures of the inclusion problem 
$L_1\subseteq L_2$ for the following cases: 
\begin{enumerate}[\upshape(1\upshape)]
\item both \(L_1\) and \(L_2\) are regular; 
\item \(L_1\) is regular and \(L_2\) is the trace language of a one-counter net; 
\item \(L_1\) is context-free and \(L_2\) is regular.
\end{enumerate}

These decision procedures, here systematically designed 
by instantiating our framework,  
are then related to existing language inclusion checking algorithms. 
We study in detail the case where both languages $L_1$ and $L_2$ are regular and represented by finite state automata. 
When our decision procedure for $L_1\subseteq L_2$ is derived from 
a well-quasiorder on $\Sigma^*$ by exploiting an automaton-based representation of \(L_2\), it turns out that 
we obtain the well-known ``antichain algorithm'' by De Wulf et al.~\shortcite{DBLP:conf/cav/WulfDHR06}. 
Also, by including a simulation relation in the definition of the well-quasiorder we derive a decision procedure that partially matches the language inclusion algorithm by Abdulla et al.~\shortcite{Abdulla2010}, and in turn also that by Bonchi and Pous~\shortcite{DBLP:conf/popl/BonchiP13}.
It is also worth pointing out that for the case in which \(L_1\) is regular and \(L_2\) is the set of traces of a one-counter net, our systematic instantiation provides 
an alternative proof for the decidability of the corresponding language inclusion problem~\cite{JANCAR1999476}.
\\
\indent
Finally, we leverage a standard duality result between abstract least and greatest 
fixpoint checking~\cite{cou00} and put forward a \emph{greatest} fixpoint approach (instead of the above \emph{least} fixpoint-based procedures) for the case where both \(L_1\) and \(L_2\) are regular languages.
Here, we exploit the properties of the overapproximating abstraction induced by the quasiorder relation in order to show that 
the Kleene iterates converging to the greatest fixpoint are finitely many.
Interestingly, the Kleene iterates of the greatest fixpoint are finitely many whether you apply the overapproximating abstraction or not, and this is shown
by relying on a second type of completeness in abstract interpretation
called forward completeness~\cite{gq01}.

\subsubsection*{Structure of the Article} In Section~\ref{sec:background} we 
recall the needed basic notions and background on order theory, abstract interpretation and formal languages. 
Section~\ref{sec:inclusion_checking_by_complete_abstractions}
defines a general method for checking the convergence of Kleene iterates on
an abstract domain, which provides the basis for designing in
Section~\ref{sec:an_algorithmic_framework_for_language_inclusion_based_on_complete_abstractions} an abstract interpretation-based framework 
for checking language inclusion, in particular by relying on abstractions
that are complete for concatenation of languages. 
This general framework is instantiated in Section~\ref{sec:instantiating_the_framework_language_based_well_quasiorders} 
to the class of abstractions 
induced by well-quasiorders on words, thus yielding effective inclusion checking algorithms for regular languages and traces of one-counter nets. 
Section~\ref{sec:novel_perspective_AC} shows that 
one specific instance of our algorithmic framework turns out to be equivalent to the well-known
antichain algorithm for language inclusion by \citet{DBLP:conf/cav/WulfDHR06}.
The instantiation of the framework  for checking the inclusion of context-free languages into regular languages is described in Section~\ref{sec:context_free_languages}. 
Section~\ref{sec:greatest_fixpoint_based_algorithm} shows how
to derive a new language inclusion algorithm 
which relies on the computation of a greatest fixpoint rather than a least fixpoint.
Finally, Section~\ref{sec:conclusions} outlines some directions 
for future work.

This article is an extended and revised version of the conference paper~\cite{grv-sas2019}, 
that includes full proofs, additional detailed examples, a simplification 
of some technical notions, and a new application for checking the inclusion of context-free languages into regular languages.

\section{Background}%
\label{sec:background}

\subsection{Order Theory}%
If $X$ is any set then $\wp(X)$ denotes its powerset. 
If $X$ is a subset of some universe set $U$ then $X^c$ denotes the complement of $X$ with respect to $U$ when $U$ is implicitly
given by the context. 
If $f:X\ra Y$ is a function between sets and $S\in \wp(X)$ then $f(S)
\ud \{f(x) \in Y \mid x\in S\}$ denotes its image on a subset $S$. A composition of two functions $f$ and $g$ is denoted both by $fg$ and $f\comp g$.

\(\tuple{D,\mathord{\leqslant}}\) is a \emph{quasiordered set} (qoset) when \(\mathord{\leqslant}\) is a quasiorder (qo) relation on $D$, i.e.\ a reflexive and
transitive binary relation  \(\mathord{\leqslant}\subseteq D\times D\). In a qoset \(\tuple{D,\mathord{\leqslant}}\) we will 
also use the following induced equivalence relation $\sim_D$: for all
$d,d'\in D$, $d \sim_D d' \ok{\udrshort} d\leqslant d' \:\wedge\: d' \leqslant d$.  
A qoset satisfies the \emph{ascending} (resp.\ \emph{descending}) \emph{chain condition} (ACC, resp.\ DCC) if there is no countably infinite  sequence of distinct elements \(\{x_i\}_{i \in \mathbb{N}}\) such that, for all $i\in\bN$, \(x_i \leqslant x_{i{+}1}\) (resp. \(x_{i{+}1} \leqslant x_{i}\)).
A qoset is called ACC (DCC) when it satisfies the ACC (DCC).
\\
\indent
A qoset \(\tuple{D,\mathord{\leqslant}}\) is a \emph{partially ordered set} (poset) when \(\mathord{\leqslant}\) is antisymmetric.
A subset $X\subseteq D$ of a poset is \emph{directed} if $X$ is nonempty and every pair of elements in $X$ has an upper bound in $X$.
A poset \(\tuple{D,\mathord{\leqslant}}\) is a \emph{directed-complete partial order} (CPO) if it has 
the least upper bound (lub) of all its directed subsets. 
A poset is a \emph{join-semilattice} if it has the lub
of all its nonempty finite subsets (therefore 
binary lubs are enough). 
A poset is a \emph{complete lattice} if it has 
the lub of all its arbitrary (possibly empty) subsets; in this case, let us recall that 
it also has the greatest lower bound (glb) of all its 
arbitrary subsets. 
\\
\indent
An \emph{antichain} in a qoset \(\tuple{D,\mathord{\leqslant}}\) is a subset $X\subseteq D$ such that 
any two distinct elements in $X$ are incomparable for $\leqslant$. 
We denote the set of antichains of a qoset $\tuple{D,\mathord{\leqslant}}$ by \(\AC_{\tuple{D,\mathord{\leqslant}}} \ud \{X\subseteq D \mid X \text{ is an antichain}\}\).
A qoset \(\tuple{D,\mathord{\leqslant}}\) is a \emph{well-quasiordered set} (wqoset), and $\mathord{\leqslant}$ is called well-quasiorder (wqo) on $D$, when for 
every countably infinite  sequence of elements \(\{x_i\}_{i\in \bN}\) there exist \(i,j\in \bN\) such that \(i<j\) and \(x_i\leqslant x_j\). 
Equivalently,  \(\tuple{D,\mathord{\leqslant}}\) is a wqoset if{}f $D$ is DCC and 
$D$  has no infinite antichain. 
For every qoset \(\tuple{D,\mathord{\leqslant}}\), let us define the following binary relation $\sqsubseteq_\leqslant$ on the powerset: 
given \(X,Y\in \wp(D)\), 
\begin{equation}\label{ordering-definition}
X\sqsubseteq Y \:\udrshort\: \forall x\in X, \exists y\in Y,\; y\leqslant x \enspace .
\end{equation}
A \emph{minor} of a subset \(X \subseteq D\), denoted by \(\minor{X}\), is a subset of the minimal elements of \(X\) w.r.t.\ \(\leqslant\), i.e.\ \(\minor{X}\subseteq\minim_{\leqslant}(X) \ud \{x \in X \mid \forall y \in X, y \leqslant x \Ra y=x\}\), such that 
\(X \sqsubseteq \minor{X}\) holds. 
Therefore, a minor $\minor{X}$ of \(X\subseteq D\) is always an antichain in $D$. 
Let us recall that every subset $X$ of a wqoset \(\tuple{D,\leqslant}\) has at least one minor set, all minor sets of $X$ are finite, 
$\minor{\{x\}}=\{x\}$, $\minor{\varnothing}=\varnothing$, and
if \(\tuple{D,\mathord{\leqslant}}\) is additionally a poset then there exists exactly one minor set of $X$.
It turns out that \(\tuple{\AC_{\tuple{D,\mathord{\leqslant}}},\sqsubseteq}\) is a qoset, which is ACC if \(\tuple{D,\leqslant}\) is a wqoset and is a poset if \(\tuple{D,\leqslant}\) is a poset.

For the sake of clarity, we overload the notation and use the same symbol for a function/relation 
and its componentwise (i.e.\ pointwise) extension on product domains, e.g., if $f:X\ra Y$ then 
$f$ also denotes the standard product function $f:X^n \ra Y^n$ which is 
componentwise defined by 
$\lambda\tuple{x_1,\ldots,x_n}\in X^n.\tuple{f(x_1),\ldots,f(x_n)}$. 
A vector \(\vect{x}\) in some product domain \(D^{|S|}\) indexed by a finite set $S$ is also denoted by \(\tuple{x_i}_{i \in S}\) and, for some $i\in S$, 
\(\vect{x}_{\!\! i}\) denotes its component \(x_i\).

Let \(\tuple{X,\leqslant}\) be a qoset and \(f:X \ra X\) be a function. $f$ is \emph{monotonic} when 
$x\leqslant y$ implies $f(x) \leqslant f(y)$. For all $n\in \bN$, the $n$-th 
power \(f^n:X \ra X\) of $f$ is inductively defined by:
$f^0 \ud \lambda x.x$; $f^{n+1} \ud f \comp f^n $ (or, equivalently, 
$f^{n+1} \ud f^n \comp f$).
The denumerable sequence of \emph{Kleene iterates}  
of \(f\) starting from an initial value \(a\in X\) is given by $\tuple{f^n(a)}_{n\in \bN}$.
If \(\tuple{X,\leqslant}\) is a poset and $a\in X$ then 
\(\lfp_a(f)\) (resp.\ \(\gfp_a(f)\)) denotes the least 
(resp.\ greatest) fixpoint of $f$ 
which is greater (resp.\ less) than or equal to \(a\), when this exists; 
in particular, \(\lfp(f)\) (resp.\ \(\gfp(f)\)) denotes the least 
(resp.\ greatest) fixpoint of $f$, when this exists. 
If
\(\tuple{X,\leqslant}\) is an ACC (resp. DCC) CPO, \(a\leqslant f(a)\) (resp.\ \(f(a)\leqslant a\)) holds and \(f\) is monotonic
then the Kleene iterates 
$\tuple{f^n(a)}_{n\in \bN}$ \emph{finitely converge} to \(\lfp_a(f)\) 
(resp.\ \(\gfp_a(f)\)), i.e., 
there exists $k\in \bN$ such that for all $n\geq k$, $f^n(a)=f^k(a)=\lfp_a(f)$ 
(resp.\ $\gfp_a(f)$). 
In particular,  if  
$\bot$ (resp.\ $\top$) is the least (greatest) element of \(\tuple{X,\leqslant}\) then $\tuple{f^n(\bot)}_{n\in \bN}$ (resp.\ $\tuple{f^n(\top)}_{n\in \bN}$)
finitely converges to \(\lfp(f)\) (resp.\ \(\gfp(f)\)).

\subsection{Abstract Interpretation}\label{subsec:abstract-interpretation}
Let us recall some basic notions on closure operators and Galois Connections commonly used in abstract interpretation (see, e.g., \cite{CC79,mine17,rival-yi}).  
Closure operators and Galois Connections are equivalent notions 
and, therefore, they are both used for 
defining the notion of approximation in abstract interpretation, where closure operators allow us to define and reason on abstract domains independently of a specific representation for abstract values which is required by 
Galois Connections.

Let \(\tuple{C,\mathord{\leq_C},\vee,\wedge}\) be a complete lattice, where $\vee$ and $\wedge$ denote, resp., lub and glb. 
An \emph{upper closure operator}, or simply \emph{closure}, on \(\tuple{C,\mathord{\leq_C}}\) is a function \(\rho:C\to C\) which is:
\begin{inparaenum}[\upshape(i)]
\item monotonic,
\item idempotent: \(\rho(\rho(x)) = \rho(x)\) for all \(x \in C\), and
\item extensive: \(x \leq_C \rho(x)\) for all \(x \in C\).
\end{inparaenum}
The set of all upper closed operators on \(C\) is denoted by \(\uco(C)\).
We often write \(c \in \rho(C)\), or simply \(c \in \rho\), to denote that  
there exists \(c' \in C\) such that \(c = \rho(c')\), and 
recall that this happens if{}f $\rho(c) = c$. 
If $\rho\in \uco(C)$ then for all  \(c_1\in C\), \(c_2\in \rho\)  and \(X \subseteq C\), 
it turns out that:
\begin{align}
&c_1 \leq_C c_2 \Lra \rho(c_1)\leq_C \rho(c_2) \Lra \rho(c_1)\leq_C c_2 \label{equation:abstractcheck}\\
&\rho ({\textstyle\vee} X) = \rho({\textstyle\vee}\rho(X)) \quad \text{and}\quad {\textstyle\wedge}\rho (X) = \rho({\textstyle\wedge}\rho(X))\enspace. \label{equation:lubAndGlb}
\end{align}
In abstract interpretation, a closure operator \(\rho\in \uco(C)\) on a concrete domain $C$ plays
the role of abstraction function for objects of $C$. Given two closures \(\rho,\rho' \in \uco(C)\), \(\rho\) is a 
\emph{coarser} abstraction 
than \(\rho'\) (or, equivalently, 
$\rho'$ is a more precise abstraction than $\rho$) if{}f the image of 
\(\rho\) is a subset of the image of \(\rho'\), i.e. \(\rho(C) \subseteq \rho'(C)\), and this happens if{}f for any $x\in C$, 
$\rho'(x) \leq_C \rho(x)$.

Let us recall that a \emph{Galois Connection} (GC) or \emph{adjunction} between two posets \(\tuple{C,\leq_C}\), called concrete domain, and \(\tuple{A,\leq_A}\), called abstract domain, consists of two functions \(\alpha\colon C\ra A\) and \(\gamma \colon A\ra C\) such that \(\alpha(c)\leq_A a \:\Lra\: c\leq_C \gamma(a)\) always holds. 
A Galois Connection is denoted by \( \tuple{C,\leq_C} \galois{\alpha}{\gamma} \tuple{A,\leq_A}\).
The function $\alpha$ is called the left-adjoint of $\gamma$, and, dually, 
$\gamma$ is called the right-adjoint of $\alpha$. This terminology is justified by the fact that if
some function $\alpha:C\ra A$ 
admits a right-adjoint $\gamma:A\ra C$ then this is unique, and this dually holds for left-adjoints.
It turns out that in a GC between complete lattices, \(\gamma\) is always co-additive (i.e., it preserves arbitrary glb's)  
while \(\alpha\) is always additive (i.e., it preserves arbitrary lub's).  
Moreover, an additive function \(\alpha : C\ra A\) uniquely determines its right-adjoint by \(\gamma\ud \lambda a\ldotp \vee_C\{c\in C \mid \alpha(c)\leq_A a\}\) and, dually, a co-additive function \(\gamma: A\ra C\) uniquely determines its left-adjoint by \(\alpha \ud \lambda c\ldotp \wedge_A\{a\in A \mid 
c\leq_C \gamma(a)\}\).

The following remark is folklore in abstract interpretation and a proof is here provided for the sake of completeness. 
\begin{lemma}\label{lemma:alpharhoequality}
Let \( \tuple{C,\leq_C} \galois{\alpha}{\gamma} \tuple{A,\leq_A}\) be a GC between complete lattices and 
\(f\colon C\rightarrow C\) be a monotonic function. Then,
\(
	\gamma( \lfp (\alpha  f \gamma )) = \lfp (\gamma \alpha f)
\).	
\end{lemma}
\begin{proof}
Let us first show that \(\lfp (\gamma\alpha f) \leq_C \gamma( \lfp (\alpha f\gamma) )\):
\begin{align*}
\gamma(\lfp(\alpha f \gamma)) \leq_C \gamma(\lfp(\alpha f\gamma)) &\Lra \quad\text{[Since \(g(\lfp(g))=\lfp(g)\)]}\\
\gamma\alpha f(\gamma(\lfp(\alpha f \gamma))) \leq_C \gamma(\lfp(\alpha f \gamma))&\Rightarrow \quad\text{[Since $g(x)\leq x \Ra \lfp(g)\leq x$]}\\
\lfp(\gamma\alpha f)\leq_C \gamma(\lfp(\alpha f \gamma)) &
\end{align*}
  Then, let us prove that \(\gamma( \lfp (\alpha f\gamma) ) \leq_C \lfp (\gamma\alpha f) \):
\begin{align*}
\lfp(\gamma\alpha f)\leq_C \lfp(\gamma\alpha f) & \Lra \quad\text{[Since $g(\lfp(g))=\lfp(g)$]}\\
\gamma \alpha f(\lfp(\gamma\alpha f)) \leq_C \lfp(\gamma\alpha f) & \Ra \quad\text{[Since $\alpha$ is monotone]}\\
\alpha \gamma \alpha f(\lfp(\gamma\alpha f)) \leq_A \alpha (\lfp(\gamma\alpha f)) & \Lra \quad\text{[Since $\alpha\gamma\alpha=\alpha$ in GCs]}\\
 \alpha f(\lfp(\gamma\alpha f)) \leq_A \alpha (\lfp(\gamma\alpha f)) & \Lra \quad\text{[Since $\gamma\alpha(\lfp(\gamma\alpha f)) = \lfp(\gamma\alpha f)$]} \\
\alpha f \gamma(\alpha(\lfp(\gamma\alpha f)))\leq_A \alpha(\lfp(\gamma\alpha f))&\Ra\quad\text{[Since $g(x)\leq x \Ra \lfp(g)\leq x$]}\\
\lfp(\alpha f\gamma) \leq_A \alpha(\lfp (\gamma\alpha f)) &\Ra\quad\text{[Since $\gamma$ is monotone]}\\
\gamma(\lfp(\alpha f \gamma)) \leq_C \gamma\alpha(\lfp(\gamma\alpha f)) &\Lra\quad\text{[Since $\gamma\alpha(\lfp(\gamma\alpha f)) = \lfp(\gamma\alpha f)$]}\\
\gamma(\lfp (\alpha f\gamma)) \leq_C \lfp(\gamma\alpha f)\tag*{\qedhere}
\end{align*}
\end{proof}

\subsection{Languages}%
Let \(\Sigma\) be an alphabet, i.e., a finite nonempty set of symbols. 
A word (or string) on $\Sigma$ is a finite (possibly empty) sequence of symbols in $\Sigma$, where \(\epsilon\) denotes the empty sequence.
\(\Sigma^*\) denotes the set of finite words on $\Sigma$. A language on $\Sigma$ is a subset $L\subseteq \Sigma^*$.
Concatenation of words and languages is denoted by simple juxtaposition, that is, 
the concatenation of words 
$u,v\in \Sigma^*$ is denoted by \(uv\in \Sigma^*\), while the concatenation of 
languages $L,L'\subseteq \Sigma^*$ is denoted by
\(LL' \ud \{uv \mid u\in L,\,v\in L'\}$. By considering a word as a singleton language, 
we also concatenate words with languages, for example 
\(uL\) and $uLv$.

\begin{figure}[t]
		\centering
	\begin{tikzpicture}[->,>=stealth',shorten >=1pt,auto,node distance=5mm and 1cm,thick,initial text=]
	\tikzstyle{every state}=[scale=0.75,fill=blue!20,draw=blue!60,text=black]
	
	\node[initial,accepting,state] (1) {\(q_1\)};
	\node[state] (2) [right=of 1] {\(q_2\)};
	
	\path (1) edge[bend left] node {\(b\)} (2)
	      (2) edge[bend left] node {\(a\)} (1)
	      (2) edge[loop above] node {\(b\)} (2)
	      (1) edge[loop above] node {\(a\)} (1)
	          ;
	\end{tikzpicture}
	\caption{A finite automaton \(\cA\) with \(\lang{\cA}= (b^*a)^*\).}
	\label{fig:A}
	\end{figure}

A \emph{finite automaton} (FA) is a tuple \(\cA=\tuple{Q,\delta,I,F,\Sigma}\) where: \(\Sigma\) is an alphabet, \(Q\) is a finite set of states, \(I\subseteq Q\) is a subset of initial states, \(F\subseteq Q\) is a subset of final states, and \(\delta\colon  Q\times \Sigma \ra \wp(Q)\) is a transition relation. 
The notation \(q\ggoes{a} q'\) is also used to denote that \(q'\in \delta(q,a)\). 
If \(u\in \Sigma^*\) and \(q,q'\in Q\) then \(q \stackrel{u}{\leadsto} q'\) means that the state \(q'\) is reachable 
from \(q\) by following the string \(u\). More formally, by induction on the length of $u\in \Sigma^*$: 
(i)~if $u=\epsilon$ then \(q \goes{\epsilon} q'\) if{}f \(q=q'\); (ii)~if $u=av$ with $a\in \Sigma,v\in \Sigma^*$ then 
\(q \goes{av} q'\) if{}f $\exists q''\in \delta(q,a),\; q''\goes{v}q'$.
The \emph{language generated} by a FA \(\cA\) is \(\lang{\cA}\ud\{u \in \Sigma^* \mid \exists q_i\in I, \exists q_f \in F, \; q_i\goes{u}q_f\}\).
An example of FA is depicted in Fig.~\ref{fig:A}.

\section{Kleene Iterates with Abstract Inclusion Check}%
\label{sec:inclusion_checking_by_complete_abstractions}

Abstract interpretation can be applied to solve a generic inclusion checking problem
by leveraging backward complete abstractions~\cite{CC77,CC79,GiacobazziRS00,Ranzato13}. 
 A closure \(\rho\in \uco(C)\) is called \emph{backward complete}
for a concrete 
monotonic function ${f:C\ra C}$ when \( \rho f=\rho f \rho \) holds. Since $\rho f(c) \leq_C \rho f \rho(c)$ always holds for all $c\in C$ (because $\rho$ is extensive and monotonic and $f$ is monotonic), 
the intuition is that 
backward completeness models an ideal situation where no loss of precision
is accumulated in the computations of $\rho f$ when 
its concrete input objects $c$ are approximated by $\rho(c)$. 
It is well known~\cite{CC79} 
that backward completeness implies completeness of least fixpoints, namely for 
all $x\in C$ such that $x\leq_C f(x)$,
\begin{equation} \label{eqn:lfpcompleteness}
\rho f=\rho f \rho \;\Ra\;
\rho(\lfp_x(f))=\lfp_x(\rho f) = \lfp_x(\rho  f \rho)
\end{equation}
provided that these least fixpoints exist (this is the case, e.g., when $C$ is a CPO).  

Given an initial value $a\in C$, 
let us define the following iterative procedure:
\[
\Kleene(\Conv,f,a) \ud \left\{ \begin{array}{l}
x:=a; \\
\textbf{while~} \neg \Conv(f(x), x) \textbf{~do~} x:=f(x);\\
\textbf{return~} x;
\end{array}
\right. 
\] 
which computes the Kleene iterates of $f$ starting from $a$ and stops when 
a convergence relation $\Conv\subseteq C\times C$ for two consecutive Kleene iterates $f^{n+1}(a)$ and $f^{n}(a)$ holds. 
When  
$\Conv = \Incl \ud \{(x,y) \mid x\leq_C y\}$ is the convergence relation
and $a\leq_C f(a)$ holds, 
the procedure $\Kleene(\Incl,f,a)$ returns $\lfp_a(f)$ 
if the Kleene iterates finitely converge. Hence, termination of $\Kleene(\Eq,f,a)$ 
is guaranteed when $C$ is an ACC CPO.  

Given a closure $\rho \in \uco(C)$, 
let us consider 
the following abstract convergence relation induced by $\rho$:
\[
\Incl_\rho \ud \{(x,y) \in C\times C \mid \rho(x)\leq_C \rho(y)\}\enspace .
\]
Hence, $\Kleene(\Incl_\rho,f,a)$ terminates if eventually 
$\rho(f(x)) \leq_C \rho(x)$ holds. Notice that $\mathord{\Incl} \subseteq  \mathord{\Incl_\rho}$ always holds by monotonicity of $\rho$ and $\mathord{\Incl} = \mathord{\Incl_\rho}$ if{}f $\rho = \id$.

\begin{theorem}\label{new-lemma-kleene}
Let $\rho \in \uco(C)$ be such that 
 $\rho$ is backward complete for $f$ and
$\rho(C)$ does not contain infinite ascending chains. Let $a\in C$ 
such that $a\leq_C f(a)$ holds. Then, the procedure\/ $\Kleene(\Incl_\rho,f,a)$ terminates and $\rho (\Kleene(\Incl_\rho,f,a))=
\rho(\lfp_a(f)) = \lfp_a(\rho f)$. 
\end{theorem}
\begin{proof}
Let us first prove by induction the following property: 
\begin{equation}\label{prop-rhof}
\forall n\in \bN,\, 
\rho\comp f^n  = (\rho\comp f)^n \comp \rho
\enspace .
\end{equation}
For $n=0$, we have that $\rho \comp f^0 = \rho = 
(\rho\comp f)^0 \comp  \rho$. 
For $n+1$,
\begin{align*}
\rho \comp f^{n+1}  &= \qquad\text{[by definition of $f^{n+1}$]}\\
\rho \comp f^n \comp f &= \qquad\text{[by inductive hypothesis]}\\
(\rho\comp f)^n \comp \rho \comp f &= \qquad\text{[by backward completeness]}\\
(\rho\comp f)^n \comp \rho \comp f \comp \rho &= \qquad\text{[by definition of $(\rho \comp f)^{n+1}$]}\\
(\rho\comp f)^{n+1} \comp \rho &\enspace .
\end{align*}
Then, let us observe that $\lfp_a(\rho f) = \lfp_{\rho(a)}(\rho f)$: 
this is a consequence of the fact that 
$\rho(f(x)) = x \wedge a\leq_C x$ if{}f $\rho(f(x)) = x \wedge \rho(a) \leq_C x$, because 
$\rho(f(x)) = x \wedge a\leq_C x$ implies $\rho(f(x)) = x \wedge \rho(a)\leq_C \rho(x)=\rho(\rho(f(x)))=\rho(f(x))=x$. 

\noindent
Since $a\leq_C f(a)$, we have that $\tuple{f^n(a)}_{n\in \bN}$ is an ascending chain, 
so that, by monotonicity of $\rho$, 
$\tuple{\rho(f^n(a))}_{n\in \bN}$ is an ascending chain in $\rho(C)$.  
Since $\rho(C)$ does not contain infinite ascending chains, 
there exists $N= \min(\{ n\in \bN \mid \rho (f^{n+1}(a)) \leq_C 
\rho (f^{n}(a))\})$. This means that 
$\Kleene(\Incl_\rho,f,a)$ terminates after $N+1$ iterations and outputs 
$f^N(a)$. 
We prove by induction on $N\in \bN$ 
that $N= \min(\{ n\in \bN \mid (\rho \comp f)^{n+1}(\rho(a)) \leq_C (\rho \comp f)^n(\rho(a))\})$. 
\begin{itemize}
\item[$(N=0):$] We have that $\rho (f^{1}(a)) \leq_C 
\rho (f^{0}(a))$, namely, 
$\rho(f(a)) \leq_C \rho(a)$. Then, by backward completeness, 
 $\rho(f(\rho(a))) \leq_C \rho(a)$, namely, 
$(\rho \comp f)^{1}(\rho(a)) \leq_C (\rho \comp f)^0(\rho(a))$. 
\item[$(N+1):$] We have that $\rho (f^{N+2}(a)) \leq_C 
\rho (f^{N+1}(a))$, so that by \eqref{prop-rhof}, $(\rho\comp f)^{N+2}(\rho(a)) \leq_C 
(\rho \comp f)^{N+1}(\rho(a))$. Moreover, $N+1$ is the minimum natural number 
such that $(\rho \comp f)^{n+1}(\rho(a)) \leq_C (\rho \comp f)^n(\rho(a))$ holds, 
because if 
$(\rho \comp f)^{k+1}(\rho(a)) \leq_C (\rho \comp f)^k(\rho(a))$ for some $k\leq N$, 
then, by \eqref{prop-rhof}, we would have that 
$\rho (f^{k+1}(a)) \leq_C 
\rho (f^{k}(a))$, thus contradicting the minimality of $N+1$ for  
$\rho (f^{n+1}(a)) \leq_C 
\rho (f^{n}(a))$. 
\end{itemize}
Since $a\leq_C f(a)$ implies, by backward completeness, 
$\rho(a)\leq_C \rho(f(a)) = (\rho\comp f) (\rho(a)))$, and 
$N= \min(\{ n\in \bN \mid (\rho \comp f)^{n+1}(\rho(a)) \leq_C (\rho \comp f)^n(\rho(a))\})$, it turns out that $(\rho\comp f)^N (\rho(a))= {\lfp_{\rho(a)} (\rho f)} = \lfp_a(\rho f)$. Thus, by \eqref{prop-rhof}, we obtain $\lfp_a(\rho f) = 
(\rho\comp f)^N (\rho(a)) = \rho (f^N (a)) = \rho(\Kleene(\Incl_\rho,f,a))$. 
Finally, by \eqref{eqn:lfpcompleteness},  
$\rho(\Kleene(\Incl_\rho,f,a)) = \lfp_a(\rho f) =\rho(\lfp_a(f))$. 
\end{proof}

We will apply the order-theoretic algorithmic 
scheme provided by $\Kleene$ under the hypotheses of 
Theorem~\ref{new-lemma-kleene} to a number of
different language inclusion problems $L_1 \subseteq L_2$, where $L_1$ can be expressed 
as least fixpoint of a monotonic function on \(\wp(\Sigma^*)\). This will allow
us to systematically design several language inclusion algorithms which rely on 
different backward complete abstractions of the complete lattice 
\(\tuple{\wp(\Sigma^*),\subseteq}\).

\section{An Algorithmic Framework for Language Inclusion}%
\label{sec:an_algorithmic_framework_for_language_inclusion_based_on_complete_abstractions}

\subsection{Languages as Fixed Points}%
\label{sub:languages_as_fixpoints}

Let \(\cA=\tuple{Q,\delta,I,F,\Sigma}\) be a FA\@.
Given \(S,T \subseteq Q\), define the set of words leading from some state in $S$ to some
state in $T$ as follows:
\[W_{S,T}^\cA \ud \{u \in \Sigma^* \mid \exists q\in S,\,\exists q'\in T, q \goes{u} q'\}\enspace .\]
When \(S=\{q\}\) or \(T=\{q'\}\) we slightly abuse the notation and write \(W^{\cA}_{q,T}\), \(W^{\cA}_{S,q'}\), or \(W^{\cA}_{q,q'}\). 
Also, we omit the automaton \(\cA\) in superscripts when this is clear from the context.
The language accepted by \(\cA\) is  therefore \(\lang{\cA} \ud W^{\cA}_{I,F}\). 
Observe that
\begin{equation}%
\label{eq:unionofrightlg}
\lang{\cA}={\textstyle\bigcup_{q\in I}} W^\cA_{q,F}={\textstyle\bigcup_{q\in F}} W^\cA_{I,q}\enspace
\end{equation}
\noindent
where, as usual, \(\textstyle{\bigcup \varnothing} = \varnothing\).

Let us recall how to define the language accepted by an automaton as a solution of a set of equations~\cite{Schutzenberger63}.
Given a generic Boolean predicate \(p(x)\) for a variable $x$ ranging in some set (typically a membership predicate $x\inqm Z$)
and two generic sets $T$ and $F$, we define the following parametric 
choice function:
\[
\nullable{p(x)}{T}{F} \ud \begin{cases}
		T & \text{if \(p(x)\) holds} \\
		F & \text{otherwise}
\end{cases} \enspace .\]

\noindent
The FA \(\cA\) induces the following set of equations, where the $X_q$'s 
are variables of type $X_q\in \wp(\Sigma^*)$ and are indexed by states $q\in Q$ of $\cA$:
\begin{equation}\label{leftEqn}
	\Eqn(\cA) \ud \{ X_q = \nullable{q \inqm F}{\lbrace\epsilon\rbrace}{\varnothing} \cup {\textstyle \bigcup_{a\in \Sigma,\, q'\in\delta(q,a)}} a X_{q'} \mid  q\in Q\} \enspace.
\end{equation}
Thus, the functions $\lambda \tuple{X_{q'}}_{q'\in Q}.\: \nullable{q \inqm F}{\lbrace\epsilon\rbrace}{\varnothing} \cup {\textstyle \bigcup_{a\in \Sigma,\, q'\in\delta(q,a)}} a X_{q'}$
in the right-hand side of the equations in
\(\Eqn(\cA)\) have
type \(\wp(\Sigma^*)^{|Q|} \ra \wp(\Sigma^*)\).
Since \(\tuple{\wp(\Sigma^*)^{|Q|},\subseteq}\) is a (product) complete lattice (as \(\tuple{\wp(\Sigma^*),\subseteq}\) is a complete lattice) and all the right-hand side functions in \(\Eqn(\cA)\) are clearly monotonic, 
the least solution \(\tuple{Y_q}_{q\in Q}\in \wp(\Sigma^*)^{|Q|}\) of \(\Eqn(\cA)\) does exist and it is easy to check  
that for every \(q\in Q\), \(Y_q = W^{\cA}_{q,F}\) holds.

It is worth noticing that, by relying on right concatenations rather than left ones 
$aX_{q'}$ used 
in \(\Eqn(\cA)\), one could also define a
set of symmetric equations whose least solution coincides with \(\tuple{W_{I,q}^{\cA}}_{q\in Q}\) instead of \(\tuple{W_{q,F}^{\cA}}_{q\in Q}\).
\begin{example}\label{ex-first}
	Let us consider the automaton \(\cA\) in Figure~\ref{fig:A}. 	
The set of equations induced by \(\cA\) are as follows: 
\[
	\Eqn(\cA)=\begin{cases}
		X_1 = \{\epsilon\} \cup aX_1 \cup bX_2\\
		X_2 = \varnothing \cup aX_1 \cup b X_2 
	\end{cases} \enspace . \tag*{\raisebox{-0.8em}[0pt][0pt]{$\eox$}}
\]
\end{example}

It is notationally convenient 
to formulate  the equations in \(\Eqn(\cA)\) by exploiting
an ``initial'' vector \(\vectarg{\epsilon}{F} \in \wp(\Sigma^*)^{|Q|}\) and a predecessor
function \(\Pre_\cA \colon \wp(\Sigma^*)^{|Q|} {\ra} \wp(\Sigma^*)^{|Q|}\) defined as follows:
\begin{align*}
\vectarg{\epsilon}{F} &\ud \tuple{\nullable{q \inqm F}{\{\epsilon\}}{\varnothing}}_{q\in Q}\,, &\qquad
\Pre_\cA(\tuple{X_{q'}}_{q'\in Q}) &\ud \tuple{ {\textstyle \bigcup_{a\in \Sigma,\, q'\in\delta(q,a)}} aX_{q'}}_{q\in Q} \enspace .
\end{align*}
The intuition for the function \(\Pre_{\cA}\) is that given the language \(W_{q',F}^{\cA}\) and a transition \(q' \in \delta(q,a) \), we have that \(aW^{\cA}_{q',F} \subseteq W^{\cA}_{q,F}\) holds, i.e., given a subset $X_q'$ of the language generated by \(\cA\) from some state \(q'\), the function \(\Pre_{\cA}\) computes a subset $X_q$ of the language generated by \(\cA\) for its predecessor state \(q\). Notice that if all the components of 
$\vect{X}$ are finite sets of words then $\Pre_\cA(\vect{X})$ is still a vector of finite sets. 
Since \(\epsilon \in W_{q,F}^{\cA}\) for all \(q \in F\), the least fixpoint computation can start from the vector 
\(\vectarg{\epsilon}{F}\) and iteratively apply $\Pre_\cA$.
Therefore, it turns out that
\begin{equation}\label{eq:WqFAequalslfp}
\tuple{W^{\cA}_{q,F}}_{q\in Q} = \lfp(\lambda \vect{X}\ldotp \vectarg{\epsilon}{F} \cup \Pre_\cA(\vect{X})) \enspace .
\end{equation}
Together with Equation~\eqref{eq:unionofrightlg}, it follows that \(\lang{\cA}\) is given by the union of the component languages of the vector 
\(\lfp(\lambda \vect{X}\ldotp \vectarg{\epsilon}{F} \cup \Pre_\cA(\vect{X}))\) that are indexed by the initial states in $I$. 

\begin{example}[Continuation of Example~\ref{ex-first}]
The fixpoint characterization of \(\tuple{W_{q,F}^{\cA}}_{q\in Q}\) is:
\begin{equation*}
	\left( \begin{array}{c}
		 W^{\cA}_{q_1,q_1} \\ W^{\cA}_{q_2,q_1}
	\end{array} \right)=
	\lfp\biggl(\lambda\ \left(\begin{array}{c}
		X_1 \\ X_2
	\end{array}\right).
	\left(\begin{array}{c}
			\{\epsilon\} \cup a X_1 \cup b X_2 \\
			\varnothing \cup a X_1 \cup b X_2
		\end{array}\right)\biggr) = \left(\begin{array}{c}
		 (a+(b^+ a))^* \\ (a+b)^*a
	\end{array}\right)\enspace.\tag*{\eox}
\end{equation*}
\end{example}

\subsection{Language Inclusion using Fixed Points}
Consider a language inclusion problem \(L_1 \subseteq L_2\), where \(L_1=\lang{\cA}\) for some FA \(\cA=\tuple{Q,\delta,I,F,\Sigma}\).
The language \(L_2\) can be formalized as a vector in \(\wp(\Sigma^*)^{|Q|}\) as follows:
\begin{equation}\label{eq:elltwo}
\vectarg{L_2}{I} \ud \tuple{\nullable{q \inqm I}{L_2}{\Sigma^*}}_{q\in Q}
\end{equation} 
whose components indexed by initial states in $I$ are $L_2$ and those indexed by noninitial states are $\Sigma^*$. Then, as a consequence of~\eqref{eq:unionofrightlg},~\eqref{eq:WqFAequalslfp} and~\eqref{eq:elltwo},  we have that
\begin{equation}\label{eq:lfp}
\lang{\cA}\subseteq L_2 \Lra
\lfp(\lambda \vect{X}\ldotp\vectarg{\epsilon}{F} \cup \Pre_{\cA}(\vect{X})) \subseteq \vectarg{L_2}{I} \enspace .
\end{equation}

\begin{theorem}\label{theorem:backComplete}
If \(\rho \in \uco(\wp(\Sigma^*))\) is backward complete for \(\lambda X\in \wp(\Sigma^*)\ldotp aX\) for all \(a\in \Sigma\),
	then, for all FAs \(\cA=\tuple{Q,\delta,I,F,\Sigma}\) on the alphabet $\Sigma$, $\rho$ is backward complete for \(\Pre_{\cA}\) and \(\lambda \vect{X}\ldotp\vectarg{\epsilon}{F} \cup \Pre_{\cA}(\vect{X})\).
\end{theorem}
\begin{proof}\renewcommand{\qedsymbol}{}
\noindent 
First, it turns out that:
\begin{align*}
\rho( \Pre_\cA(\tuple{X_q}_{q\in Q})) &= \quad\text{~[by definition]}\\
\rho ({\textstyle \bigcup_{a\in \Sigma,\, q'\in \delta(q,a)}} aX_{q'}) &= 
\quad\text{~[by~\eqref{equation:lubAndGlb}]}\\
\rho ({\textstyle \bigcup_{a\in \Sigma,\, q'\in \delta(q,a)}} \rho(aX_{q'})) &= \quad\text{~[by backward completeness of $\rho$ for \(\lambda X\ldotp aX\)]}\\
\rho ({\textstyle \bigcup_{a\in \Sigma,\, q'\in \delta(q,a)}} \rho(a \rho(X_{q'}))) &=\quad\text{~[by~\eqref{equation:lubAndGlb}]}\\
\rho ({\textstyle \bigcup_{a\in \Sigma,\, q'\in \delta(q,a)}} a \rho(X_{q'})) &=\quad\text{~[by definition]}\\
\rho( \Pre_\cA(\rho(\tuple{X_q}_{q\in Q}))) & \enspace .
\end{align*}
As a consequence, \(\rho\) is backward complete for \(\lambda \vect{X}\ldotp \vectarg{\epsilon}{F} \cup  {\Pre}_\cA (\vect{X})\):
\begin{align*}
\rho (\vectarg{\epsilon}{F} \cup \Pre_\cA (\rho(\vect{X}))) & = 
\quad\text{~[by~\eqref{equation:lubAndGlb}]}\\
\rho (\rho (\vectarg{\epsilon}{F}) \cup  \rho (\Pre_\cA (\rho (\vect{X})))) & = 
\quad\text{~[by backward completeness of $\rho$ for \(\Pre_{\cA}\)]}\\
\rho (\rho (\vectarg{\epsilon}{F}) \cup  \rho (\Pre_\cA (\vect{X}))) & = 
\quad\text{~[by~\eqref{equation:lubAndGlb}]}\\
\rho (\vectarg{\epsilon}{F} \cup\: \Pre_\cA (\vect{X})) & \enspace . \tag*{$\Box$}
\end{align*}
\end{proof}
Then, by resorting to least fixpoint transfer of completeness~\eqref{eqn:lfpcompleteness}, we also obtain the following 
consequence.

\begin{corollary}\label{corol:rholfp}
If \(\rho \in \uco(\wp(\Sigma^*))\) 
is backward complete for \(\lambda X\in \wp(\Sigma^*)\ldotp aX\) for all \(a\in\Sigma\) then
\(\rho (\lfp(\lambda \vect{X}\ldotp\vectarg{\epsilon}{F} \cup \Pre_\cA(\vect{X}))) = 
\lfp(\lambda \vect{X}\ldotp \rho (\vectarg{\epsilon}{F} \cup \Pre_\cA(\vect{X})))\).
\end{corollary}

Note that if \(\rho\) is backward complete for \(\lambda X. aX\) for all \(a \in \Sigma\) 
and \(L_2\in \rho\) then, by Theorem~\ref{new-lemma-kleene} and Corollary~\ref{corol:rholfp}, the equivalence~\eqref{eq:lfp} becomes
\begin{equation}\label{equation:checklfpRLintoRL}
\lang{\cA}\subseteq L_2 \Lra \lfp(\lambda \vect{X}\ldotp\rho (\vectarg{\epsilon}{F} \!\cup \Pre_{\cA}(\vect{X}))) \subseteq \vectarg{L_2}{I} \enspace.
\end{equation}

\subsubsection{Right Concatenation}\label{rightwqo}
Let us consider the symmetric case of right concatenation \(\lambda X. Xa\).  
Recall that 
\(W_{I,q} = \{w \in \Sigma^* \mid \exists q_i \in I, q_i \goes{w}q\}\) and
that \(W_{I,q} = \nullable{q\inqm I}{\{\epsilon\}}{\varnothing} \cup {\textstyle\bigcup_{a\in\Sigma,a\in W_{q',q}}} W_{I,q'}a\) holds. Correspondingly, 
we define
a set of fixpoint equations on \(\wp(\Sigma^*)\) which is based on right concatenation 
and is symmetric to the equations defined in~\eqref{leftEqn}:
\[\Eqnr(\cA) \ud \{X_q = \nullable{q\inqm I}{\{\epsilon\}}{\varnothing} \cup {\textstyle\bigcup_{a\in\Sigma,\, q\in \delta(q',a)}} X_{q'}a \mid q \in Q\} \enspace .\]
In this case, if \(\vect{Y}=\tuple{Y_q}_{q\in Q}\) is the 
least fixpoint solution of \(\Eqnr(\cA)\) then 
\(Y_q = W^{\cA}_{I,q}\) for every \(q\in Q\).
Also, by defining \(\vectarg{\epsilon}{I} \in \wp(\Sigma^*)^{|Q|}\) and 
\(\Post_\cA \colon \wp(\Sigma^*)^{|Q|} {\ra} \wp(\Sigma^*)^{|Q|}\) as follows: 
\begin{align*}
\vectarg{\epsilon}{I} &\ud \tuple{\nullable{q \inqm I}{\{\epsilon\}}{\varnothing}}_{q\in Q} \quad & \Post_\cA(\tuple{X_q}_{q\in Q}) &\ud \tuple{ {\textstyle \bigcup_{a\in \Sigma,\, q\in \delta(q',a)}} X_{q'}a}_{q\in Q}
\end{align*}
we have that
\begin{equation}\label{eq:WIqAequalslfp}
\tuple{W_{I,q}}_{q\in Q} = \lfp(\lambda \vect{X}\ldotp \vectarg{\epsilon}{I} \cup \Post_\cA(\vect{X})) \enspace .
\end{equation}
Thus, by~\eqref{eq:unionofrightlg}, it turns out that 
\(\lang{\cA} = {\textstyle\bigcup_{q_f \in F}} W_{I,q_f}\) holds, that is, 
\(\lang{\cA}\) is the union of the component 
languages of the vector 
\(\lfp(\lambda \vect{X}\ldotp \vectarg{\epsilon}{I} \cup \Post_\cA(\vect{X}))\) 
indexed by the final states in  \(F\). 

\begin{example}\label{ex-firstPost}
Consider again the FA \(\cA\) in Figure~\ref{fig:A}.
The set of right equations 	for \(\cA\) is as follows:
\[
	\Eqnr(\cA)=\begin{cases}
		X_1 = \{\epsilon\} \cup X_1a \cup X_2 a\\
		X_2 = \varnothing \cup X_1b \cup X_2b
	\end{cases}
\]
so that 
\begin{equation*}
	\left( \begin{array}{c}
		 W_{q_1,q_1} \\ W_{q_1,q_2}
	\end{array} \right)=
	\lfp\biggl(\lambda\ \left(\begin{array}{c}
		X_1 \\ X_2
	\end{array}\right).
	\left(\begin{array}{c}
			\{\epsilon\} \cup X_1 a \cup X_2 a \\
			\varnothing \cup X_1 b\cup X_2 b
		\end{array}\right)\biggr) = \left(\begin{array}{c}
		 (a+(b^+ a))^* \\ a^*b(b+a^+b)^*
	\end{array}\right)\enspace.\tag*{\eox}
\end{equation*}
\end{example}

In a language inclusion problem \(\lang{\cA} \subseteq L_2\), 
we consider the  vector 
\(\vectarg{L_2}{F} \ud \tuple{\nullable{q \inqm F}{L_2}{\Sigma^*}}_{q\in Q} \in \wp(\Sigma^*)^{|Q|}\), so that, by~\eqref{eq:WIqAequalslfp}, it turns out that: 
\begin{align*}
\lang{\cA}\subseteq L_2 \:\Lra\: 
\lfp(\lambda \vect{X}\ldotp\vectarg{\epsilon}{I} \cup \Post_{\cA}(\vect{X})) \subseteq \vectarg{L_2}{F}\enspace.
\end{align*}
We therefore have the following symmetric version 
of Theorem~\ref{theorem:backComplete} for right concatenation.

\begin{theorem}\label{theorem:backCompleteRight}
If \(\rho \in \uco(\wp(\Sigma^*))\) 
is backward complete for \(\lambda X\ldotp Xa\) for all \(a\in \Sigma\) then, for all FAs $\cA$ on the alphabet $\Sigma$, \(\rho\) is backward complete for \(\lambda \vect{X}\ldotp\vectarg{\epsilon}{I} \cup \Post_{\cA}(\vect{X})\).
\end{theorem}

\subsection{A Language Inclusion Algorithm with Abstract Inclusion Check}
Let us now apply the general Theorem~\ref{new-lemma-kleene} to design 
an algorithm that solves a language inclusion problem \(\lang{\cA} \subseteq L_2\)
by exploiting a language abstraction \(\rho\) that satisfies 
a list of requirements of backward completeness and computability.

\begin{theorem}\label{theorem:FiniteWordsAlgorithmGeneral}
Let \(\cA=\tuple{Q,\delta,I,F,\Sigma}\) be a FA, \(L_2\in \wp(\Sigma^*)\) and 
\(\rho \in \uco(\wp(\Sigma^*))\). 
Assume that the following properties hold:
\begin{compactenum}[\upshape(\rm i\upshape)]
\item The closure \(\rho\) is backward complete for \(\lambda X\in \wp(\Sigma^*)\ldotp aX\), for all \(a\in \Sigma\), and satisfies \(\rho(L_2) = L_2\).\label{theorem:FiniteWordsAlgorithmGeneral:rho}
\item \(\rho(\wp(\Sigma^*))\) does not contain
infinite ascending chains.  \label{theorem:FiniteWordsAlgorithmGeneral:ACC}
\item If $X,Y\in \wp(\Sigma^*)$ are finite sets of words then 
the inclusion  $\rho(X) \subseteqm \rho(Y)$ is decidable. 
\label{theorem:FiniteWordsAlgorithmGeneral:EQ}
\item If $Y\in \wp(\Sigma^*)$ is a finite set of words then 
the inclusion $\rho(Y) \subseteqm L_2$ is decidable. 
\label{theorem:FiniteWordsAlgorithmGeneral:inclrho}
\end{compactenum}
\medskip
Then,

\medskip
$
\begin{array}{l}
\tuple{Y_q}_{q\in Q} := \Kleene (\Incl_\rho,\lambda \vect{X}.
\vectarg{\epsilon}{F} \!\cup \Pre_{\cA}(\vect{X}), \vect{\varnothing}) \emph{;}\\
\emph{\textbf{return}}\: \Incl_\rho(\tuple{Y_q}_{q\in Q},\vectarg{L_2}{I})\emph{;}
\end{array}
$

\medskip
\noindent
is a decision algorithm for \(\lang{\cA} \subseteq L_2\).
\end{theorem}

\begin{proof}
Conditions~(\ref{theorem:FiniteWordsAlgorithmGeneral:rho}),~(\ref{theorem:FiniteWordsAlgorithmGeneral:ACC}) and~(\ref{theorem:FiniteWordsAlgorithmGeneral:EQ}) guarantee that 
the hypotheses of
Theorem~\ref{new-lemma-kleene} are satisfied. Thus, 
$\Kleene(\Incl_\rho, \lambda \vect{X}\ldotp \vectarg{\epsilon}{F} \!\cup \Pre_{\cA}(\vect{X}),\vect{\varnothing})$ is an algorithm that terminates  
with output $\tuple{Y_q}_{q\in Q}$ and 
\begin{equation*}\label{eq:lfpKleeneQO}
\rho(\lfp(\lambda \vect{X}\ldotp \vectarg{\epsilon}{F} \!\cup \Pre_{\cA}(\vect{X}))) = 
\rho(\tuple{Y_q}_{q\in Q}).
\end{equation*}
Moreover, by \eqref{eq:lfp},
$\lang{\cA} \subseteq L_2 \Lra \rho(\lang{\cA}) \subseteq \rho(L_2)=L_2 \Lra
\rho(\lfp(\lambda \vect{X}\ldotp\vectarg{\epsilon}{F} \cup \Pre_{\cA}(\vect{X}))) \subseteq \vectarg{L_2}{I} \Lra \rho(\tuple{Y_q}_{q\in Q}) \subseteq \rho(\vectarg{L_2}{I}) 
\Lra  \Incl_\rho(\tuple{Y_q}_{q\in Q},\vectarg{L_2}{I})$. Finally,  by condition~\eqref{theorem:FiniteWordsAlgorithmGeneral:inclrho}, $\Incl_\rho(\tuple{Y_q}_{q\in Q},\vectarg{L_2}{I})$ is decidable. 
\end{proof}

It is worth noticing that Theorem~\ref{theorem:FiniteWordsAlgorithmGeneral} can also be stated in a symmetric version for \(\lambda \vect{X}\ldotp\vectarg{\epsilon}{I} \cup \Post_{\cA}(\vect{X})\) similarly to Theorem~\ref{theorem:backCompleteRight}.

\section{Instantiating the Framework with Quasiorders}%
\label{sec:instantiating_the_framework_language_based_well_quasiorders}

We instantiate the general algorithmic framework of 
Section~\ref{sec:an_algorithmic_framework_for_language_inclusion_based_on_complete_abstractions} to the class of closure operators induced by quasiorder relations on words. 

\subsection{Word-based Abstractions}\label{sec:word-based}
Let \(\mathord{\leqslant} \subseteq \Sigma^* \times \Sigma^*\) be a quasiorder relation on words in $\Sigma^*$.  
The corresponding closure operator \(\rho_\leqslant \in \uco(\wp(\Sigma^*))\) is defined as follows: 
\begin{equation}\label{eq:qo-up-closure}
\rho_\leqslant(X) \ud  %
\{v\in \Sigma^* \mid \exists u\in X, \;u \leqslant v \} \enspace .
\end{equation}
Thus, $\rho_\leqslant(X)$ is the $\leqslant$-upward closure
of $X$ and it is easy to check that $\rho_\leqslant$ is indeed a closure
on the complete lattice $\tuple{\wp(\Sigma^*),\subseteq}$.
\\
\indent
Following \cite{deLuca1994}, a quasiorder \(\leqslant\) on \(\Sigma^*\) is \emph{left-monotonic} (resp.\ \emph{right-monotonic})
if 
\begin{equation*}
\forall y,x_1,x_2 \in \Sigma^*,\: x_1\leqslant x_2 \,\Ra\, y x_1 \leqslant y x_2 
\quad \text{(resp.\ $x_1 y \leqslant x_2 y$) \enspace.}
\end{equation*}
Also, \(\leqslant\) is called monotonic if it is both left- and right-monotonic. 
It turns out that \(\leqslant\) is left-monotonic (resp.\ right-monotonic) if{}f 
\begin{equation}\label{def-leftmon}
\forall x_1,x_2 \in \Sigma^*,\forall a\in \Sigma,\: x_1\leqslant x_2 \,\Ra\, a x_1 \leqslant a x_2 \quad \text{(resp.\ $x_1 a \leqslant x_2 a$)} \enspace .
\end{equation}
In fact, if $x_1\leqslant x_2$ then \eqref{def-leftmon} implies 
that for all $y\in \Sigma^*$, $y x_1 \leqslant y x_2$: by induction 
on the length $|y|\in\bN$, we have that:
(i)~if $y=\epsilon$ then $y x_1 \leqslant y x_2$; (ii)~if $y=av$ with $a\in \Sigma,v\in \Sigma^*$ then, by inductive
hypothesis, 
$v x_1 \leqslant v x_2$, so that by \eqref{def-leftmon}, $yx_1=av x_1 \leqslant av x_2=yx_2$.

\begin{definition}[$L$-Consistent Quasiorder]\label{def:LConsistent}\rm
	Let $L\in$ $\wp(\Sigma^*)$. A quasiorder \(\mathord{\leqslant_L} \subseteq \Sigma^* \times \Sigma^*\) is called \emph{left} (resp.\ \emph{right}) \(L\)\emph{-consistent} when:   
	\begin{compactenum}[\upshape(a)]
	\item \(\mathord{\leqslant}_L \cap (L\times \neg L) = \varnothing \);\label{eq:LConsistentPrecise}
	\item \(\mathord{\leqslant}_L\) is left-monotonic (resp.\ right-monotonic). \label{eq:LConsistentmonotonic}
	\end{compactenum} 
	Also, \(\mathord{\leqslant}_L\) is called \emph{\(L\)-consistent} when it is both left and right \(L\)-consistent.\hfill{\rule{0.5em}{0.5em}}
\end{definition}

It turns out that a quasiorder is $L$-consistent if{}f it induces a closure which includes $L$ in its image and it is 
backward complete for concatenation. 
\begin{lemma}\label{lemma:properties}
Let \(L\in \wp(\Sigma^*)\) and \(\mathord{\leqslant_L}\) be a quasiorder on \(\Sigma^*\).
Then, \(\mathord{\leqslant_L}\) is a 
left (resp.\ right) \(L\)-consistent quasiorder on \(\Sigma^*\) if and only if
\begin{compactenum}[\upshape(\rm a\upshape)]
\item \(\rho_{\leqslant_L}(L) = L\), and \label{lemma:properties:L}
\item \(\rho_{\leqslant_L}\) is backward complete for \(\lambda X\ldotp a X\) (resp.\ \(\lambda X\ldotp Xa\)) for all \(a\in \Sigma\).\label{lemma:properties:bw}
\end{compactenum}
\end{lemma}
\begin{proof}%
We consider the left case, the right case is symmetric. 

\begin{enumerate}[\upshape(\rm a\upshape)]
\item The inclusion 
\(L\subseteq \rho_{\leqslant_L}(L)\) always 
holds because \(\rho_{\leqslant_L}\) is an upper closure. Then, it turns out that $\rho_{\leqslant_L}(L)= L$ if{}f 
$\rho_{\leqslant_L}(L)\subseteq L$ if{}f $\forall v\in \Sigma^*$, 
$(\exists u\in L,\, u \leqslant_L v) \:\Ra\: v\in L$ if{}f \(\mathord{\leqslant}_L \cap (L\times \neg L) = \varnothing\). 
Thus, $\rho_{\leqslant_L}(L)= L$ if{}f
condition~(\ref{eq:LConsistentPrecise}) of Definition~\ref{def:LConsistent} holds. 

\item We first prove that if \(\mathord{\leqslant}_L\) is left-monotonic then for all $X\in \wp(\Sigma^*)$, 
\(\rho_{\leqslant_L}(a X) = \rho_{\leqslant_L}(a \rho_{\leqslant_L}(X))\) for all \(a\in\Sigma\). 
Monotonicity of concatenation together with monotonicity and extensivity of 
$\rho_{\leqslant_L}$ imply that \(\rho_{\leqslant_L}(a X) \subseteq \rho_{\leqslant_L}(a \rho_{\leqslant_L}(X))\) holds.
For the reverse inclusion, we have that:
\begin{align*}
	\rho_{\leqslant_L}(a \rho_{\leqslant_L}(X)) %
	=& \quad \text{~[by def.\ of \(\rho_{\leqslant_L}\)]}\\
	\rho_{\leqslant_L}\left( \{ a y \mid \exists x\in X, x \leqslant_L y \} \right)
	=& \quad \text{~[by def.\ of \(\rho_{\leqslant_L}\)]}\\
	\{ z \mid \exists x\in X, y\in \Sigma^*,\, x\leqslant_L y \land a y \leqslant_L z \}
	\subseteq& \quad \text{~[by left monotonicity of \(\leqslant_L\)]}\\
	\{ z \mid \exists x\in X, y\in \Sigma^*,\, ax\leqslant_L ay \land a y \leqslant_L z \}
	=& \quad \text{~[by transitivity of \(\leqslant_L\)]}\\
	\{ z \mid \exists x\in X , a x\leqslant_L z\}
	=& \quad \text{~[by def.\ of \(\rho_{\leqslant_L}\)]}\\
	\rho_{\leqslant_L}(a X) \;\phantom{=} &.   %
\end{align*}
Conversely, assume that for all $X\in \wp(\Sigma^*)$ and \(a\in\Sigma\),
\(\rho_{\leqslant_L}(a X) = \rho_{\leqslant_L}(a \rho_{\leqslant_L}(X))\). 
Consider $x_1,x_2\in \Sigma^*$ and $a\in \Sigma$. 
If $x_1 \leqslant_L x_2$ then
$\{x_2\} \subseteq \rho_{\leqslant_L}(\{x_1 \})$, and in turn 
$a\{x_2\} \subseteq  a\rho_{\leqslant_L}(\{x_1 \})$.
Then, by applying the monotonic function 
$\rho_{\leqslant_L}$, 
$\rho_{\leqslant_L}(a\{x_2\}) \subseteq  \rho_{\leqslant_L}(a\rho_{\leqslant_L}(\{x_1 \}))$, so that, by backward completeness, 
$\rho_{\leqslant_L}(a\{x_2\}) \subseteq  \rho_{\leqslant_L}(a\{x_1 \})$.
Hence, $a\{x_2\} \subseteq \rho_{\leqslant_L}(a\{x_1 \})$, namely, 
$ax_1 \leqslant_L ax_2$. By \eqref{def-leftmon}, this shows that $\leqslant_L$ is left-monotonic. \qedhere
\end{enumerate}
\end{proof}

We can apply Theorem~\ref{theorem:FiniteWordsAlgorithmGeneral} 
to the closure $\rho_{\leqslant^l_{L_2}}$ induced by a left $L_2$-consistent 
well-quasiorder, since it satisfies all the required hypotheses, 
thus obtaining the following Algorithm~\AlgRegularW which solves the language inclusion problem \(\lang{\cA} \subseteq L_2\) for any automaton \(\cA\).
This algorithm is called ``word-based'' because the output 
vector \(\tuple{Y_q}_{q \in Q}\) computed by $\Kleene$ 
consists of finite sets of words. Here, the convergence relation 
$\Incl_{\rho_{\leqslant^l_{L_2}}}$ of $\Kleene$ coincides with the relation  
\(\mathord{\sqsubseteq_{\leqslant^{l}_{L_2}}}\) 
because 
$\Incl_{\rho_{\leqslant^l_{L_2}}}(X,Y)$ if{}f $\rho_{\leqslant^l_{L_2}}(X) \subseteq 
\rho_{\leqslant^l_{L_2}}(Y)$ if{}f 
 \(X \sqsubseteq_{\leqslant^{l}_{L_2}} Y\).

\begin{figure}[!ht]
\RemoveAlgoNumber
\begin{algorithm}[H]
\SetAlgorithmName{\AlgRegularW}{}
\SetSideCommentRight
\caption{Word-based algorithm for \(\lang{\cA} \subseteq L_2\)\label{alg:RegIncW}}

\KwData{FA \(\cA=\tuple{Q,\delta,I,F,\Sigma}\); decision procedure for \(u\mathrel{\inqm} L_2\); decidable left \(L_2\)-consistent wqo \(\mathord{\leqslant^l_{L_2}}\).}

\(\tuple{Y_q}_{q\in Q} := \Kleene (\sqsubseteq_{\leqslant^l_{L_2}}, \lambda \vect{X}\ldotp\vectarg{\epsilon}{F} \cup \Pre_{\cA}(\vect{X}), \vect{\varnothing})\)\;

\ForAll{\(q \in I\)}{
  \ForAll{\(u \in Y_q\)}{
    \lIf{\(u \notin L_2\)}{\Return \textit{false}}
  }
}
\Return \textit{true}\;
\end{algorithm}
\end{figure}

\begin{theorem}\label{theorem:quasiorderAlgorithm}
Let \(\cA=\tuple{Q,\delta,I,F,\Sigma}\) be a FA and \(L_2\in \wp(\Sigma^*)\) be a language such that: 
\begin{inparaenum}[\upshape(i\upshape)]
\item membership in $L_2$ is decidable; \label{theorem:quasiorderAlgorithm:membership}
\item there exists a decidable left \(L_2\)-consistent wqo on $\Sigma^*$.\label{theorem:quasiorderAlgorithm:decidableL}%
\end{inparaenum}%
Then, Algorithm \AlgRegularW decides the inclusion problem \(\lang{\cA} \subseteq L_2\).
\end{theorem}
\begin{proof}
Let  \(\leqslant_{L_2}^l\) be the 
decidable left \(L_2\)-consistent wqo on $\Sigma^*$.
Let us check that the hypotheses~(i)-(ii)-(iii) 
of Theorem~\ref{theorem:FiniteWordsAlgorithmGeneral} are satisfied.

\begin{enumerate}[\upshape(\rm i\upshape)]
\item It follows from hypothesis~(\ref{theorem:quasiorderAlgorithm:decidableL}) and Lemma~\ref{lemma:properties} that \(\leqslant_{L_2}^l\) is backward complete for left concatenation and satisfies \(\rho_{\leqslant_{L_2}^l}(L_2) = L_2\).

\item Since \(\leqslant_{L_2}^l\) is a well-quasiorder, it follows that \(\{\rho_{\leqslant_{L_2}^l}(S) \mid S \in \wp(\Sigma^*)\}\) does not contain infinite ascending chains.  

\item For finite sets for finite sets $X$ and $Y$, the abstract inclusion 
$\Incl_{\rho_{\leqslant^l_{L_2}}}(X,Y)$ $\Lra$ 
 \(X \sqsubseteq_{\leqslant^{l}_{L_2}} Y\) is decidable 
since \(\leqslant_{L_2}^l\) is a decidable wqo.
\end{enumerate}
Moreover, it turns out that 
the check $\Incl_{\rho_{\leqslant^l_{L_2}}}(\tuple{Y_q}_{q\in Q},\vectarg{L_2}{I})$ of Theorem~\ref{theorem:FiniteWordsAlgorithmGeneral} is decidable and 
is performed
by lines 2-5 of Algorithm~\AlgRegularW.
In fact, since, by Theorem~\ref{theorem:FiniteWordsAlgorithmGeneral}, 
$\Kleene (\sqsubseteq_{\leqslant^l_{L_2}}, \lambda \vect{X}\ldotp\vectarg{\epsilon}{F} \cup \Pre_{\cA}(\vect{X}), \vect{\varnothing})$ terminates after a finite number of steps 
with output $\tuple{Y_q}_{q\in Q}$, each set of words 
$Y_q$ of the output turns out to be finite. 
Also, since \(\vectarg{L_2}{I} = \tuple{\nullable{q \inqm I}{L_2}{\Sigma^*}}_{q \in Q})\), the abstract inclusion trivially holds for all components \(Y_q\) with \(q \notin I\).
Therefore, it suffices to check whether \(Y_q \sqsubseteq_{\leqslant^l_{L_2}} L_2\) holds for all $q\in I$. Since \(Y_q \sqsubseteq_{\leqslant^l_{L_2}} L_2\) if{}f $\rho_{\leqslant^l_{L_2}}(Y_q) \subseteq \rho_{\leqslant^l_{L_2}}(L_2)=L_2$ if{}f 
$Y_q \subseteq L_2$ and since $Y_q$ is a finite set,  
\(Y_q \sqsubseteq_{\leqslant^l_{L_2}} L_2\) 
can be decided by performing the finitely many membership check $u\inqm L_2$
at lines 2-5, where 
by hypothesis~(\ref{theorem:quasiorderAlgorithm:decidableL}), any membership 
check is decidable. Thus, hypothesis~(iv) 
of Theorem~\ref{theorem:FiniteWordsAlgorithmGeneral} is satisfied.

\noindent 
Summing up, we have shown that 
Algorithm~\AlgRegularW decides the inclusion \(\lang{\cA} \subseteq L_2\).
\end{proof}

\begin{remark}\rm 
It is worth noticing that in each   iteration of $\Kleene (\sqsubseteq_{\leqslant^l_{L_2}}, \lambda \vect{X}\ldotp\vectarg{\epsilon}{F} \cup \Pre_{\cA}(\vect{X}), \vect{\varnothing})$
in Algorithm~\AlgRegularW, in the current
vector $\tuple{Y_q}_{q\in Q}$ one
could safely remove from a component $Y_q$ any word $w\in Y_q$ such that there
exists a word $u\in Y_q$ such that $u \leqslant^l_{L_2} w$ and $u\neq w$. 
This enables replacing each finite set $Y_q$ occurring in Kleene iterates 
with a minor subset $\minor{Y_q}$ w.r.t.\ $\leqslant^l_{L_2}$. 
This replacement is correct, i.e.\ Theorem~\ref{theorem:quasiorderAlgorithm} still holds
for the corresponding modified language inclusion algorithm, because an
inclusion check $X\sqsubseteq_{\leqslant^l_{L_2}}Y$ holds if{}f the check
$\minor{X}\sqsubseteq_{\leqslant^l_{L_2}}\minor{Y}$ for the corresponding minor subsets 
holds. 
\eox
\end{remark}

\subsubsection{Right Concatenation}
Following Section~\ref{rightwqo},
a symmetric version, called \AlgRegularWr, of the algorithm \AlgRegularW (and of Theorem~\ref{theorem:quasiorderAlgorithm}) for \emph{right} \(L_2\)-consistent wqos can be given  as follows. 
\begin{figure}[H]
\RemoveAlgoNumber
\begin{algorithm}[H]
\SetAlgorithmName{\AlgRegularWr}{}

\caption{Word-based algorithm for \(\lang{\cA} \subseteq L_2\)} \label{alg:RegIncWr}

\KwData{FA \(\cA=\tuple{Q,\delta,I,F,\Sigma}\); decision procedure for \(u\inqm L_2\); decidable right \(L_2\)-consistent wqo \(\mathord{\leqslant^r_{L_2}}\).}

\(\tuple{Y_q}_{q\in Q} := \Kleene (\sqsubseteq_{\leqslant^r_{L_2}},\lambda \vect{X}\ldotp\vectarg{\epsilon}{I}\cup \Post_{\cA}(\vect{X}), \vect{\varnothing})\)\;

\ForAll{\(q \in F\)}{
	\ForAll{\(u \in Y_q\)} {
		\lIf{\(u \notin L_2\)}{\Return \textit{false}}
	}
}
\Return \textit{true}\;
\end{algorithm}
\end{figure}

\begin{theorem}\label{theorem:quasiorderAlgorithmR}
Let \(\cA=\tuple{Q,\delta,I,F,\Sigma}\) be a FA and \(L_2\in \wp(\Sigma^*)\) be a language such that: 
\begin{inparaenum}[\upshape(i\upshape)]
\item membership in $L_2$ is decidable; \label{theorem:quasiorderAlgorithmR:membership}
\item there exists a decidable right \(L_2\)-consistent wqo on $\Sigma^*$.\label{theorem:quasiorderAlgorithmR:decidableL}%
\end{inparaenum}%
Then, Algorithm \AlgRegularWr decides the inclusion problem \(\lang{\cA} \subseteq L_2\).
\end{theorem}

In the following, we will consider different quasiorders on $\Sigma^*$ and we will show that they fulfill the requirements of Theorem~\ref{theorem:quasiorderAlgorithm}, therefore yielding algorithms for solving language inclusion problems.

\subsection{Nerode Quasiorders}\label{sec:nerode}
\label{sub:the_left_nerode_quasi_order_relative_to_a_language}
The notions of \emph{left}  and \emph{right quotient} of a language \(L \in \wp(\Sigma^*)\) w.r.t.\ a word \(w\in \Sigma^*\) are standard:  
\begin{align*}
w^{-1}L &\ud \{u\in \Sigma^* \mid  wu\in L\} \,, & 
Lw^{-1} &\ud \{u\in \Sigma^* \mid  uw\in L\} \enspace .
\end{align*}
Correspondingly, let us define the following quasiorder relations on \(\Sigma^*\):
\begin{align}\label{def-Nerodeqo}
	u\leqq_L^l v &\udr\; L u^{-1} \subseteq L v^{-1} \,,&
	u\leqq_L^r v &\udr\; u^{-1} L \subseteq v^{-1} L \enspace . 
\end{align}
De Luca and Varricchio \citeyearpar[Section~2]{deLuca1994} call them, resp., the \emph{left} 
($\leqq_L^l$)
and \emph{right} ($\leqq_L^r$) 
\emph{Nerode quasiorders relative to \(L\)}. 
The following result shows that Nerode quasiorders are the weakest (i.e., greatest w.r.t.\ set inclusion of binary relations) \(L_2\)-consistent quasiorders for which the algorithm \AlgRegularW can be instantiated to decide a language inclusion \(\lang{\cA}\subseteq L_2\).

\begin{lemma}\label{lemma:leftrightnerodegoodqo}
Let $L\in \wp(\Sigma^*)$.  
\begin{compactenum}[\upshape(\rm a\upshape)]
\item \(\mathord{\leqq_L^l}\) and {}
\(\mathord{\leqq_L^r}\) are, resp., left and right 
\(L\)-consistent quasiorders.
	If $L$ is regular then, additionally, \(\mathord{\leqq_L^l}\) and \(\mathord{\leqq_L^r}\)
are decidable wqos. 	\label{lemma:leftrightnerodegoodqo:Consistent}
\item If \(\mathord{\leqslant}\) is a left (resp.\ right) \(L\)-consistent quasiorder on $\Sigma^*$ then \( \rho_{\leqq_L^l}(\wp(\Sigma^*)) \subseteq \rho_{\leqslant}(\wp(\Sigma^*)) \) (resp.\ \( \rho_{\leqq_L^r}(\wp(\Sigma^*)) \subseteq \rho_{\leqslant}(\wp(\Sigma^*)) \)).\label{lemma:leftrightnerodegoodqo:Incl}
\end{compactenum}
\end{lemma}
\begin{proof}
Let us consider point~(\ref{lemma:leftrightnerodegoodqo:Consistent}).
De Luca and Varricchio~\citeyearpar[Section~2]{deLuca1994} observe that \(\mathord{\leqq_L^l}\) and \(\mathord{\leqq_L^r}\) are,
resp., left and right monotonic. 
Moreover, De Luca and Varricchio~\citeyearpar[Theorem~2.4]{deLuca1994} show that  if 
$L$ is regular then both \(\mathord{\leqq_L^l}\) and \(\mathord{\leqq_L^r}\) are wqos. 
Let us also observe that given \(u \in L\) and \(v \notin L\) we have that \(\epsilon \in Lu^{-1}\) and \(\epsilon \in u^{-1}L\) while \(\epsilon \notin Lv^{-1}\) and \(\epsilon \notin v^{-1}L\). Hence, \(\mathord{\leqq_L^l}\) (\(\mathord{\leqq_L^r}\)) is a left (right) \(L\)-consistent quasiorder.
Finally, if $L$ is regular then both relations are 
clearly decidable.

\noindent
Let us now consider point (\ref{lemma:leftrightnerodegoodqo:Incl}) for the left case (the right case is symmetric). 
By the characterization of left consistent quasiorders given by Lemma~\ref{lemma:properties}, 
De Luca and Varricchio~\citeyearpar[Section~2, point~4]{deLuca1994} observe that \(\mathord{\leqq_L^l}\) is maximum in the set of all left \(L\)-consistent quasiorders, i.e.\ every left \(L\)-consistent quasiorder \(\leqslant\) is such that 
	\(x \leqslant y \Ra x \leqq_L^l y \).
As a consequence, \(\rho_{\leqslant}(X) \subseteq \rho_{\leqq_L^l}(X)\) holds for all \(X\in \wp(\Sigma^*)\), namely, \( \rho_{\leqq_L^l}(\wp(\Sigma^*)) \subseteq \rho_{\leqslant}(\wp(\Sigma^*)) \).
\end{proof}

This allows us to derive a first instantiation of Theorem~\ref{theorem:quasiorderAlgorithm}. %
Because membership is decidable for regular languages $L_2$, Lemma~\ref{lemma:leftrightnerodegoodqo}~(\ref{lemma:leftrightnerodegoodqo:Consistent}) for \(\leqq^l_{L_2}\) implies that the hypotheses (\ref{theorem:quasiorderAlgorithm:membership}) and (\ref{theorem:quasiorderAlgorithm:decidableL}) of Theorem~\ref{theorem:quasiorderAlgorithm} are satisfied, so that the algorithm \AlgRegularW instantiated to \(\leqq^l_{L_2}\)
decides the inclusion \(\lang{\cA} \subseteq L_2\) when $L_2$ is regular. 
Furthermore, under these hypotheses,
Lemma~\ref{lemma:leftrightnerodegoodqo}~(\ref{lemma:leftrightnerodegoodqo:Incl}) shows that \(\leqq_{L_2}^l\) is the weakest
left \(L_2\)-consistent quasiorder relation on $\Sigma^*$ for which the algorithm \AlgRegularW can be instantiated 
for deciding an inclusion $\lang{\cA}\subseteq L_2$.

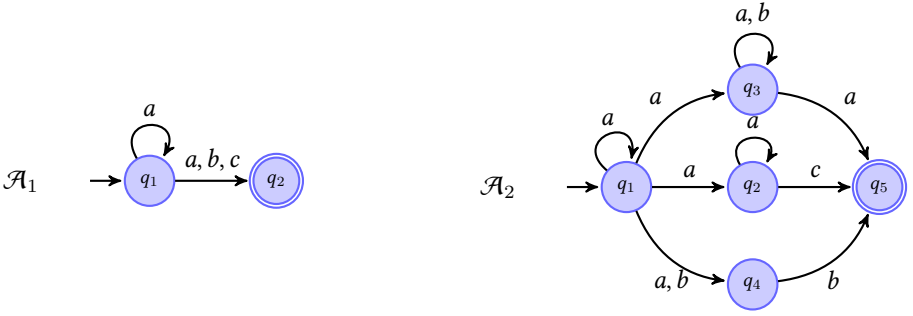
\begin{figure}[t]
    \centering
  \hfill\begin{minipage}{0.45\textwidth}
  \begin{tikzpicture}[->,>=stealth',shorten >=1pt,auto,node distance=5mm and 1cm,thick,initial text=]
  \tikzstyle{every state}=[scale=0.75,fill=blue!20,draw=blue!60,text=black]

  \node[initial, state] (1) {\(q_1\)};
  \node[accepting, state] (2) [right=of 1] {\(q_2\)};
  \node (A1) [left=of 1] {$\cA_1$};
  
  \path (1) edge node {\(a, b, c\)} (2)
        (1) edge[loop, in=60, out=120, looseness=5] node {\(a\)} (1)
            ;
  \end{tikzpicture}
  \end{minipage}
  \begin{minipage}{0.45\textwidth}
  \begin{tikzpicture}[->,>=stealth',shorten >=1pt,auto,node distance=6mm and 1cm,thick,initial text=]
  \tikzstyle{every state}=[scale=0.75,fill=blue!20,draw=blue!60,text=black]
  
  \node[initial,state] (1) {\(q_1\)};
  \node[state] (2) [right=of 1] {\(q_2\)};
  \node[accepting, state] (5) [right=of 2] {\(q_5\)};
  \node[state] (3) [above=of 2] {\(q_3\)};
  \node[state] (4) [below=of 2] {\(q_4\)};
  \node (A2) [left=of 1] {$\cA_2$};
  
  \path (1) edge[loop, in=80, out=140, looseness=5] node[above]  {\(a\)} (1)
        (1) edge node {\(a\)} (2)
        (1) edge[bend left=30] node {\(a\)} (3)
        (1) edge[bend right=30] node[below] {\(a,b\)} (4)
        (3) edge[bend left=30] node {\(a\)} (5)
        (3) edge[loop, in=60, out=120, looseness=5] node {\(a,b\)} (3)
        (2) edge node {\(c\)} (5)
        (2) edge[loop, in=60, out=120, looseness=4] node {\(a\)} (2)
        (4) edge[bend right=30] node[below] {\(b\)} (5)
          ;
  \end{tikzpicture}
  \end{minipage}
  \caption{Two automata \(\cA_1\) (left) and \(\cA_2\) (right) generating the regular languages \(\lang{\cA_1} = a^*(a+b+c)\) and \(\lang{\cA_2}= a^* (a(a+b)^*a+a^+c+ab+bb)\).}
    \label{fig:B}
  \end{figure}

\begin{example}\label{example:Word_Regular_LInc}
We illustrate the use of the left Nerode quasiorder in Algorithm~\AlgRegularW 
for solving the language inclusion \(\lang{\cA_1} \subseteq \lang{\cA_2}\), where \(\cA_1\) and \(\cA_2\) are the FAs shown in Figure~\ref{fig:B}.
The equations for  \(\cA_1\) are as follows:
\[
  \Eqn(\cA_1)=\begin{cases}
    X_1 = \varnothing \cup aX_1 \cup aX_2 \cup bX_2 \cup cX_2\\
    X_2 = \{\epsilon\}
  \end{cases} \enspace .
\]

\noindent
We have the following quotients (among others) for \(L = \lang{\cA_2}=a^* (a(a+b)^*a+a^+c+ab+bb)\): 
\begin{align*}
L \epsilon^{-1} = \; & a^* (a(a+b)^*a+a^+c+ab+bb) & L b^{-1} = \; & a^* (a + b)  \\
L a^{-1} = \; & a^* a(a+b)^* =a^+(a+b)^* &  L c^{-1} = \; & a^* a^+ = a^+\\
L w^{-1} = \; & a^* \text{ if{}f } w \in (a(a+b)^*a+ac+ab+bb) \span
\end{align*}

\noindent
Hence, among others, the following relations hold: 
\(c \leqq_L^l a\), \(c \leqq_L^l b\) and \(c \leqq_L^l w\)
 for all $w \in (a(a+b)^*a+ac+ab+bb)$.
Then, let us show the computation of the Kleene iterates performed by the Algorithm \AlgRegularW.
\begin{align*}
\vect{Y}^{(0)} &= \vect{\varnothing}\\
\vect{Y}^{(1)} &= \vectarg{\epsilon}{F} = \tuple{\varnothing, \{\epsilon\}} \\
\vect{Y}^{(2)} &= \vectarg{\epsilon}{F} \cup \Pre_{\cA_1}(\vect{Y}^{(1)}) 
= \tuple{\varnothing, \{\epsilon\}} \cup \tuple{\varnothing \cup a\varnothing \cup a\{\epsilon\} \cup b\{\epsilon\} \cup c\{\epsilon\}, \{\epsilon\}}\\
&= \tuple{\{a,b,c\}, \{\epsilon\}} \\
\vect{Y}^{(3)} &= \vectarg{\epsilon}{F} \cup \Pre_{\cA_1}(\vect{Y}^{(2)}) = 
\tuple{\varnothing, \{\epsilon\}} \cup \tuple{\varnothing \cup a\{a,b,c\} \cup a\{\epsilon\} \cup b\{\epsilon\} \cup c\{\epsilon\}, \{\epsilon\}}\\
&=
\tuple{\{aa,ab,ac, a, b, c\}, \{\epsilon\}} 
\end{align*}
It turns out that 
$\tuple{\{aa,ab,ac, a, b, c\}, \{\epsilon\}} \sqsubseteq_{\leqq_L^l} \tuple{\{a,b,c\}, \{\epsilon\}}$ because $c \leqq_L^l aa$, $c \leqq_L^l ab$ and 
$c \leqq_L^l ac$ hold, so that $\Kleene$ stops  with $\vect{Y}^{(3)}$ and outputs 
$\vect{Y}=\tuple{\{a,b,c\}, \{\epsilon\}}$. 
Since $c\in \vect{Y}_1$ and \(c \notin \lang{\cA_2}\), the Algorithm~\AlgRegularW correctly concludes that \(\lang{\cA_1} \subseteq \lang{\cA_2}\) does not hold. \eox
\end{example}

\subsubsection{On the Complexity of Nerode quasiorders}\label{Nerode-remark}\rm 
For the inclusion problem between languages generated by finite automata, deciding the 
(left or right) Nerode quasiorder relation between words  can be easily shown to be as hard as the language inclusion problem itself, which is PSPACE-complete.
In fact, given the automata \(\cA_1=(Q_1,\delta_1,I_1,F_1,\Sigma)\) and \(\cA_2=(Q_2,\delta_2,I_2,F_2,\Sigma)\), one can define the union automaton \(\cA_3\ud (Q_1\cup Q_2\cup\{q^{\iota}\}, \delta_3, \{q^{\iota}\}, F_1\cup F_2)\) where \(\delta_3 \) maps \((q^\iota,a)\) to \(I_1\), \( (q^\iota,b) \) to \(I_2\) and behaves like \(\delta_1\) or \(\delta_2\) elsewhere. Then, it turns out that \(a \leqq^r_{\lang{\cA_3}} b \Lra a^{-1}\lang{\cA_3} \subseteq b^{-1}\lang{\cA_3} \Lra \lang{\cA_1}\subseteq \lang{\cA_2}\).

Also, for the inclusion problem of a language generated by an 
automaton within the trace set of a one-counter net (cf.\ Section~\ref{sub:containment_in_one_counter_languages}), the right Nerode quasiorder is a right language-consistent well-quasiorder but it turns out to be undecidable (cf.~Lemma~\ref{lemma:RightNerodeOcnwqo}).

\subsection{State-based Quasiorders}\label{subsec:state-qos}
Consider an inclusion problem \(\lang{\cA_1} \subseteq \lang{\cA_2}\) where \(\cA_1\) and \(\cA_2\) are FAs.
In the following, we study a class of well-quasiorders based on \(\cA_2\), that we call state-based quasiorders. 
These quasiorders are strictly stronger (i.e., lower w.r.t.\ set inclusion of binary relations) than the Nerode quasiorders defined in Section~\ref{sub:the_left_nerode_quasi_order_relative_to_a_language} and sidestep the untractability or undecidability of Nerode quasiorders (cf.\ Section~\ref{Nerode-remark}) yet allowing to define an algorithm solving the language inclusion \(\lang{\cA_1} \subseteq \lang{\cA_2}\).

\subsubsection{Inclusion in Regular Languages.}
\label{sub:automata_based}

We  define the 
quasiorders \(\leq^l_{\cA}\)  and \(\leq^r_{\cA}\) on \(\Sigma^*\) 
induced by a FA \(\cA=\tuple{Q,\delta,I,F,\Sigma}\)
as follows: for all $u,v\in \Sigma^*$,
\begin{align}\label{eqn:state-qo}
u \leq^l_{\cA} v & \udr \pre^{\cA}_{u}(F) \subseteq \pre^{\cA}_{v}(F)
&
u \leq^r_{\cA} v & \udr \post^{\cA}_{u}(I) \subseteq \post^{\cA}_{v}(I) 
\end{align}
where, for all $X\in \wp(Q)$,
\(\pre_u^{\cA}(X) \ud \{q\in Q\mid  u\in W^{\cA}_{q,X}\}\)
and \(\post_u^{\cA}(X) \ud \{q'\in Q\mid  u\in W^{\cA}_{X,q'}\}\) denote the standard predecessor/successor state 
transformers in $\cA$.
The superscripts in $\leq^l_{\cA}$ and $\leq^r_{\cA}$ stand, resp., for left/right because the following result holds. 

\begin{lemma}\label{lemma:LAconsistent}
The relations \(\mathord{\leq^l_{\cA}}\) and \(\mathord{\leq^r_{\cA}}\) are, resp., decidable left and right 
\(\lang{\cA}\)-consistent wqos.
\end{lemma}
\begin{proof}
Since, for every \(u \in \Sigma^*\), \(\pre^{\cA}_u(F)\) is a finite and computable set, it turns out that \(\mathord{\leq^l_{\cA}}\) is a decidable wqo. 
Let us check that \(\mathord{\leq^l_{\cA}}\) is left \(\lang{A}\)-consistent according to Definition~\ref{def:LConsistent}~(\ref{eq:LConsistentPrecise})-(\ref{eq:LConsistentmonotonic}). 

\noindent
\eqref{eq:LConsistentPrecise}	By picking \(u\in \lang{\cA}\) and \(v\notin \lang{\cA}\) we have that \(\pre^{\cA}_u(F)\) contains some initial state while \(\pre^{\cA}_v(F)\) does not, hence \(u \nleq^l_{\cA} v\).

\noindent
\eqref{eq:LConsistentmonotonic} Let us check that $\leq^l_{\cA}$ is left monotonic.
Observe that, for all $x\in \Sigma^*$, $\pre^\cA_x$ is a monotonic function and that 
\begin{eqnarray}
\pre^{\cA}_{uv} = \pre^{\cA}_{u} \comp \pre^{\cA}_v \enspace .\label{eq:prepre}
\end{eqnarray}
Therefore, for all $x_1,x_2\in \Sigma^*$ and $a\in \Sigma$, 
\begin{align*}
x_1 \leq^l_{\cA} x_2 & \Ra  \quad\text{~[by def.\ of \(\leq^l_{\cA}\)]} \\
\pre^{\cA}_{x_1}(F) \subseteq \pre^{\cA}_{x_2}(F) & \Ra  \quad\text{~[as $\pre^\cA_a$ is monotonic]} \\
\pre^{\cA}_{a}(\pre^{\cA}_{x_1}(F)) \subseteq \pre^{\cA}_{a}(\pre^{\cA}_{x_2}(F)) & \Lra  \quad\text{~[by~\eqref{eq:prepre}]} \\
\pre^{\cA}_{ax_1}(F) \subseteq \pre^{\cA}_{ax_2}(F) & \Lra  \quad\text{~[by def.\ of \(\leq^l_{\cA}\)]} \\
ax_1 \leq^l_{\cA} ax_2 & \enspace . 
\end{align*}
\noindent
The proof that \(\leq_{\cA}^r\) is a decidable right \(\lang{\cA}\)-consistent quasiorder is symmetric. 
\end{proof}

As a consequence, Theorem~\ref{theorem:quasiorderAlgorithm} applies to the wqo \(\mathord{\leq^l_{\cA_2}}\) (and 
\(\mathord{\leq^r_{\cA_2}}\)), so that one can instantiate the algorithm \AlgRegularW to  $\mathord{\leq^l_{\cA_2}}$ for deciding 
an inclusion $\lang{\cA_1}\subseteq \lang{\cA_2}$. 

Turning back to the left Nerode wqo
$\leqq_{\lang{\cA_2}}^l$, it turns out that the following equivalences hold:
\begin{align*}
u \leqq_{\lang{\cA_2}}^l v  \Lra \lang{\cA_2}u^{-1} \subseteq \lang{\cA_2} v^{-1} 
\Lra W_{I,\pre^{\cA_2}_u(F)} \subseteq W_{I,\pre^{\cA_2}_v(F)} \enspace .
\end{align*}
Since \(\pre^{\cA_2}_u(F) \subseteq \pre^{\cA_2}_v(F)$ entails  $W_{I,\pre^{\cA_2}_u(F)} \subseteq W_{I,\pre^{\cA_2}_v(F)}\), it follows that \(u \leq_{\cA_2}^l v \Ra u \leqq_{\lang{\cA_2}}^l v\) and, in turn,
\( \rho_{\leqq_{\lang{\cA_2}}^l}(\wp(\Sigma^*)) \subseteq \rho_{\leq^l_{\cA_2}}(\wp(\Sigma^*))\).

\begin{example}\label{example:Word_Regular_LInc:states}
We illustrate the left state-based quasiorder by using it to solve the language inclusion \(\lang{\cA_1} \subseteq \lang{\cA_2}\) of Example~\ref{example:Word_Regular_LInc}.
We have, among others, the following set of predecessors of $F_{\cA_2}$:
\begin{align*}
\pre_{\epsilon}^{\cA_2}(F_{\cA_2}) & = \{q_5\} & \pre_{a}^{\cA_2}(F_{\cA_2}) & = \{q_3\} & \pre_{b}^{\cA_2}(F_{\cA_2}) & = \{q_4\} & \pre_{c}^{\cA_2}(F_{\cA_2}) & = \{q_2\} \\
\pre_{aa}^{\cA_2}(F_{\cA_2}) & = \{q_1, q_3\} & \pre_{ab}^{\cA_2}(F_{\cA_2}) & = \{q_1\} & \pre_{ac}^{\cA_2}(F_{\cA_2}) & = \{q_1, q_2\} & \pre_{aaa}^{\cA_2}(F_{\cA_2}) & = \{q_1,q_3\}\\
\pre_{aab}^{\cA_2}(F_{\cA_2}) & = \{q_1\} & \pre_{aac}^{\cA_2}(F_{\cA_2}) & = \{q_1,q_2\}
\end{align*}
Recall from Example~\ref{example:Word_Regular_LInc} that, for the Nerode quasiorder, we have \(c \leqq_{\lang{\cA_2}}^l b\), \(c \leqq_{\lang{\cA_2}}^l a\) while none of these relations hold for \(\leq^l_{\cA_2}\).

\noindent
Let us next show the Kleene iterates computed by Algorithm \AlgRegularW when using the quasiorder \(\leq^l_{\cA_2}\).
\begin{align*}
\vect{Y}^{(0)} &=\vect{\varnothing}\\
\vect{Y}^{(1)} &= \vectarg{\epsilon}{F} = \tuple{\varnothing, \{\epsilon\}} \\
\vect{Y}^{(2)} &= \vectarg{\epsilon}{F} \cup \Pre_{\cA_1}(\vect{Y}^{(1)}) = \tuple{\{a, b, c\}, \{\epsilon\}}\\
\vect{Y}^{(3)} &= \vectarg{\epsilon}{F} \cup \Pre_{\cA_1}(\vect{Y}^{(2)}) = \tuple{\{aa, ab, ac, a, b, c\}, \{\epsilon\}} \\
\vect{Y}^{(4)} &= \vectarg{\epsilon}{F} \cup \Pre_{\cA_1}(\vect{Y}^{(3)}) = \tuple{\{aaa, 
aab, aac, aa, ab, ac, a, b, c\}, \{\epsilon\}}
\end{align*}
It turns out that $\tuple{\{aaa, 
aab, aac, aa, ab, ac, a, b, c\}, \{\epsilon\}} \sqsubseteq_{\leq^l_{\cA_2}} \tuple{\{aa, ab, ac, a, b, c\}, \{\epsilon\}}$ so that $\Kleene$ outputs the vector
$\vect{Y}=\tuple{\{aa, ab, ac, a, b, c\}, \{\epsilon\}}$.
Since $c\in \vect{Y}_0$ and \(c \notin \lang{\cA_2}\), Algorithm~\AlgRegularW concludes that the language inclusion \(\lang{\cA_1} \subseteq \lang{\cA_2}\) does not hold. \eox
\end{example}

\subsubsection{Simulation-based Quasiorders.}\label{subsec:simulation}
Recall that a \emph{simulation} on a FA $\cA= \tuple{Q,\delta,I,F,\Sigma}$  is a binary relation \(\mathord{\preceq} \subseteq Q\times Q\) such that for all $p,q\in Q$ such that \(p\preceq q\) the following two conditions hold:
\begin{compactenum}[\upshape(\rm i\upshape)]
\item if \(p\in F\) then \(q\in F\);
\item for every transition \(p \xrightarrow{a} p'\), there exists a transition \(q \xrightarrow{a} q'\) such that \(p'\preceq q'\).
\end{compactenum}
It is well known that simulation relations are closed under arbitrary unions, where
the greatest (w.r.t.\ inclusion) simulation relation $\mathord{\preceq_A} \ud 
\cup \{\mathord{\preceq}\subseteq Q\times Q \mid \mathord{\preceq}$ is a simulation on $\cA\}$
is a quasiorder, called simulation quasiorder
of $\cA$. 
It is also well known that simulation implies language inclusion, i.e., if $\preceq$ is 
a simulation on $\cA$ then
\[ q \preceq q'  \Ra W^{\cA}_{q,F}\subseteq W^{\cA}_{q',F} \enspace .\]
A relation \(\mathord{\leq}\subseteq Q\times Q\) on states 
can be lifted in the standard universal/existential way to a relation $\leq^{\forall\exists}\subseteq \wp(Q)\times \wp(Q)$ on sets of states as follows: 
\[ X \leq^{\forall\exists} Y \:\udr\: \forall x\in X, \exists y\in Y,\: x\leq y \enspace.\]
In particular, if $\leq$ is a quasiorder 
then $\leq^{\forall\exists}$ is a quasiorder as well. 
Also, if $\preceq$ is a simulation relation then its lifting $\preceq^{\forall\exists}$ is such 
that \(X \preceq^{\forall\exists} Y \Ra W^{\cA}_{X,F} \subseteq W^{\cA}_{Y,F}\) holds. This suggests us to
define a \emph{right simulation-based quasiorder} \(\preceq^{r}_{\cA}\) on $\Sigma^*$ induced by a simulation $\preceq$ on $\cA$ as follows: for all $u,v\in \Sigma^*$,
\begin{equation}\label{eq:sim-qo}
u \preceq^{r}_{\cA} v \:\udr\:  \post^{\cA}_u(I) \preceq^{\forall\exists} \post^{\cA}_v(I) \enspace .
\end{equation}
\begin{lemma}
	Given a simulation relation \(\mathord{\preceq}\) on $\cA$, the right simulation-based quasiorder \(\mathord{\preceq^r_{\cA}}\) is a decidable right \(\lang{\cA}\)-consistent wqo.
\end{lemma}
\begin{proof}

Let \(u \in \lang{\cA}\) and \(v \notin \lang{\cA}\), so that
\(F \cap \post^{\cA}_u(I) \neq \varnothing\) and \((F \cap \post^{\cA}_v(I)) = \varnothing\) hold. Hence, there exists \(q \in \post^{\cA}_u(F) \cap F\) such that \(q \preceq^r_{\cA} q'\) for no \(q' \in \post^{\cA}_v(F)\) since, by simulation, this would imply \(q' \in \post^{\cA}_v(F) \cap F\), which would contradict \(F \cap \post^{\cA}_v(I) = \varnothing\).
Therefore, \(u \npreceq^r_{\cA} v\) holds.

\noindent
Next we show that \(\preceq^r_{\cA}\) is right monotonic. By \eqref{def-leftmon}, we check that for all $u,v\in \Sigma^*$ and $a\in \Sigma$,  
\(u\preceq^r_{\cA} v \Ra ua \preceq^r_{\cA} va\):  
\begin{align*}
u \preceq^r_{\cA} v & \Lra \quad \text{~[def. \(\preceq^r_{\cA}\)]}\\
\post^{\cA}_u(I) \preceq^{\forall\exists} \post^{\cA}_v(I) & \Lra \quad \text{~[by def.\ of \(\preceq^{\forall\exists}\)]} \\
\forall x \in \post^{\cA}_u(I), \exists y \in \post^{\cA}_v(I), x \preceq y & \Ra \quad \text{~[by def.\ of \(\preceq\)]} \\
\forall x \ggoes{a} x' ,\; x \in \post^{\cA}_u(I),\; \exists y \ggoes{a} y' ,\; y\in \post^{\cA}_v(u), x' \preceq y' & \Lra \quad \text{~[by \(\post^{\cA}_{u}\comp \post^{\cA}_{a} = \post^{\cA}_{ua}(I)\)]}\\
\forall x' \in \post^{\cA}_{ua}(I), \exists y' \in \post^{\cA}_{va}(I), \; x' \preceq y' & \Lra \quad\text{~[by def.\ of \(\preceq^{\forall\exists}\)]}\\
\post^{\cA}_{ua}(I) \preceq^{\forall\exists} \post^{\cA}_{va}(I) & \Lra \quad\text{~[by def.\ of \(\preceq_{\cA}^r\)]}\\
ua \preceq_{\cA}^r va & \enspace .
\end{align*}
Thus, \(\preceq_{\cA}^r\) is a right \(\lang{\cA}\)-consistent quasiorder.

\noindent
Finally, since \(\wp(Q)\) is finite, it follows that \(\preceq_{\cA}^r\) is a well-quasiorder and, since \(\post^{\cA}_u(I)\) is finite and computable for every \(u\in \Sigma^*\), it follows that \(\preceq_{\cA}^r\) is decidable. %
\end{proof}

Thus, once again, Theorem~\ref{theorem:quasiorderAlgorithmR} applies to 
\(\mathord{\preceq^r_{\cA_2}}\) and this allows us to instantiate the 
algorithm \AlgRegularWr to the quasiorder 
$\mathord{\preceq^r_{\cA_2}}$ for deciding an inclusion $\lang{\cA_1}\subseteq \lang{\cA_2}$. 

Note that it is possible to define a left simulation \(\preceq^{\forall\exists}_{R}\) on an automaton \(\cA\) by applying \(\preceq^{\forall\exists}\) on the reverse automaton 
$\cA^R$ of \(\cA\) where arrows are flipped and initial/final states are swapped.
This left simulation induces a \emph{left simulation-based quasiorder} on \(\Sigma^*\) as follows: for all $u,v\in \Sigma^*$,

\begin{equation}\label{eq:sim-qo:left}
u \preceq^{l}_{\cA} v \:\udr\:  \pre^{\cA}_u(F) \preceq^{\forall\exists}_R \pre^{\cA}_v(F) \enspace .
\end{equation}

It is straightforward to check that Theorem~\ref{theorem:quasiorderAlgorithm} applies to \(\mathord{\preceq^l_{\cA_2}}\) and, therefore, we can instantiate the Algorithm~\AlgRegularW for deciding \(\lang{\cA_1} \subseteq \lang{\cA_2}\). 

\begin{example}\label{example:Word_Regular_LInc:sim}
Let us illustrate the use of the left simulation-based quasiorder to solve the language inclusion \(\lang{\cA_1} \subseteq \lang{\cA_2}\) of Example~\ref{example:Word_Regular_LInc}.
For the set $F_{\cA_2}$ of final states \(\cA_2\) 
we have the same set of predecessors computed in Example~\ref{example:Word_Regular_LInc:states} and, among others, the following left simulations between these sets w.r.t.\ the simulation quasiorder 
$\preceq_{\cA_2^R}$
of the reverse of \(\cA_2\)
(recall that 
\(\preceq^{\forall\exists}\) is defined w.r.t.\ simulations of $\cA_2^R$):
\begin{align*}
\pre_{c}^{\cA_2}(F_{\cA_2}) = \{q_2\}  & \preceq^{\forall\exists}_R \{q_3\} = \pre_{a}^{\cA_2}(F_{\cA_2}) & \pre_{c}^{\cA_2}(F_{\cA_2}) = \{q_2\} & \preceq^{\forall\exists}_R \{q_4\} = \pre_{b}^{\cA_2}(F_{\cA_2}) \\
\pre_{c}^{\cA_2}(F_{\cA_2}) = \{q_2\} & \preceq^{\forall\exists}_R \{q_1,q_3\} = \pre_{aa}^{\cA_2}(F_{\cA_2}) & \pre_{c}^{\cA_2}(F_{\cA_2}) = \{q_2\} & \preceq^{\forall\exists}_R \{q_1\} = \pre_{ab}^{\cA_2}(F_{\cA_2})\\
 \pre_{c}^{\cA_2}(F_{\cA_2}) = \{q_2\} & \preceq^{\forall\exists}_R \{q_1,q_2\} = \pre_{ac}^{\cA_2}(F_{\cA_2})
\end{align*}
because $q_2 \preceq_{\cA_2^R} q_1$, $q_2 \preceq_{\cA_2^R} q_3$ and $q_2 \preceq_{\cA_2^R} q_4$ hold. 

\noindent
Let us show the computation of the Kleene iterates performed by Algorithm \AlgRegularW when using the quasiorder \(\sqsubseteq_{\mathord{\preceq_{\cA_2}^l}}\) as abstract inclusion check:
\begin{align*}
\vect{Y}^{(0)} &= \vect{\varnothing}\\
\vect{Y}^{(1)} &= \vectarg{\epsilon}{F} = \tuple{\varnothing, \{\epsilon\}} \\
\vect{Y}^{(2)} &= \vectarg{\epsilon}{F} \cup \Pre_{\cA_1}(\vect{Y}^{(1)}) = \tuple{\{a, b, c\},\{\varepsilon\}}\\
\vect{Y}^{(3)} &= \vectarg{\epsilon}{F} \cup \Pre_{\cA_1}(\vect{Y}^{(2)}) = \tuple{\{aa,ab,ac, a, b, c\}, \{\varepsilon\}} %
\end{align*}
It turns out that $\tuple{\{aa, ab, ac, a, b, c\}, \{\epsilon\}} \sqsubseteq_{\preceq^l_{\cA_2}} \tuple{\{a, b, c\}, \{\epsilon\}}$, so that $\Kleene$ outputs the vector
$\vect{Y}=\tuple{\{a, b, c\}, \{\epsilon\}}$. Thus, once again,
since $c\in \vect{Y}_0$ and \(c \notin \lang{\cA_2}\), Algorithm~\AlgRegularW concludes that \(\lang{\cA_1} \subseteq \lang{\cA_2}\) does not hold. \eox
\end{example}

Let us observe that \(u \preceq^{r}_{\cA_2} v\) implies \(W_{\post^{\cA_2}_u(I),F} \subseteq W_{\post^{\cA_2}_v(I),F}\), which is equivalent to the right Nerode quasiorder \(u\leqq^r_{\lang{\cA_2}} v\) for $\lang{\cA_2}$ defined in \eqref{def-Nerodeqo}, 
so that \(u \preceq^{r}_{\cA_2} v \Ra u\leqq^r_{\lang{\cA_2}} v\) holds. Furthermore, for the state-based quasiorder defined in
\eqref{eqn:state-qo}, we have that
\(u \leq^{r}_{\cA_2} v \Ra u\preceq^r_{\cA_2} v\) trivially holds.
Summing up, the following containments relate (the right versions of) state-based, 
simulation-based and Nerode quasiorders:
\[
\mathord{\leq^r_{\cA_2}} \,\subseteq\, \mathord{\preceq^r_{\cA_2}} \,\subseteq\, \mathord{\leqq^r_{\lang{\cA_2}}} \enspace .
\]
All these quasiorders are decidable \(\lang{\cA_2}\)-consistent wqos so
that the 
algorithm \AlgRegularW can be instantiated to each of them for deciding an inclusion $\lang{\cA_1}\subseteq \lang{\cA_2}$.
Examples~\ref{example:Word_Regular_LInc},~\ref{example:Word_Regular_LInc:states} and~\ref{example:Word_Regular_LInc:sim} show how \AlgRegularW  behaves for each of the three quasiorders considered in this section.
Despite their simplicity, the examples show the differences in the behavior of the algorithm when considering the different quasiorders.
In particular, we observe that the iterations of $\Kleene$ for \(\mathord{\leqq^r_{\lang{\cA_2}}}\) coincides with those for \(\mathord{\preceq^r_{\cA_2}}\) and, as expected, these Kleene iterates converge faster than those for \(\mathord{\leq^r_{\cA_2}}\).
Recall that \(\mathord{\leqq^r_{\lang{\cA_2}}}\) is the coarsest well-quasiorder for which Algorithm~\AlgRegularW works, hence its corresponding Kleene iterates exhibit optimal behavior in terms of number of iterations to converge. 
The drawback of using the Nerode  quasiorder \(\mathord{\leqq^r_{\lang{\cA_2}}}\) 
is that it requires checking language inclusion in order to decide whether two words are related, and this is a PSPACE-complete problem.
Therefore, the coincidence of the Kleene iterates for \(\mathord{\leqq^r_{\lang{\cA_2}}}\) and \(\mathord{\preceq^r_{\cA_2}}\) is of special interest since it highlights that Algorithm~\AlgRegularW might exhibit optimal behavior while using a ``simpler'' (i.e., finer) well-quasiorder such as \(\mathord{\preceq^r_{\cA_2}}\), which is a polynomial approximation of \(\mathord{\leqq^r_{\lang{\cA_2}}}\).

\subsection{Inclusion in Traces of One-Counter Nets.}%
\label{sub:containment_in_one_counter_languages}
We show that our framework can be instantiated to systematically derive an algorithm for deciding an inclusion \(\lang{\cA} \subseteq L_2\) where \(L_2\) is the trace set of a one-counter net (OCN). 
This is accomplished by defining a decidable \(L_2\)-consistent quasiorder so that Theorem~\ref{theorem:quasiorderAlgorithm} can be applied.

Intuitively, an OCN is a FA endowed with a nonnegative integer counter which 
can be incremented, decremented or
left unchanged by a transition. 
Formally, a one-counter net~\cite{hofman_trace_2018} is a tuple $\cO=\tuple{Q,\Sigma,\delta}$ where $Q$ is a finite set of states, $\Sigma$ is an alphabet and $\delta\subseteq Q\times \Sigma\times \{-1,0,1\}\times Q$ is a set of transitions.
A configuration of \(\cO\) is a pair $qn$ consisting of a state \(q\in Q\) and a value \(n\in\bN\) for the counter. 
Given two configurations \(qn, q'n'\in Q\times \bN\) we write \(qn \xrightarrow{a} q'n'\) and call it a \(a\)-step (or simply step) if there exists a transition \( (q,a,d,q')\in\delta \) such that \(n'=n+d\).
Given \(qn\in Q\times\bN\), the \emph{trace set} \(T(qn)\subseteq \Sigma^*\) of an OCN is defined as follows:
\begin{align*}
	T(qn) & \ud \{u \in \Sigma^* \mid Z_u^{qn} \neq \varnothing\} \quad \text{ where } \\
	Z_u^{qn} & \ud \{ q_k n_k \in Q\times \bN \mid qn=q_0n_0 \xrightarrow{a_1} q_1n_1\xrightarrow{a_2}\cdots \xrightarrow{a_k} q_kn_k,\: a_1\cdots a_k=u \}\enspace .
\end{align*}

\noindent
Observe that \(Z_{\epsilon}^{qn}= \{ qn \}\) and \(Z_u^{qn}\) is a finite set for every word \(u\in\Sigma^*\).

Let us consider the poset $\tuple{\bN_{\bot}\ud \bN\cup\{\bot\},\leq_{\bN_{\bot}}}$ where \(\bot\leq_{\bN_{\bot}} n\) holds for all \(n\in\bN_{\bot}\), while for all $n,n'\in
\bN$, $n\leq_{\bN_{\bot}}n'$ is the standard ordering relation between numbers.  
For a finite set of states \(S \subseteq Q\times\bN\), define the so-called macro state \(M_S \colon Q \ra \bN_{\bot}\) as follows: 
\[M_S( q ) \ud \max \{ n\in \bN \mid q n \in S \}\,,\]
where $\max\varnothing\ud\bot$. %
Let us define the following quasiorder $\mathord{\leq_{{qn}}^r}\subseteq \Sigma^*\times \Sigma^*$:
\begin{equation}\label{eq:ocnleq}
	u \leq_{{qn}}^r v \:\udr\:\forall q\in Q,\, M_{Z_u^{qn}}(q) \leq_{\bN_{\bot}} M_{Z_v^{qn}}(q) \enspace .
\end{equation}

\begin{figure}[t]
	\centering
\begin{tikzpicture}[->,>=stealth',shorten >=1pt,auto,node distance=6mm and 1cm,thick,initial text=]
	\tikzstyle{every state}=[inner sep=3pt,minimum size=19pt,fill=blue!20,draw=blue!60,text=black]
            \node[state,initial] (P) at (0,0) {$q_1$};
            \node[state] (Q) at (2, 0)        {$q_2$};
            \node[state] (R) at (1,1.5)       {$q_3$};   
          
            \path[->]
                  (P) edge node[above] {$a, 1$} (Q)
                  (Q) edge node[right] {$a, 0$} (R)
                  (R) edge node[left]  {$a,-1$} (P);
     \path[->]
		(R) edge [loop left] node[left] {$a,1$} (R);           
	\end{tikzpicture}
	\caption{ A one-counter net \(\cO\).}%
	\label{fig:ocnexample}
\end{figure}

\begin{example}\label{example:ocn}
	Figure~\ref{fig:ocnexample} depicts an OCN over the singleton alphabet \(\Sigma = \{ a \}\).
	For \(\cO\) we have the following sets:
	\begin{align*}
		Z_{\epsilon}^{q_10} &= \{ q_10 \} & Z_{a}^{q_10} &= \{ q_21 \} & Z_{aa}^{q_10}&= \{ q_31 \}\\
		Z_{aaa}^{q_10} &= \{ q_32, q_10 \} & Z_{aaaa}^{q_10} &= \{ q_33,q_11,q_21 \} & Z_{b}^{q_10}&= \varnothing
	\end{align*}
	Hence, we have that: 		
	\begin{align*}
		M_{Z_{\epsilon}^{q_10}} &= 
		\left(%
			\begin{array}{c}
				q_1 \mapsto 0 \\
				q_2 \mapsto \bot  \\
				q_3 \mapsto \bot
		  \end{array}%
		\right) %
		& M_{Z_{a}^{q_10}} &=
		\left(%
			\begin{array}{c}
				q_1 \mapsto\bot \\
				q_2 \mapsto 1 \\
				q_3 \mapsto \bot
		  \end{array}%
		\right) %
		& M_{Z_{aa}^{q_10}} &=
		\left(%
			\begin{array}{c}
				q_1 \mapsto\bot \\
				q_2 \mapsto \bot \\
				q_3 \mapsto 1 
		  \end{array}%
		\right)
		& M_{Z_{aaa}^{q_10}} &=
		\left(%
			\begin{array}{c}
				q_1 \mapsto 0 \\
				q_2 \mapsto \bot \\
				q_3 \mapsto 2 
		  \end{array}%
		\right)
	\end{align*}
	Therefore, the words \(\epsilon,a\) and \(aa\) are pairwise incomparable for \(\mathord{\leq_{{q_10}}^r}\), while we have that \(aa \leq_{q_10}^r aaa \)  and \(\epsilon \leq_{q_10}^r aaa \).\eox
\end{example}

\begin{lemma}\label{lemma:ocnwqo}
	Let \(\cO\) be an OCN. For any configuration $q n$ of $\cO$, \(\mathord{\leq_{{qn}}^r}\) is a right \(T(qn)\)-consistent decidable wqo.
\end{lemma}
\begin{proof}
It follows from Dickson's Lemma \cite[Section~II.7.1.2]{Sakarovitch} that \(\mathord{\leq_{{qn}}^r}\) is a wqo.
Since \(Z_u^{qn}\) and \(Z_v^{qn}\) are finite sets of configurations, the macro state functions 
\(M_{Z_u^{qn}}\) and \(M_{Z_v^{qn}}\) are computable, hence the relation \(\mathord{\leq_{{qn}}^r}\) is decidable.
If \(u\in T(qn)\) and \(v\notin T(qn)\) then \(u  \not\leq_{{qn}}^r v \), otherwise we would have that 
\(M_{Z_u^{qn}}(q')\neq \bot\) for some \(q'\in Q\), hence \(M_{Z_v^{qn}}(q')\neq \bot\), and this would be 
a contradiction because \(Z_v^{qn}=\varnothing\), so that \(M_{Z_v^{qn}}(q')= \bot\).

\noindent
Finally, let us show that 
\(u \leq_{{qn}}^r v\) implies \(ua \leq_{{qn}}^r va\)
for all \(a\in \Sigma\), since, by \eqref{def-leftmon}, this is equivalent to the fact that $\leq_{{qn}}^r$ is right monotonic. 
We proceed by contradiction.
Assume that \(u \leq_{{qn}}^r v\) and \(\exists q' \in Q\),  \(M_{Z^{qn}_{ua}}(q') \not\leq_{\bN_{\bot}}  M_{Z^{qn}_{va}}(q')\).
Then, \(m_1\ud\max\{n \mid pn \in Z^{qn}_{ua}\} \not\leq_{\bN_{\bot}} m_2\ud\max\{n \mid pn \in Z^{qn}_{va}\}\), which implies, since
$m_1\neq \bot$, that
$m_1,m_2\in \bN$ and $m_1 > m_2$. 
Thus, for all \( (q,a,d,q') \in \delta\) we have \(q'(m_1-d) \in Z_u^{qn}\) and \(q'(m_2-d) \in Z_v^{qn}\).
Since \(m_1-d > m_2-d\) we have that \(\max\{n \mid pn \in Z_u^{qn}\} > \max\{n \mid pn \in Z_v^{qn}\}\), which contradicts \(u \leq_{{qn}}^r v\). %
\end{proof}

By Theorem~\ref{theorem:quasiorderAlgorithm},
Lemma~\ref{lemma:ocnwqo} and the decidability of membership \(u\inqm T(qn)\),
the following known decidability result for inclusion of regular languages into traces of OCNs~\cite[Theorem 3.2]{JANCAR1999476} is systematically derived as a consequence of our algorithmic framework.

\begin{corollary}\label{theorem:ocncontainment}
Let \(\cA\) be a FA and \(\cO\) be an OCN. For any configuration \(qn \) of $\cO$, the language inclusion problem
\(\lang{\cA} \subseteq T(qn)\) is decidable.
\end{corollary}

\begin{example}
	Consider the OCN of Figure~\ref{fig:ocnexample} and the problem of deciding whether \(\Sigma^* = a^*\) is included into \(T(q_00)\), i.e., whether the trace set of \(\cO\) is universal.
	By considering 
	the equation \(X = X a \cup \{ \epsilon \}\) which defines \(\Sigma^*\), it turns
	out that the 
  Kleene iterates computed by Algorithm \AlgRegularW when using the abstract inclusion check
  given by  \(\sqsubseteq_{\leq^r_{q_00}}\) are as follows:
\begin{align*}
	Y^{(0)} &=\varnothing & Y^{(1)} &= \{ \epsilon \} & Y^{(2)} &= \{ a,\epsilon \} &
	Y^{(3)} &= \{ aa,a,\epsilon \} & Y^{(4)} &= \{ aaa,aa,a,\epsilon \} \enspace .
\end{align*}
We have that $Y^{(4)} \sqsubseteq_{\leq^r_{q_00}} Y^{(3)}$ because 
\(aa \leq^r_{q_00} aaa \) holds, as shown in Example~\ref{example:ocn}, so that
the output of $\Kleene$ is $Y^{(3)}=\{ aa,a,\epsilon \}$. 
Since $\{aa,a,\epsilon\}$ is a set of traces of \(\cO\) (i.e. \(\{ aa,a,\epsilon \}\subseteq T(q_00)\)) we conclude that \(\cO\) is universal. \eox
\end{example}

Moreover, by exploiting 
Lemma~\ref{lemma:ocnwqo} and \cite[Theorem 20]{Hofman:2013:DWS:2591370.2591405}, the following result settles a conjecture made by de Luca and Varricchio \citeyearpar[Section~6]{deLuca1994} on the right Nerode quasiorder
for traces of OCNs. 

\begin{lemma}\label{lemma:RightNerodeOcnwqo}
The right Nerode quasiorder \(\mathord{\leqq^r_{T(qn)}}\) is an undecidable well-quasiorder.
\end{lemma}
\begin{proof}
As already recalled, de Luca and Varricchio~\citeyearpar[Section~2, point~4]{deLuca1994} 
show that \(\mathord{\leqq^r_{T(qn)}}\) is maximum in the set of all right \(T(qn)\)-consistent quasiorders, so that  \(u\leq^r_{{qn}}v\) $\Ra$ {$u\leqq^r_{T(qn)} v$}, for all \(u,v\in\Sigma^*\).
By Lemma~\ref{lemma:ocnwqo}, \(\leq^r_{{qn}}\) is a wqo, so that  \(\mathord{\leqq^r_{T(qn)}}\) is a wqo as well. 
Undecidability of \(\mathord{\leqq^r_{T(qn)}}\) follows from the undecidability of the trace inclusion problem for nondeterministic OCNs \cite[Theorem 20]{Hofman:2013:DWS:2591370.2591405} by an argument similar to the automata case.
\end{proof}

It is worth remarking that, by Lemma~\ref{lemma:leftrightnerodegoodqo}~(\ref{lemma:leftrightnerodegoodqo:Consistent}), the left and right Nerode quasiorders \(\mathord{\leqq^l_{T(qn)}}\) and \(\mathord{\leqq^r_{T(qn)}}\) are \(T(qn)\)-consistent. 
However, the left Nerode quasiorder does not need to be a wqo, otherwise \(T(qn)\) would be regular.

We conclude this section by conjecturing that our framework could be instantiated for extending 
Corollary~\ref{theorem:ocncontainment} to traces of Petri Nets, a result 
which is already known to be true~\cite{JANCAR1999476}.

\section{A Novel Perspective on the Antichain Algorithm}%
\label{sec:novel_perspective_AC}

In this section we 
will show how 
to solve the language inclusion problem by computing Kleene iterates 
in an abstract domain of $\wp(\Sigma^*)$ as  
defined by a Galois connection. 
This is of practical interest since it allows us to decide a language inclusion problem \(\lang{\cA} \subseteq L_2\) by manipulating an 
automaton representation for \(L_2\). 

\subsection{A Language Inclusion Algorithm Using Galois Connections}\label{sec:usingGC}
The next result provides a formulation of Theorem~\ref{theorem:FiniteWordsAlgorithmGeneral}
by using a Galois Connection  \(\tuple{\wp(\Sigma^*),\subseteq} \galois{\alpha}{\gamma}\tuple{D,\leq_D}\)
rather than a closure operator $\rho \in \uco(\wp(\Sigma^*)$ and shows 
how to design 
an algorithm that solves a language inclusion 
$\lang{\cA} \subseteq L_2$ by computing Kleene iterates on the abstract domain $D$.

\begin{theorem}\label{theorem:EffectiveAlgorithm}
Let \(\cA=\tuple{Q,\delta,I,F,\Sigma}\) be a FA and \(L_2\in \wp(\Sigma^*)\).
Let \( \tuple{D,\leq_D}\) be a poset and 
\(\tuple{\wp(\Sigma^*),\subseteq} \galois{\alpha}{\gamma}\tuple{D,\leq_D}\) be a GC.
Assume that the following properties hold:
\begin{compactenum}[\upshape(\rm i\upshape)]
\item \(L_2\in\gamma(D)\) and for all \( a \in \Sigma\) and \(X \in \wp(\Sigma^*)\), \(\gamma\alpha(a X) = \gamma\alpha(a \gamma\alpha(X))\).\label{theorem:EffectiveAlgorithm:prop:rho}
\item \(\tuple{D,\leq_D,\sqcup,\bot_D}\) is an effective domain, meaning that: \(\tuple{D,\leq_D,\sqcup,\bot_D}\) is an ACC join-semilattice with bottom $\bot_D$, 
every element of \(D\) has a finite representation, the binary relation 
\(\leq_D\) is decidable and the binary lub \(\sqcup\) is computable.\label{theorem:EffectiveAlgorithm:prop:absdecidable}
\item There is an algorithm, say \(\Pre^{\sharp}\), which computes \(\alpha\comp \Pre_{\cA}\comp \gamma\).\label{theorem:EffectiveAlgorithm:prop:abspre}
\item There is an algorithm, say \(\epsilon^{\sharp}\), which computes \(\alpha(\vectarg{\epsilon}{F})\).\label{theorem:EffectiveAlgorithm:prop:abseps}
\item There is an algorithm, say \(\absincl\), which decides 
\(\vect{X}^\sharp \leq_D \alpha(\vectarg{L_2}{I})\), for all 
\(\vect{X}^\sharp\in \alpha(\wp(\Sigma^*))^{|Q|}\).\label{theorem:EffectiveAlgorithm:prop:absincl}
\end{compactenum}
\medskip
Then,

\medskip
\(\tuple{Y^\sharp_q}_{q\in Q} := \Kleene (\leq_D,\lambda \vect{X}^\sharp\ldotp\epsilon^{\sharp} \sqcup \Pre^{\sharp}(\vect{X}^\sharp), \vect{\bot_D})\)\emph{;}

\emph{\textbf{return}} \(\absincl(\tuple{Y^\sharp_q}_{q\in Q})\)\emph{;}

\medskip
\noindent
is a decision algorithm for \(\lang{\cA} \subseteq L_2\).
\end{theorem}
\begin{proof}
Let \(\rho \ud \gamma \alpha\in \uco(\wp(\Sigma^*))\), so that hypothesis~(\ref{theorem:EffectiveAlgorithm:prop:rho}) can be stated as 
\(\rho(L_2)=L_2\) and \(\rho(aX) = \rho(a\rho(X))\), and this allows us to apply 
Corollary~\ref{corol:rholfp}.
It turns out that:
\begin{align*}
	\lang{\cA}\subseteq L_2 &\Lra\quad
	\text{~[by~Corollary~\ref{corol:rholfp} and \eqref{equation:checklfpRLintoRL}]}\\
	\lfp(\lambda \vect{X}\ldotp\gamma\alpha (\vectarg{\epsilon}{F} \cup \Pre_{\cA}(\vect{X}))) \subseteq \vectarg{L_2}{I} &\Lra \quad
	\text{~[by Lemma~\ref{lemma:alpharhoequality}]}\\
		\gamma(\lfp (\lambda \vect{X}^\sharp\ldotp \alpha (\vectarg{\epsilon}{F} \cup \Pre_{\cA}(\gamma(\vect{X}^\sharp))))) \subseteq \vectarg{L_2}{I} &\Lra 
	\quad \text{~[by GC]}\\
	\gamma(\lfp (\lambda \vect{X}^\sharp\ldotp\alpha(\vectarg{\epsilon}{F}) \sqcup \alpha(\Pre_{\cA}(\gamma(\vect{X}^\sharp))))) \subseteq \vectarg{L_2}{I} &\Lra 
	\quad \text{~[by GC and since, by (i), $L_2\in \gamma(D)$]}\\
	\lfp (\lambda \vect{X}^\sharp\ldotp\alpha(\vectarg{\epsilon}{F}) \sqcup \alpha(\Pre_{\cA}(\gamma(\vect{X}^\sharp)))) \leq_D \alpha(\vectarg{L_2}{I}) &\end{align*}

\noindent
Thus, by hypotheses~(\ref{theorem:EffectiveAlgorithm:prop:absdecidable}), (\ref{theorem:EffectiveAlgorithm:prop:abspre}) and (\ref{theorem:EffectiveAlgorithm:prop:abseps}), it turns out that 
\(\Kleene (\leq_D,\lambda \vect{X}^\sharp\ldotp\epsilon^{\sharp} \sqcup \Pre^{\sharp}(\vect{X}^\sharp), \vect{\bot_D})\) is an algorithm computing the least fixpoint
$\lfp (\lambda \vect{X}^\sharp\ldotp\alpha(\vectarg{\epsilon}{F}) \sqcup \alpha(\Pre_{\cA}(\gamma(\vect{X}^\sharp))))$.  In particular, 
(\ref{theorem:EffectiveAlgorithm:prop:absdecidable}), (\ref{theorem:EffectiveAlgorithm:prop:abspre}) and (\ref{theorem:EffectiveAlgorithm:prop:abseps}) ensure that the Kleene iterates of 
$\lambda \vect{X}^\sharp\ldotp\epsilon^{\sharp} \sqcup \Pre^{\sharp}(\vect{X}^\sharp)$
starting from $\vect{\bot_D}$
are computable and finitely many and that 
it is decidable when the iterates converge for $\leq_D$, namely, reach the least fixpoint.
Finally, hypothesis~(\ref{theorem:EffectiveAlgorithm:prop:absincl}) ensures the 
decidability of the  $\leq_D$-inclusion check of this least fixpoint 
in $\alpha(\vectarg{L_2}{I})$.  
\end{proof}

It is worth pointing out that, analogously to Theorem~\ref{theorem:backCompleteRight}, 
the above Theorem~\ref{theorem:EffectiveAlgorithm} can be also stated in a symmetric version
for right (rather than left) concatenation. 
 
\subsection{Antichains as a Galois Connection}\label{sec:usingAntiChains}

Let \(\cA_1 = \tuple{Q_1,\delta_1,I_1,F_1,\Sigma}\) and \(\cA_2 = \tuple{Q_2,\delta_2,I_2,F_2,\Sigma}\) be two FAs 
and consider
the state-based left \(\lang{\cA_2}\)-consistent wqo
 \(\mathord{\leqslant_{\cA_2}^l}\) defined by~\eqref{eqn:state-qo}. 
Theorem~\ref{theorem:quasiorderAlgorithm} shows that Algorithm \AlgRegularW decides \(\lang{\cA_1} \subseteq \lang{\cA_2}\) by
computing vectors of finite sets of words. 
Since \(u \leqslant_{\cA_2}^l v \Lra \pre^{\cA_2}_u(F_2) \subseteq \pre^{\cA_2}_v(F_2)\), we can equivalently consider 
the 
set of states \(\pre^{\cA_2}_u(F_2)\in \wp(Q_2)\) rather than
a word $u\in \Sigma^*$. 
This observation suggests to design a version of  Algorithm~\AlgRegularW that 
computes Kleene iterates on the poset
\(\tuple{\AC_{\tuple{\wp(Q_2),\subseteq}},\sqsubseteq}\) of antichains 
of sets of states of the complete lattice $\tuple{\wp(Q_2),\subseteq}$. 
In order to do this, \(\tuple{\AC_{\tuple{\wp(Q_2),\subseteq}},\sqsubseteq}\) is viewed as an abstract domain through the following maps \(\alpha\colon \wp(\Sigma^*) \ra \AC_{\tuple{\wp(Q_2),\subseteq}}\) and \(\gamma\colon \AC_{\tuple{\wp(Q_2),\subseteq}}\ra\wp(\Sigma^*)\). Moreover, we use 
the abstract function \({\Pre}_{\cA_1}^{\cA_2}:(\AC_{\tuple{\wp(Q_2),\subseteq}})^{|Q_1|}\ra (\AC_{\tuple{\wp(Q_2),\subseteq}})^{|Q_1|}\) defined as follows:
\begin{align}
& \alpha(X) \ud \lfloor \{ \pre_u^{\cA_2}(F_2) \in \wp(Q_2) \mid u\in X\} \rfloor \,,\nonumber\\
& \gamma(Y) \ud \{v \in \Sigma^* \mid \exists y \in Y, \, y \subseteq \pre_{v}^{\cA_2}(F_2)\} \,,\label{def-antichain-state-abs}\\
&\Pre_{\cA_1}^{\cA_2}(\tuple{X_q}_{q\in Q_1}) \ud \langle \lfloor \{ \pre_a^{\cA_2}(S)  \in \wp(Q_2) \mid  \exists a\in \Sigma, q'\in Q_1,q'\in\delta_1(q,a) \wedge S \in X_{q'} \} \rfloor \rangle_{q\in Q_1} \nonumber ,
\end{align}

\noindent
where $\minor{X}$ is the unique minor set w.r.t.\ subset inclusion 
of $X\subseteq \wp(Q_2)$.   
Observe that the functions $\alpha$ and ${\Pre}_{\cA_1}^{\cA_2}$ are well-defined because minors of finite subsets of $\wp(Q_2)$ are uniquely defined antichains. 

\begin{lemma}\label{lemma:rhoisgammaalpha}
The following properties hold:
\begin{compactenum}[\upshape(\rm a\upshape)]
\item \(\tuple{\wp(\Sigma^*),\subseteq}\galois{\alpha}{\gamma}\tuple{\AC_{\tuple{\wp(Q_2),\subseteq}},\sqsubseteq}\) is a GC.
\label{lemma:rhoisgammaalpha:GC}
\item \(\gamma \comp \alpha = \rho_{\leqslant^l_{\cA_2}}\).\label{lemma:rhoisgammaalpha:rho}
\item \(\Pre_{\cA_1}^{\cA_2} = {\alpha\comp \Pre_{\cA_1} \comp \gamma }\). \label{lemma:rhoisgammaalpha:pre}
\end{compactenum}
\end{lemma}
\begin{proof}\hfill

\begin{enumerate}[\upshape(\rm a\upshape)]
\item 
Let us first observe that $\alpha$ is well-defined: in fact,
$\alpha(X)$ is an antichain of  $\tuple{\wp(Q_2),\subseteq}$ since it is a minor for the well-quasiorder \(\subseteq\) and, therefore, it is finite.
Then, for all $X\in \wp(\Sigma^*)$ and 
$Y\in \AC_{\tuple{\wp(Q_2),\subseteq}}$, 
it turns out that:
\begin{align*}
\alpha(X) \sqsubseteq Y & \Lra \quad\text{~[by definition of \(\sqsubseteq\)]} \\
\forall z \in \alpha(X), \exists y \in Y, \; y \subseteq z &\Lra \quad\text{~[by definition of  \(\alpha\)]} \\
\forall z \in \lfloor \{ \pre_u^{\cA_2}(F_2) \in \wp(Q_2) \mid u\in X\} \rfloor, \exists y \in Y, \; y \subseteq z &\Lra \quad\text{~[by definition of \(\minor{\cdot}\)]} \\
\forall u \in X, \exists y \in Y, \; y \subseteq \pre^{\cA_2}_u(F_2) &\Lra  \quad\text{~[by definition of  \(\gamma\)]} \\
X \subseteq \gamma(Y) &\enspace . 
\end{align*}

\item For all \(X \in \wp(\Sigma^*)\):
\begin{align*}
  \gamma(\alpha(X)) &=\quad\text{~[by definition of $\alpha,\gamma$]}\\
  \{v \in \Sigma^* \mid \exists u\in \Sigma^*, \pre_{u}^{\cA_2}(F_2) \in \lfloor \{ \pre_w^{\cA_2}(F_2) \mid w\in X\} \rfloor  
  \span \land \pre_{u}^{\cA_2}(F_2) \subseteq \pre_{v}^{\cA_2}(F_2)\} \\
  &=\quad\text{~[by definition of minor]} \\
  \{v \in \Sigma^* \mid \exists u\in X,\, \pre_{u}^{\cA_2}(F_2) \subseteq \pre_{v}^{\cA_2}(F_2)\} &=\quad\text{~[by definition of \(\mathord{\leqslant^l_{\cA_2}}\)]}\\
  \{v \in \Sigma^* \mid \exists u \in X ,\,  u \leqslant^l_{\cA_2} v\}&=\quad\text{~[by definition of\ \(\rho_{\leqslant^l_{\cA_2}}\)]}\\
  \rho_{\leqslant^l_{\cA_2}}(X) &\enspace .
\end{align*}

\item For all \(\vect{X} \in (\AC_{\tuple{\wp(Q_2),\subseteq}})^{|Q_1|}\):
\begin{align*}
\alpha(\Pre_{\cA_1}(\gamma(\vect{X}))) &= \qquad\qquad\qquad\\
\span\specialcell{\hfill \text{~[by definition of \(\Pre_{\cA_1}\)]}}\\
\tuple{\alpha({\textstyle \bigcup_{a \in \Sigma, q\ggoes{a}_{\cA_1} q'}} a\gamma(\vect{X}_{q'}))}_{q \in Q_1} &=\\
\span\specialcell{\hfill \text{~[by definition of \(\alpha\)]}}\\
\langle\lfloor \{ \pre^{\cA_2}_u(F_2)  \mid u \in {\textstyle \bigcup_{a \in \Sigma, q\ggoes{a}_{\cA_1} q'}} a\gamma(\vect{X}_{q'})\rfloor\rangle_{q\in Q_1} &= \\
\span\specialcell{\hfill\text{~[by \(\pre^{\cA_2}_{av} = \pre^{\cA_2}_a\comp \pre^{\cA_2}_v\)]}}\\
\langle\lfloor \{ \pre^{\cA_2}_a(\{ \pre^{\cA_2}_v(F_2) \mid v \in {\textstyle\bigcup_{q\ggoes{a}_{\cA_1} q'}}\gamma(\vect{X}_{q'})\})  \mid a \in \Sigma\}\rfloor\rangle_{q\in Q_1} &=\\
\span\specialcell{\hfill\text{~[by rewriting]}}\\
\langle\lfloor \{ \pre^{\cA_2}_a(S)  \mid a \in \Sigma, q\ggoes{a}_{\cA_1} q', S\in \{ \pre^{\cA_2}_v(F_2) \mid v \in \gamma(\vect{X}_{q'})\}\}\rfloor\rangle_{q\in Q_1} &= \\
\span\specialcell{\hfill\text{~[by \( \minor{\{\pre^{\cA_2}_a(S)\mid S\in Y\}} = \minor{\{\pre^{\cA_2}_a(S)\mid S\in \minor{Y}\}}\)]}}\\
\langle\lfloor \{ \pre^{\cA_2}_a(S)  \mid a \in \Sigma, q\ggoes{a}_{\cA_1} q', S\in \minor{\{ \pre^{\cA_2}_v(F_2) \mid v \in \gamma(\vect{X}_{q'})\}}\}\rfloor\rangle_{q\in Q_1} & = \\
\span\specialcell{\hfill\text{~[by definition of \(\alpha\)]}}\\
\langle\lfloor \{ \pre^{\cA_2}_a(S)  \mid a \in \Sigma, q\ggoes{a}_{\cA_1} q', S\in \alpha(\gamma(\vect{X}_{q'}))\rfloor\rangle_{q\in Q_1} &= \\
\span\specialcell{\hfill\text{~[since \(\vect{X} \in \alpha\), \(\alpha(\gamma(\vect{X}_{q'})) = \vect{X}_{q'}\)]}}\\
\langle \lfloor \{ \pre^{\cA_2}_a(S)  \mid a\in\Sigma, q\ggoes{a}_{\cA_1}q', S \in \vect{X}_{q'} \} \rfloor   \rangle_{q\in Q_1} &=\\ 
\span\specialcell{\hfill\text{~[by definition of ${\Pre}_{\cA_1}^{\cA_2}$]}}\\
{\Pre}_{\cA_1}^{\cA_2}(\vect{X}) & \enspace .\tag*{\qedhere} 
\end{align*}%
\end{enumerate}%
\end{proof}

Thus, by Lemma~\ref{lemma:LAconsistent} and 
Lemma~\ref{lemma:rhoisgammaalpha}, it turns out that the GC \(\tuple{\wp(\Sigma^*),\subseteq}\galois{\alpha}{\gamma}\tuple{\AC_{\tuple{\wp(Q_2),\subseteq}},\sqsubseteq}\) and the abstract 
function \(\Pre_{\cA_1}^{\cA_2}\) satisfy the hypotheses (\ref{theorem:EffectiveAlgorithm:prop:rho})-(\ref{theorem:EffectiveAlgorithm:prop:abseps}) of Theorem~\ref{theorem:EffectiveAlgorithm}.
In order to obtain an algorithm for deciding \(\lang{\cA_1} \subseteq \lang{\cA_2}\) it remains to show that the hypothesis~(\ref{theorem:EffectiveAlgorithm:prop:absincl}) of Theorem~\ref{theorem:EffectiveAlgorithm} holds, i.e., 
there is an algorithm to decide whether \(\vect{Y} \sqsubseteq \alpha(\vectarg{L_2}{I_2})\) for every \(\vect{Y} \in \alpha(\wp(\Sigma^*))^{|Q_1|}\).

Notice that the Kleene iterates of  \(\lambda \vect{X}^\sharp\ldotp\alpha(\vectarg{\epsilon}{F_1}) \sqcup \Pre_{\cA_1}^{\cA_2}(\vect{X}^\sharp)\) of Theorem~\ref{theorem:EffectiveAlgorithm} are vectors of antichains in \(\tuple{\AC_{\tuple{\wp(Q_2),\subseteq}},\sqsubseteq}\), where 
each component is indexed by some \(q\in Q_1\) and represents, 
through its minor, a set of sets of states that are predecessors of \(F_2\) in \(\cA_2\) through a word $u$ generated by \(\cA_1\) from that state \(q\), i.e., \(\pre_u^{\cA_2}(F_2)\) with \(u \in W^{\cA_1}_{q,F_1}\).
Since \(\epsilon \in W_{q,F_1}^{\cA_1}\) for all \(q \in F_1\) and \(\pre_\epsilon^{\cA_2}(F_2) = F_2\), the first 
iteration of $\Kleene$ gives the vector \(\alpha(\vectarg{\epsilon}{F_1}) = \tuple{\nullable{q\inqm F_1}{F_2}{\varnothing}}_{q\in Q_1}\).
Let us also observe that by taking the minor of each vector component, 
we are considering smaller sets which still preserve the relation \(\sqsubseteq\) since 
the following equivalences hold: \(A \sqsubseteq B \Lra \minor{A} \sqsubseteq B \Lra A \sqsubseteq \minor{B} \Lra \minor{A} \sqsubseteq \minor{B}\).
Let \(\tuple{Y_q}_{q\in Q_1}\) be the output of 
$\Kleene (\sqsubseteq, \lambda \vect{X}^\sharp\ldotp\alpha(\vectarg{\epsilon}{F_1}) \sqcup \Pre_{\cA_1}^{\cA_2}(\vect{X}^\sharp),\vect{\varnothing})$. 
Hence, we have that, for each component $q\in Q_1$, \(Y_q = \minor{\{\pre_u^{\cA_2}(F_2)\mid u \in W_{q,F_1}^{\cA_1}\}}\) holds.
Whenever the inclusion \(\lang{\cA_1} \subseteq \lang{\cA_2}\) holds, all the sets of states in \(Y_q\) for some initial state \(q \in I_1\) are predecessors of \(F_2\) in \(\cA_2\) through words in \(\lang{\cA_2}\), so that for
each $q\in I_1$ and $S\in Y_q$, $S \cap I_2\neq\varnothing$ must hold. 
As a result, the following  state-based
algorithm \AlgRegularA  ($\mathtt{S}$ stands for state) decides the
inclusion \(\lang{\cA_1} \subseteq \lang{\cA_2}\) by computing on the abstract domain of antichains \(\tuple{\AC_{\tuple{\wp(Q_2),\subseteq}},\sqsubseteq}\). 

\begin{figure}[H]
\RemoveAlgoNumber
\begin{algorithm}[H]
\SetAlgorithmName{\AlgRegularA}{}

\caption{State-based algorithm for {\(\lang{\cA_1} \subseteq \lang{\cA_2}\)}}\label{alg:RegIncA}

\KwData{FAs \(\cA_1=\tuple{Q_1,\delta_1,I_1,F_1,\Sigma}\) and \(\cA_2=\tuple{Q_2,\delta_2,I_2,F_2,\Sigma}\).}
\medskip
\(\tuple{Y_q}_{q\in Q_1} := \Kleene (\sqsubseteq, \lambda \vect{X}^\sharp\ldotp\alpha(\vectarg{\epsilon}{F_1}) \sqcup \Pre_{\cA_1}^{\cA_2}(\vect{X}^\sharp),\vect{\varnothing})\)\;

\ForAll{\(q \in I_1\)}{
	\ForAll{\(S \in Y_q\)} {
		\lIf{\(S \cap I_2 = \varnothing\)}{\Return \textit{false}}
	}
}
\Return \textit{true}\;
\end{algorithm}
\end{figure}

\begin{theorem}\label{theorem:statesQuasiorderAlgorithm}
	The algorithm \AlgRegularA decides the inclusion problem \(\lang{\cA_1} \subseteq \lang{\cA_2}\).
\end{theorem}
\begin{proof}
We show that all the hypotheses~(i)-(v) of Theorem~\ref{theorem:EffectiveAlgorithm} are satisfied for the abstract domain \(\tuple{D,\leq_D}=\tuple{\AC_{\tuple{\wp(Q_2),\subseteq}},\sqsubseteq}\) as defined by the GC of Lemma~\ref{lemma:rhoisgammaalpha}. 

\begin{compactenum}[\upshape(i\upshape)]
\item Since \(\rho_{\leq^l_{\cA_2}}(X) = \gamma(\alpha(X))\), it follows from Lemmata~\ref{lemma:properties} and~\ref{lemma:LAconsistent} that \(L_2\in \gamma(D)\).
Moreover, by Lemma~\ref{lemma:properties}~(\ref{lemma:properties:bw}) with \(\rho_{\leq^l_{\cA_2}} = \gamma\alpha\), we have that
for all \(a\in\Sigma\), \(X\in\wp(\Sigma^*)\), \(\gamma(\alpha(a X)) = \gamma(\alpha(a\gamma(\alpha(X))))\).

\item \( \tuple{\AC_{\tuple{\wp(Q_2),\subseteq}},\sqsubseteq,\sqcup,\varnothing}\) is an effective domain 
because 
$Q_2$ is finite.
\item  By Lemma~\ref{lemma:rhoisgammaalpha}~(\ref{lemma:rhoisgammaalpha:pre}) we have that
\(\alpha(\Pre_{\cA_1}(\gamma(\vect{X}^\sharp))) = {\Pre}_{\cA_1}^{\cA_2}(\vect{X}^\sharp)\) for all \(\vect{X}^\sharp\in \alpha(\wp(\Sigma^*))^{|Q_1|}\), and ${\Pre}_{\cA_1}^{\cA_2}$ is computable.
\item \(\alpha(\{\epsilon\}) = \{F_2\}\) and \(\alpha({\varnothing})=\varnothing\), hence \(\alpha(\vectarg{\epsilon}{F_1})\) is trivial to compute. \label{prop:alphaepsilon}

\item Since \(\alpha(\vectarg{L_2}{I_1})=\tuple{\alpha(\nullable{q\inqm I_1}{L_2}{\Sigma^*})}_{q\in Q_1}\), for all $\vect{Y}\in\alpha(\wp(\Sigma^*))^{|Q_1|}$   the relation \(\tuple{Y_q}_{q\in Q_1} \sqsubseteq \alpha(\vectarg{L_2}{I_1})\) trivially holds for all components \(q \notin I_1\), since $\alpha(\Sigma^*)$ is the greatest antichain. 
For the components $q\in I_1$, it suffices to show that
\(Y_q \sqsubseteq \alpha(L_2) \Lra \forall S \in Y_q, \; S \cap I_2 \neq \varnothing\), which is the check performed by lines 2-5 of algorithm \AlgRegularA:
\begin{align*}
Y_q \sqsubseteq \alpha(L_2) & \Lra \quad \text{~[because \(Y_q = \alpha(U)\) for some \(U \in \wp(\Sigma^*)\)]} \\
\alpha(U) \sqsubseteq \alpha(L_2) & \Lra \quad \text{~[by GC]} \\
U \subseteq \gamma(\alpha(L_2)) & \Lra \quad \text{~[by (i), $\gamma(\alpha(L_2))=L_2$]} \\
U \subseteq L_2 & \Lra \quad \text{~[by definition of \(\pre_u^{\cA_2}\)]} \\
\forall u \in U, \pre_u^{\cA_2}(F_2) \cap I_2 \neq \varnothing & \Lra \quad \text{~[because \(Y_q=\alpha(U)= \lfloor \{ \pre_u^{\cA_2}(F_2) \mid u\in U\} \rfloor \)]} \\
\forall S \in Y_q, S \cap I_2 \neq \varnothing &\enspace .
\end{align*}

\end{compactenum}
Thus, by Theorem~\ref{theorem:EffectiveAlgorithm}, the algorithm \AlgRegularA solves the inclusion problem \(\lang{\cA_1} \subseteq \lang{\cA_2}\). %
\end{proof}

\subsection{Relationship to the Antichain Algorithm}%
\label{sub:relationship_to_the_antichain_algorithm}
De Wulf et al.~\shortcite{DBLP:conf/cav/WulfDHR06} introduced two so-called antichain algorithms, called 
forward and backward, for deciding the universality of the language accepted by a FA, i.e., whether the language is $\Sigma^*$ or not.
Then, they extended the backward algorithm in order to decide inclusion of languages accepted by FAs.
In what follows we show that the above algorithm \AlgRegularA is equivalent to the corresponding extension of the forward antichain algorithm and, therefore, dual to the backward antichain algorithm for language inclusion put forward by De Wulf et al.~\shortcite[Theorem 6]{DBLP:conf/cav/WulfDHR06}.
In order to do this, we first define the poset of antichains in which the forward antichain algorithm computes its fixpoint.
Then, we give a formal definition of the forward antichain algorithm for deciding language inclusion and show that this algorithm coincides with \AlgRegularA when applied to the reverse automata.
Since language inclusion between the languages generated by two FAs holds if{}f inclusion holds between the languages generated by their reverse FAs, this entails that our algorithm \AlgRegularA is equivalent to the forward antichain algorithm.

Consider a language inclusion problem 
\(\lang{\cA_1} \subseteq \lang{\cA_2}\) where
\(\cA_1=\tuple{Q_1,\delta_1,I_1,F_1,\Sigma}\) and \(\cA_2=\tuple{Q_2,\delta_2,I_2,F_2,\Sigma}\).
Let us consider the following poset of antichains 
\( \tuple{\AC_{\tuple{\wp(Q_2),\subseteq}},\wsqsubseteq} \) where
\[X \wsqsubseteq Y \udr \forall y \in Y, \exists x \in X, \; x \subseteq y\enspace \]
and notice that \(\wsqsubseteq\) coincides with the reverse 
\(\sqsubseteq^{-1}\) of the relation defined by \eqref{ordering-definition}. 
As observed by De Wulf et al.~\shortcite[Lemma~1]{DBLP:conf/cav/WulfDHR06}, it turns out that \( \tuple{\AC_{\tuple{\wp(Q_2),\subseteq}},\wsqsubseteq, \wsqcup, \wsqcap, \{\varnothing\}, \varnothing} \) is a finite lattice, where \(\wsqcup\) and \(\wsqcap\) denote, resp., lub and glb, and $\{\varnothing\}$ and $\varnothing$ are, resp., the least and greatest elements. 
This lattice \( \tuple{\AC_{\tuple{\wp(Q_2),\subseteq}},\wsqsubseteq} \) is the domain in which the forward antichain algorithm computes on for deciding language universality \cite[Theorem~3]{DBLP:conf/cav/WulfDHR06}.
The following result extends this forward algorithm in order to decide language inclusion.

\begin{theorem}[\textbf{{\cite[Theorems~3 and 6]{DBLP:conf/cav/WulfDHR06}}}] \label{theorem:antichainpaper}
Let
\begin{align*}
\vect{\fp} \ud \textstyle{\wbigsqcup}\{\vect{X} \in (\AC_{\tuple{\wp(Q_2),\subseteq}})^{|Q_1|} \mid \vect{X} = \Post_{\cA_1}^{\cA_2}(\vect{X})\;\wsqcap\; \tuple{\nullable{q \inqm I_1}{\{I_2\}}{\varnothing}}_{q\in Q_1}\}
\end{align*}
where \(\Post_{\cA_1}^{\cA_2}(\tuple{X_q}_{q\in Q_1}) \ud  \langle \minor{\{\post_a^{\cA_2}(S) \in \wp(Q_2) \mid \exists a \in \Sigma, q'\in Q_1,\, 
q\in\delta_1(q',a) \wedge S \in X_{q'}  \}} \rangle_{q \in Q_1}\).
Then, \(\lang{\cA_1} \nsubseteq \lang{\cA_2}\) if and only if there exists \(q \in F_1\) such that \(\vect{\fp}_q \,\wsqsubseteq\, \{F_2^c\} \).
\end{theorem}
\begin{proof}%
Let us first introduce some notation to describe the forward antichain algorithm by De Wulf et. al~\shortcite{DBLP:conf/cav/WulfDHR06} which decides \(\lang{\cA_1} \subseteq \lang{\cA_2}\).
Let us consider the poset
\(\tuple{Q_1\times \wp(Q_2),\subseteq_\times}\) where $(q_1,S_1) \subseteq_\times
(q_2,S_2) \udrshort q_1=q_2 \wedge S_1 \subseteq S_2$. Then, 
let
\(\tuple{\AC_{\tuple{Q_1\times \wp(Q_2),\subseteq_\times}},\wsqsubseteq_\times,\wsqcup_\times, \wsqcap_\times}\) be the lattice of antichains over \(\tuple{Q_1\times \wp(Q_2),\subseteq_\times}\) where:
\begin{align*}
X \wsqsubseteq_\times Y & \udrshort \forall (q,T) \in Y, \exists (q,S) \in X , S \subseteq T & &\\
X \wsqcup_\times Y & \ud \textstyle{\min_{\times}}(\{(q,S \cup T) \mid (q,S) \in X,\, (q,T) \in Y\}) \\
X \wsqcap_\times Y & \ud \textstyle{\min_{\times}}(\{(q,S) \mid (q,S) \in X \cup Y \}) \\
\text{with}\quad \textstyle{\min_{\times}}(X) &\ud \{(q,S) \in X \mid \forall (q',S') \in X, q=q' \Ra S' \not\subset S\} \enspace .
\end{align*}
Also, let $\Post: \AC_{\tuple{Q_1\times \wp(Q_2),\subseteq_\times}} \ra \AC_{\tuple{Q_1\times \wp(Q_2),\subseteq_\times}}$ be defined as follows:  
\[\Post(X) \ud \textstyle{\min_{\times}}(\{ (q,\post_a^{\cA_2}(S)) \in Q_1\times 
\wp(Q_2) \mid \exists a \in \Sigma, q\in Q_1,(q',S) \in X , q' \ggoes{a}_{\cA_1} q\}) \enspace .\]
Then, the dual of the backward antichain algorithm in \cite[Theorem~6]{DBLP:conf/cav/WulfDHR06} states that \(\lang{\cA_1} \nsubseteq \lang{\cA_2}\) if{}f there exists \(q \in F_1\) such that \(\fp \mathrel{\wsqsubseteq_\times} \{(q,F_2^c)\}\) where
\[\fp = {\textstyle\wbigsqcup_\times}\{X \in \AC_{\tuple{Q_1\times \wp(Q_2),\subseteq_\times}} \mid X = \Post(X)\;\wsqcap_\times\; (I_1 \times \{I_2\})\}\enspace .\]

\noindent
We observe that for some \(X\in\AC_{\tuple{Q_1\times \wp(Q_2),\subseteq_\times}}\), a pair \((q,S) \in Q_1\times \wp(Q_2)\) such that $(q,S)\in X$ is used by 
\cite[Theorem~6]{DBLP:conf/cav/WulfDHR06} simply as
a way to associate states $q$ of \(\cA_1\) with sets $S$ of states  of \(\cA_2\).
In fact, an antichain 
\(X\in\AC_{\tuple{Q_1\times \wp(Q_2),\subseteq_\times}}\)
can be equivalently formalized 
by a vector \(\tuple{\{S \in \wp(Q_2) \mid (q,S)\in X\}}_{q \in Q_1}
\in (\AC_{\tuple{\wp(Q_2),\subseteq}})^{|Q_1|}\) whose components are indexed by states \(q\in Q_1\) and are antichains of set of states in $\AC_{\tuple{\wp(Q_2),\subseteq}}$. 
Correspondingly, we consider 
the lattice \(\tuple{\AC_{\tuple{\wp(Q_2),\subseteq}},\wsqsubseteq}\), where for 
all $X,Y\in \AC_{\tuple{\wp(Q_2),\subseteq}}$:
\begin{align*}
X \wsqsubseteq Y &\udrshort \forall T \in Y, \exists S \in X , S \subseteq T\\
X \wsqcup Y & \ud \textstyle{\min}(\{S \cup T \in \wp(Q_2) \mid S \in X, T \in Y\})\\
X \wsqcap Y & \ud \textstyle{\min}(\{S \in \wp(Q_2) \mid S \in X \cup Y\}) \\
\text{with}\quad \textstyle{\min}(X) &\ud \{S \in X \mid \forall S' \in X, S' \not\subset S\} \enspace .
\end{align*}

\noindent
Then, these definitions allow us to replace \(\Post\) by an equivalent function 
\[
\Post_{\cA_1}^{\cA_2}: (\AC_{\tuple{\wp(Q_2),\subseteq}})^{|Q_1|} \ra (\AC_{\tuple{\wp(Q_2),\subseteq}})^{|Q_1|}
\] that transforms vectors of antichains as follows:
\[\Post_{\cA_1}^{\cA_2}(\tuple{X_q}_{q\in Q_1}) \ud \tuple{\textstyle{\min}(\{\post_a^{\cA_2}(S) \in \wp(Q_2) \mid \exists a \in \Sigma, q'\in Q_1, S \in X_{q'} , q' \ggoes{a}_{\cA_1} q \})}_{q \in Q_1}\enspace .\]
In turn, the above \(\fp\in \AC_{\tuple{Q_1\times \wp(Q_2),\subseteq_\times}}\) is replaced by the 
following equivalent vector:
\[\vect{\fp} \ud \textstyle{\wbigsqcup}\{\vect{X}\in (\AC_{\tuple{\wp(Q_2),\subseteq}})^{|Q_1|} \mid \vect{X}=
 \Post_{\cA_1}^{\cA_2}(\vect{X})\;\wsqcap\; \tuple{\nullable{q \inqm I_1}{\{I_2\}}{\varnothing}}_{q\in Q_1}\} \enspace .\]
Finally, the condition 
\(\exists q \in F_1 , \fp \mathrel{\wsqsubseteq_\times} \{(q,F_2^c)\}\) is equivalent 
to \(\exists q \in F_1 ,  \vect{\fp}_q \mathrel{\wsqsubseteq} \{F_2^c\} \).
\end{proof}

Let us recall that \(\cA^R\) denotes the reverse automaton of \(\cA\), where arrows are flipped and the initial/final states become final/initial.
Note that language inclusion can be decided by considering the reverse automata since \(\lang{\cA_1} \subseteq \lang{\cA_2} \Lra \lang{\cA_1^R} \subseteq \lang{\cA_2^R}\) holds.
Furthermore, let us observe that \(\Post_{\cA_1}^{\cA_2} = \Pre_{\cA_1^R}^{\cA_2^R}\).
We therefore obtain the following consequence of Theorem~\ref{theorem:antichainpaper}. 
\begin{corollary}\label{theorem:antichainpaperReverse}
Let
\begin{align*}
\vect{\fp} \ud \textstyle{\wbigsqcup}\{\vect{X} \in (\AC_{\tuple{\wp(Q_2),\subseteq}})^{|Q_1|} \mid \vect{X} = \Pre_{\cA_1}^{\cA_2}(\vect{X})\;\wsqcap\; \tuple{\nullable{q \inqm F_1}{\{F_2\}}{\varnothing}}_{q\in Q_1}\} \enspace .
\end{align*}
Then, \(\lang{\cA_1} \nsubseteq \lang{\cA_2}\) if{}f \(\exists q \in I_1 , \vect{\fp}_q \,\wsqsubseteq\, \{I_2^c\} \).
\end{corollary}

Since \(\wsqsubseteq = \mathord{\sqsubseteq^{-1}}\), we have that \(\wsqcap = \sqcup\), \(\wsqcup = \sqcap\) and the greatest element $\varnothing$ for $\wsqsubseteq$ is the least element for $\mathord{\sqsubseteq}$.
Moreover, by~\eqref{def-antichain-state-abs}, $\alpha(\vectarg{\epsilon}{F_1}) = \tuple{\nullable{q \inqm F_1}{\{F_2\}}{\varnothing}}_{q\in Q_1}$.
Therefore, we can rewrite the vector 
$\vect{\fp}$ of 
Corollary~\ref{theorem:antichainpaperReverse} as 
\[
\vect{\fp} = {\textstyle\bigsqcap}\{\vect{X} \in (\AC_{\tuple{\wp(Q_2),\subseteq}})^{|Q_1|} \mid \vect{X} = \Pre_{\cA_1}^{\cA_2}(\vect{X})\;\sqcup\; \alpha(\vectarg{\epsilon}{F_1})\}
\] 
which is precisely the least fixpoint in $\tuple{(\AC_{\tuple{\wp(Q_2),\subseteq}})^{|Q_1|}, \sqsubseteq}$ of $\Pre_{\cA_1}^{\cA_2}$ above
$\alpha(\vectarg{\epsilon}{F_1})$. 
Hence, it turns out that the Kleene iterates of the least fixpoint computation 
that converge to \(\vect{\fp}\) exactly coincide with the iterates computed by the $\Kleene$ procedure of the state-based algorithm 
\AlgRegularA.
In particular, if  \(\vect{Y}\) is the output vector of 
$\Kleene (\sqsubseteq, \lambda \vect{X}\ldotp\alpha(\vectarg{\epsilon}{F_1}) \sqcup \Pre_{\cA_1}^{\cA_2}(\vect{X}),\vect{\varnothing})$
at line~1 of  
\AlgRegularA then  \(\vect{Y} = \vect{\fp}\).
Furthermore, \(\exists q\in I_1, \vect{\fp}_q \:\wsqsubseteq\: \{I_2^c\} \Lra \exists q\in I_1, \exists S \in \vect{\fp}_q, \; S \cap I_2 = \varnothing\).
Summing up, the \(\sqsubseteq\)-lfp algorithm \AlgRegularA exactly coincides with the \(\wsqsubseteq\)-gfp antichain algorithm as given by Corollary~\ref{theorem:antichainpaperReverse}.

We can easily derive an antichain algorithm which is perfectly equivalent to \AlgRegularA by considering the antichain 
lattice \(\tuple{\AC_{\tuple{\wp(Q_2),\supseteq}},\sqsubseteq}\) for the dual lattice $\tuple{\wp(Q_2),\supseteq}$ 
and by replacing the functions 
\(\alpha\), \(\gamma\) and \(\Pre_{\cA_1}^{\cA_2}\) of Lemma~\ref{lemma:rhoisgammaalpha}, resp., with the following dual versions:
\begin{align*}
\alpha^c(X) &\ud \lfloor \{ \cpre_u^{\cA_2}(F_2^c) \in \wp(Q_2) \mid u\in X\} \rfloor \, ,&
\quad \gamma^c(Y) &\ud \{v \in \Sigma^* \mid \exists y \in Y , y \supseteq \cpre_{v}^{\cA_2}(F^c_2) \}\, , \\
{\CPre}_{\cA_1}^{\cA_2}(\tuple{X_q}_{q\in Q_1})  \ud 
\langle \lfloor \{ \cpre_a^{\cA_2}(S) \in \wp(Q_2) \mid  \exists a\in \Sigma, q'\in Q_1, \, q'\in\delta_1(q,a) \wedge  S \in X_{q'} \} \rfloor \rangle_{q\in Q_1} \,. \span \span \span
\end{align*}
where \(\cpre_u^{\cA_2}(S) \ud (\pre_u^{\cA_2}(S^c))^c\) for $u\in \Sigma^*$.
When using these functions, the corresponding algorithm computes 
on the abstract domain \(\tuple{\AC_{\tuple{\wp(Q_2),\supseteq}},\sqsubseteq}\) and
it turns out that \(\lang{\cA_1} \subseteq \lang{\cA_2}\) if{}f \(\Kleene(\sqsubseteq,\lambda \vect{X}^\sharp \ldotp \alpha^c(\vectarg{\epsilon}{F_1}) \sqcup \CPre_{\cA_1}^{\cA_2}(\vect{X}^\sharp),\vect{\varnothing}) \sqsubseteq \alpha^c(\vectarg{L_2}{I_1})\).
This language inclusion algorithm coincides with the backward antichain algorithm defined by De Wulf et al.~\shortcite[Theorem~6]{DBLP:conf/cav/WulfDHR06} since both compute on the same lattice, \(\minor{X}\) corresponds to the maximal (w.r.t.\ set inclusion) elements of \(X\), \(\alpha^c(\{\epsilon\}) = \{F_2^c\}\) and for all \(X \in \alpha^c(\wp(\Sigma^*))\), we have that \(X \sqsubseteq \alpha^c(L_2) \Lra \forall S \in X, \; I_2 \nsubseteq S\).

We have thus shown that the two forward/backward antichain algorithms introduced by De Wulf et al.~\shortcite{DBLP:conf/cav/WulfDHR06} can be systematically derived by instantiating our framework.
The original antichain algorithms were later improved by Abdulla et al.~\shortcite{Abdulla2010} and, subsequently, by Bonchi and Pous~\shortcite{DBLP:conf/popl/BonchiP13}. Among their improvements, they showed how to exploit a precomputed binary relation between pairs of states of the input automata such that language inclusion holds for all the pairs in the relation.
When that binary relation is a simulation relation, our framework allows to partially match their results by using the simulation-based quasiorder \(\preceq^{r}_{\cA}\) defined in Section~\ref{subsec:simulation}.
However, this relation \(\preceq^{r}_{\cA}\) does not consider pairs of states \(Q_2 \times Q_2\) whereas the aforementioned algorithms do.

\section{Inclusion for Context Free Languages}%
\label{sec:context_free_languages}

A \emph{context-free grammar} (CFG) is a tuple \(\cG=\tuple{\cV,\Sigma, P}\) where \(\cV=\{X_0,\ldots,X_n\}\) is a finite set of variables including a start symbol \(X_0\), \(\Sigma\) is a finite alphabet of terminals and \(P\) is a finite set of productions  \(X_i\ra \beta\) where \(\beta\in (\cV\cup\Sigma)^*\). 
We assume, for simplicity and without loss of generality, that CFGs are in Chomsky Normal Form (CNF), that is, every production \(X_i \ra \beta\in P\) is such that \(\beta\in (\cV\times \cV) \cup \Sigma \cup \{\epsilon\}\) and if $\beta=\epsilon$ then \(i=0\) \cite{DBLP:journals/iandc/Chomsky59a}.
We also assume that for all \(X_i \in \cV\) there exists a production \(X_i \ra \beta\in P\), otherwise \(X_i\) can be safely removed from \(\cV\).
Given two strings \(w, w' \in (\cV \cup \Sigma)^*\) we write \(w \ra w'\) if{}f there exists \(u, v \in (\cV \cup \Sigma)^*\) and \(X \to \beta \in P\) such that \(w = u X v\) and \(w' = u\beta v\).
We denote by \(\ra^*\) the reflexive-transitive closure of \(\ra\). 
The language generated by a \(\cG\) is \(\lang{\cG} \ud \{w \in \Sigma^* \mid X_0 \ra^* w\}\).

\subsection{Extending the Framework to CFGs}\label{sec:CFGs}
Similarly to the case of automata,  a CFG \(\cG = (\cV,\Sigma,P)\) in CNF induces the following set of equations:
\begin{equation*}
\Eqn(\cG) \ud \{X_i = {\textstyle \bigcup_{X_i \to \beta_j \in P}} \beta_j \mid i \in [0,n]\} \enspace .
\end{equation*}
Given a subset of variables \(S \subseteq \cV\) of a grammar, the set of words generated from some variable in \(S\) is defined as
\[W_{S}^{\cG} \ud \{w \in \Sigma^* \mid \exists X \in S, \; X \ra^* w\} \enspace .\]
When \(S = \{X\}\) we slightly abuse the notation and write \(W_{X}^{\cG}\). 
Also, we drop the superscript \(\cG\) when the grammar is clear from the context.
The language generated by \(\cG\) is therefore \(\lang{\cG} = W^{\cG}_{X_0}\).

We define the vector \(\vect{b} \in \wp(\Sigma^*)^{|\cV|}\) and the function \(\Fn_{\cG}: \wp(\Sigma^*)^{|\cV|}\to \wp(\Sigma^*)^{|\cV|}\), which are used to formalize the fixpoint equations in \(\Eqn(\cG)\), as follows:
\begin{align*}
\vect{b} & \ud\tuple{b_i}_{i\in[0,n]} \in \wp(\Sigma^*)^{|\cV|} &&\text{ where } b_i \ud \{ \beta \mid X_i\ra \beta\in P,\:\beta\in \Sigma\cup \{ \epsilon \}\} \enspace , \\
\Fn_{\cG }(\tuple{X_i}_{i\in[0,n]}) & \ud \tuple{\beta_1^{(i)}\cup\ldots\cup\beta_{k_i}^{(i)}}_{i\in[0,n]} &&\text{ where } \beta_j^{(i)}\in\cV^2 \text{ and } X_i\ra\beta_j^{(i)}\in P \enspace .
\end{align*}

\noindent
Notice that \(\lambda \vect{X}\ldotp \vect{b}\mathrel{\cup} \Fn_{\cG}(\vect{X})\) is a well-defined monotonic function in  \(\wp(\Sigma^*)^{|\cV|}\ra \wp(\Sigma^*)^{|\cV|}\), which therefore has the least fixpoint
\(\tuple{Y_i}_{i\in[0,n]} = \lfp (\lambda \vect{X}\ldotp \vect{b}\cup \Fn_{\cG}(\vect{X}))\). It is known \cite{ginsburg} that the language \(\lang{\cG}\) accepted by \(\cG\) is such that \(\lang{\cG} = Y_{0}\).

\begin{example}\label{example:cfg}
Consider the CFG  \(\cG = \tuple{\{X_0, X_1\}, \{a,b\}, \{X_0\ra X_0X_1 \mid X_1X_0 \mid b,\: X_1 \ra a\}}\) in CNF. 
The corresponding equation system is
\[\Eqn(\cG) = \begin{cases}
    X_0 = X_0X_1 \cup X_1X_0 \cup \{b\}\\
    X_1 =\{a\}
  \end{cases}\]
so that
\begin{equation*}
  \left( \begin{array}{c}
     W_{X_0} \\ W_{X_1}
  \end{array} \right)=
  \lfp\biggl(\lambda \left( \begin{array}{c}
    X_0 \\ X_1
  \end{array} \right) .
  \left(\begin{array}{c}
      X_0X_1 \cup X_1X_0 \cup \{b\} \\
      \{a\}
    \end{array}\right)\biggr) = \left( \begin{array}{c}
     a^*ba^* \\ a
  \end{array} \right) \enspace .
\end{equation*}

\noindent
Moreover, we have that \(\vect{b} \in \wp(\Sigma^*)^2\) and \(\Fn_{\cG }:\wp(\Sigma^*)^2 \ra \wp(\Sigma^*)^2\) are given by
\begin{align*}
\vect{b} & =\tuple{\{b\},\{a\}} & \Fn_{\cG }(\tuple{X_0,X_1}) &=\tuple{X_0X_1 \cup X_1X_0, \varnothing} \enspace. \tag*{\eox}
\end{align*}
\end{example}

It turns out that
\[\lang{\cG} \subseteq L_2 \:\Lra\:  
\lfp (\lambda\vect{X}\ldotp \vect{b}\cup \Fn_{\cG}(\vect{X})) \subseteq \vectarg{L_2}{X_0} \] 
where \(\vectarg{L_2}{X_0} \ud \tuple{\nullable{i\eqqm 0}{L_2}{\Sigma^*}}_{i\in[0,n]}\).

\begin{theorem}\label{theorem:rhoCFG}
Let \(\cG=(\cV,\Sigma,P)\) be a CFG in CNF. 
If \(\rho \in \uco(\wp(\Sigma^*))\) is backward complete for both \(\lambda X. Xa\) and \(\lambda X. aX\), for all \(a \in \Sigma\) then \(\rho\) is backward complete for \(\lambda\vect{X}\ldotp \vect{b}\cup \Fn_{\cG}(\vect{X})\). 
\end{theorem}
\begin{proof}%
Let us first show that backward completeness for left and right concatenation can be extended from letter to words.
We give the proof for left concatenation, the right case is symmetric.
We prove that \(\rho(w X) = \rho(w \rho(X))\) for every \(w\in\Sigma^*\).
We proceed by induction on \(|w|\geq 0\).
The base case $|w|=0$ if{}f $w=\epsilon$ is trivial because \(\rho\) is idempotent.
For the inductive case \(|w| > 0\) let \(w = a u\) for some 
\(u\in\Sigma^*\) and \(a\in \Sigma\), so that:
\begin{align*}
	\rho(a u X) &= \quad\text{~[by backward completeness for \(\lambda X\ldotp a X\)]}\\
	\rho(a \rho(u X)) &= \quad\text{~[by inductive hypothesis]}\\
	\rho(a \rho(u \rho(X))) &= \quad\text{~[by backward completeness for \(\lambda X\ldotp a X\)]}\\
	\rho(a u \rho(X)) &\enspace .
\end{align*}
Next we turn to the binary concatenation case, i.e., we prove that \(\rho(Y Z) = \rho(\rho(Y)\rho(Z))\) for all \(Y, Z \in \wp(\Sigma^*)\):
\begin{align*}
	\rho(\rho(Y)\rho(Z)) &=\quad \text{~[by definition of concatenation]}\\
	\rho(\textstyle{\bigcup_{u\in\rho(Y)}} u \rho(Z)) &=\quad \text{~[by \eqref{equation:lubAndGlb}]}\\
	\rho(\textstyle{\bigcup_{u\in\rho(Y)}} \rho(u \rho(Z)) ) &=\quad \text{~[by backward completeness of \(\lambda X\ldotp u X\)]}\\
	\rho(\textstyle{\bigcup_{u\in\rho(Y)}} \rho(u Z))
	&=\quad\text{~[by~\eqref{equation:lubAndGlb}]}\\
	\rho(\textstyle{\bigcup_{u\in\rho(Y)}} u Z) &=\quad\text{~[by definition of concatenation]}\\
	\rho(\rho(Y) Z) &=\quad\text{~[by definition of concatenation]}\\
	\rho(\textstyle{\bigcup_{v\in Z}} \rho(Y) v)&=\quad
	\text{~[by~\eqref{equation:lubAndGlb}]}\\
	\rho(\textstyle{\bigcup_{v\in Z}} \rho(\rho(Y) v))&=\quad
	\text{~[by backward completeness of \(\lambda X\ldotp X v\)]}\\
	\rho(\textstyle{\bigcup_{v\in Z}} \rho(Y v))&=\quad
	\text{~[by~\eqref{equation:lubAndGlb}]}\\
	\rho(\textstyle{\bigcup_{v\in Z}} Y v)&=\quad\text{~[by definition of concatenation]}\\
	\rho(Y Z) & \enspace .
\end{align*}
Then, the proof follows the same lines of the proof of Theorem~\ref{theorem:backComplete}.
Indeed, it follows from the definition of \(\Fn_{\cG}(\tuple{X_i}_{i\in[0,n]})\)
that: \begin{align*}
	\rho({\textstyle\bigcup_{j=1}^{k_i}}\beta^{(i)}_j)	&=\quad \text{~[by definition of \(\beta^{(i)}_j\)]}\\
\rho({\textstyle\bigcup_{j=1}^{k_i}}X^{(i)}_j  Y^{(i)}_j) &=\quad
	\text{~[by~\eqref{equation:lubAndGlb}]}\\
\rho({\textstyle\bigcup_{j=1}^{k_i}}\rho(X^{(i)}_j  Y^{(i)}_j)) &=\quad
  \text{~[by backward completeness of \(\rho\) for binary concatenation]}\\
\rho({\textstyle\bigcup_{j=1}^{k_i}}\rho( \rho(X^{(i)}_j)  \rho(Y^{(i)}_j))) &=\quad
	\text{~[by~\eqref{equation:lubAndGlb}]}\\
\rho({\textstyle\bigcup_{j=1}^{k_i}}\rho(X^{(i)}_j)  \rho(Y^{(i)}_j)) & \enspace . 
\end{align*}

\noindent
Hence, by a straightforward
componentwise application on vectors in \(\wp(\Sigma^*)^{|\cV|}\), we obtain that \(\rho\) is backward complete for \(\Fn_\cG\). 
Finally,  \(\rho\) is backward complete for 
\(\lambda\vect{X}\ldotp (\vect{b}\cup \Fn_{\cG}(\vect{X}))\),
because: 
\begin{align*}
	\rho(\vect{b}\cup \Fn_{\cG}(\rho(\vect{X}))) &= 
\quad\text{~[by~\eqref{equation:lubAndGlb}]}\\
	\rho(\rho(\vect{b})\cup\rho(\Fn_{\cG}(\rho (\vect{X})))) &= 
\quad\text{~[by backward completeness for \(\Fn_{\cG}\)]}\\
\rho(\rho(\vect{b})\cup\rho(\Fn_{\cG}(\vect{X}))) &=
\quad\text{~[by~\eqref{equation:lubAndGlb}]}\\
\rho(\vect{b}\cup \Fn_{\cG}(\vect{X})) &\enspace . \tag*{\qedhere}
\end{align*}
\end{proof}

The following result, which is an adaptation of Theorem~\ref{theorem:FiniteWordsAlgorithmGeneral} to
grammars, relies on Theorem~\ref{theorem:rhoCFG} for designing 
an algorithm that solves the 
inclusion problem \(\lang{\cG} \subseteq L_2\)
by exploiting a language abstraction \(\rho\) that satisfies some requirements
of backward completeness and computability. 
 
\begin{theorem}\label{theorem:FiniteWordsAlgorithmGeneral:CFG}
Let \(\cG=\tuple{\cV,\Sigma,P}\) be a CFG in CNF, 
\(L_2\in \wp(\Sigma^*)\) 
and  \(\rho \in \uco(\Sigma^*)\). 
Assume that the following properties hold:
\begin{enumerate}[\upshape(\rm i\upshape)]
\item The closure \(\rho\) is backward complete for both \(\lambda X\in \wp(\Sigma^*)\ldotp aX\) and \(\lambda X\in \wp(\Sigma^*)\ldotp Xa\) for all \(a\in \Sigma\) and
satisfies \(\rho(L_2) = L_2\).\label{theorem:FiniteWordsAlgorithmGeneral:CFG:rho}
\item \(\rho(\wp(\Sigma^*))\) does not contain
infinite ascending chains. \label{theorem:FiniteWordsAlgorithmGeneral:CFG:ACC}
\item If $X,Y\in \wp(\Sigma^*)$ are finite sets of words then 
the inclusion  $\rho(X) \subseteqm \rho(Y)$ is decidable. 
\label{theorem:FiniteWordsAlgorithmGeneral:CFG:EQ}
\item If $Y\in \wp(\Sigma^*)$ is a finite set of words then 
the inclusion $\rho(Y) \subseteqm L_2$ is decidable. 
\label{theorem:FiniteWordsAlgorithmGeneral:CFG:inclrho}
\end{enumerate}

\medskip
\noindent
Then,

\medskip
\(\tuple{Y_i}_{i \in [0,n]} := \Kleene (\Incl_\rho,\lambda\vect{X}\ldotp \vect{b}\cup \Fn_{\cG}(\vect{X}), \vect{\varnothing})\)\emph{;}

\emph{\textbf{return}}
 \(\Incl_\rho(\tuple{Y_i}_{i \in [0,n]},\vectarg{L_2}{X_0})\)\emph{;}

\medskip
\noindent
is a decision algorithm for \(\lang{\cG} \subseteq L_2\).
\end{theorem}
\begin{proof}
Analogous to the proof of Theorem~\ref{theorem:FiniteWordsAlgorithmGeneral}.
\end{proof}

\subsection{Instantiating the Framework}
Let us instantiate the general algorithmic framework provided by Theorem~\ref{theorem:FiniteWordsAlgorithmGeneral:CFG} to the class of closure operators induced by quasiorder relations on words. 
As a consequence of Lemma~\ref{lemma:properties},
we have the following characterization of $L$-consistent quasiorders. 

\begin{lemma}\label{lemma:propertiesCFG}
Let \(L\in \wp(\Sigma^*)\) and \(\mathord{\leqslant_L}\) be a quasiorder on \(\Sigma^*\).
Then, \(\mathord{\leqslant_L}\) is a \(L\)-consistent quasiorder on \(\Sigma^*\) if and only if
\begin{compactenum}[\upshape(\rm a\upshape)]
\item \(\rho_{\leqslant_L}(L) = L\), and \label{lemma:propertiesCFG:L}
\item \(\rho_{\leqslant_L}\) is backward complete for 
for \(\lambda X\ldotp a X\) and \(\lambda X\ldotp Xa\) for all \(a\in \Sigma\).\label{lemma:propertiesCFG:bw}
\end{compactenum}
\end{lemma}

Analogously to Section~\ref{sec:word-based} for automata, Theorem~\ref{theorem:FiniteWordsAlgorithmGeneral:CFG} 
induces an algorithm for deciding the language inclusion \(\lang{\cG} \subseteq L_2\) for any CFG \(\cG\) and regular language \(L_2\).
More in general, given a language \(L_2\in \wp(\Sigma^*)\) whose membership problem is decidable and a decidable \(L_2\)-consistent wqo, the following algorithm \AlgGrammarW ($\mathtt{CFG}$ $\mathtt{Inc}$lusion based on $\mathtt{W}$ords) decides 
\(\lang{\cG} \subseteq L_2\). 

\begin{figure}[H]
\RemoveAlgoNumber
\begin{algorithm}[H]
\SetAlgorithmName{\AlgGrammarW}{}
\SetSideCommentRight
\caption{Word-based algorithm for \(\lang{\cG} \subseteq L_2\)}\label{alg:CFGIncW}

\KwData{CFG \(\cG=\tuple{\cV,\Sigma, P}\); decision procedure for \(u\mathrel{\inqm} L_2\); decidable \(L_2\)-consistent wqo \(\mathord{\leqslant_{L_2}}\).}

\(\tuple{Y_i}_{i \in [0,n]} := \Kleene (\sqsubseteq_{\leqslant_{L_2}}, \lambda \vect{X}\ldotp \vect{b} \cup \Fn_{\cG}(\vect{X}), \vect{\varnothing})\)\;

\ForAll{\(u \in Y_0\)}{
  \lIf{\(u \notin L_2\)}{\Return \textit{false}}
}
\Return \textit{true}\;
\end{algorithm}
\end{figure}

\begin{theorem}\label{theorem:quasiorderAlgorithmGr}
  Let \(\cG=\tuple{Q,\delta,I,F,\Sigma}\) be a CFG and let \(L_2\in \wp(\Sigma^*)\) be a language such that: 
  \begin{inparaenum}[\upshape(\rm i\upshape)]
  \item membership $u\inqm L_2$ is decidable; \label{theorem:quasiorderAlgorithmGr:membership}
  \item there exists a decidable \(L_2\)-consistent wqo on $\Sigma^*$. \label{theorem:quasiorderAlgorithmGr:decidableL}
  \end{inparaenum} 
  Then, algorithm \AlgGrammarW decides the inclusion \(\lang{\cG} \subseteq L_2\).
\end{theorem}
\begin{proof}
The proof is analogous to the proof of Theorem~\ref{theorem:quasiorderAlgorithm}:
it applies Theorem~\ref{theorem:FiniteWordsAlgorithmGeneral:CFG} and Lemma~\ref{lemma:propertiesCFG} in the same way of the proof of Theorem~\ref{theorem:quasiorderAlgorithm} where the role of 
a left $L_2$-consistent wqo on $\Sigma^*$
is replaced by a \(L_2\)-consistent wqo. 
\end{proof}

\subsubsection{Myhill and State-based Quasiorders}
In the following, we will consider two quasiorders on $\Sigma^*$ and we will show that they fulfill the requirements of Theorem~\ref{theorem:quasiorderAlgorithmGr}, so that they correspondingly yield algorithms for deciding the language inclusion \(\lang{\cG} \subseteq L_2\) for every CFG \(\cG\) and regular language \(L_2\).

The \emph{context} for a language \(L \in \wp(\Sigma^*)\) w.r.t.\ a given word \(w\in \Sigma^*\) is defined as usual:  
\begin{align*}
\ctx_L(w) &\ud \{(u, v) \in \Sigma^* \times \Sigma^* \mid  uwv\in L\}\enspace .
\end{align*}
Correspondingly, let us define the following quasiorder relation on \(\mathord{\leqq_L}\subseteq \Sigma^*\times \Sigma^*\):
\begin{align}\label{def-Myhillqo}
  u\leqq_L v &\udr\; \ctx_L(u) \subseteq \ctx_L(v) \enspace . 
\end{align}
De Luca and Varricchio \citeyearpar[Section~2]{deLuca1994} call \(\leqq_L\) the \emph{Myhill quasiorder relative to \(L\)}. 
The following result is the analogue
of Lemma~\ref{lemma:leftrightnerodegoodqo} for  the Nerode quasiorder: 
it shows that the Myhill quasiorder is the weakest 
\(L_2\)-consistent quasiorder for which the above algorithm \AlgGrammarW can be instantiated to decide a language inclusion \(\lang{\cG}\subseteq L_2\).

\begin{lemma}\label{lemma:myhillgoodqo}
Let $L\in \wp(\Sigma^*)$.  
\begin{compactenum}[\upshape(\rm a\upshape)]
\item \(\mathord{\leqq_L}\) is a \(L\)-consistent quasiorder.
If $L$ is regular then, additionally, \(\mathord{\leqq_L}\) is a decidable wqo. \label{lemma:myhillgoodqo:Consistent}
\item If \(\mathord{\leqslant}\) is a \(L\)-consistent quasiorder on $\Sigma^*$ then \( \rho_{\leqq_L}(\wp(\Sigma^*)) \subseteq \rho_{\leqslant}(\wp(\Sigma^*)) \).\label{lemma:myhillgoodqo:Incl}
\end{compactenum}
\end{lemma}
\begin{proof}
The proof follows the same lines of the proof of Lemma~\ref{lemma:leftrightnerodegoodqo}.

\noindent
Let us consider \eqref{lemma:leftrightnerodegoodqo:Consistent}.
De Luca and Varricchio~\citeyearpar[Section~2]{deLuca1994} observe that \(\mathord{\leqq_L}\) is monotonic. 
Moreover, if 
$L$ is regular then \(\mathord{\leqq_L}\) is a wqo \cite[Proposition~2.3]{deLuca1994}. 
Let us observe that given \(u \in L\) and \(v \notin L\) we have that \((\epsilon, \epsilon) \in \ctx_L(u)\) while \((\epsilon, \epsilon) \notin \ctx_L(v)\). 
Hence, \(\mathord{\leqq_L}\) is a \(L\)-consistent quasiorder.
Finally, if $L$ is regular then \(\leqq_L\) is clearly decidable.

\noindent
Let us consider \eqref{lemma:leftrightnerodegoodqo:Incl}.
By the characterization of $L$-consistent quasiorders of Lemma~\ref{lemma:propertiesCFG}, 
De Luca and Varricchio~\citeyearpar[Section~2, point~4]{deLuca1994} observe that \(\mathord{\leqq_L}\) is maximum in the set of all \(L\)-consistent quasiorders, i.e.\ every \(L\)-consistent quasiorder \(\leqslant\) is such that 
  \(x \leqslant y \Ra x \leqq_L y \).
As a consequence, \(\rho_{\leqslant}(X) \subseteq \rho_{\leqq_L}(X)\) holds for all \(X\in \wp(\Sigma^*)\), namely, 
\( \rho_{\leqq_L}(\wp(\Sigma^*)) \subseteq \rho_{\leqslant}(\wp(\Sigma^*)) \). 
\end{proof}

\begin{figure}[t]
    \centering
  \begin{tikzpicture}[->,>=stealth',shorten >=1pt,auto,node distance=5mm and 1cm,thick,initial text=]
  \tikzstyle{every state}=[scale=0.75,fill=blue!20,draw=blue!60,text=black]
  
  \node[initial,state] (1) {\(q_1\)};
  \node[state] (2) [right=of 1] {\(q_2\)};
  \node[state, accepting] (3) [right=of 2] {\(q_3\)};
  
  \path (1) edge node {\(a\)} (2)
        (2) edge node {\(a\)} (3)
        (2) edge[loop above] node {\(b\)} (2)
        (3) edge[loop above] node {\(a,b\)} (3)
        (1) edge[bend right=50] node {\(b\)} (3)
            ;
  \end{tikzpicture}
  \caption{A finite automaton \(\cA\) with \(\lang{\cA}= (b+ab^*a)(a+b)^*\).}\label{fig:C}
\end{figure}

\begin{example}\label{example:CFGIncA}
Let us illustrate the use of the Myhill quasiorder $\leqq_{\lang{\cA}}$ in Algorithm~\AlgGrammarW 
for solving the language inclusion \(\lang{\cG} \subseteq \lang{\cA}\), where \(\cG\) is the CFG in Example~\ref{example:cfg} and \(\cA\) is the FA depicted in Figure~\ref{fig:C}.
The equations for \(\cG\) are as follows:
\[\Eqn(\cG) = \begin{cases}
    X_0 = X_0X_1 \cup X_1X_0 \cup \{b\}\\
    X_1 =\{a\}
  \end{cases} \enspace .\]

\noindent
We write \(\{(S,T)\} \cup \{(X,Y)\}\) to compactly denote a set \(\{(u,v) \mid (u,v) \in S\times T \cup X \times Y\}\).
Then, we have the following contexts (among others) for $L=\lang{\cA}=(b+ab^*a)(a+b)^*$:
\begin{align*}
\ctx_L(\epsilon) & = \{(\epsilon, L)\} \cup \{(ab^*, b^*a\Sigma^*)\} \cup \{(L, \Sigma^*)\}\\ 
\ctx_L(a) & = \{(\epsilon, b^*a\Sigma^*)\} \cup \{ab^*, \Sigma^*\} \cup \{(L,\Sigma^*)\} \\
\ctx_L(b) &= \{(\epsilon, \Sigma^*)\} \cup \{(ab^*, b^*a\Sigma^*)\} \cup \{(L, \Sigma^*)\}\\ 
\ctx_L(ba) & = \{(\epsilon, \Sigma^*)\} \cup \{(ab^*, \Sigma^*)\} \cup \{(L, \Sigma^*)\}
\end{align*}
Notice that \(a \leqq_{L} ba\) and 
\(\ctx_L(ab) = \ctx_L(a)\) and \(\ctx_L(ba) = \ctx_L(baa) =\ctx_L(aab) = \ctx_L(aba)\).
Next, we show the computation of the Kleene iterates according to Algorithm \AlgGrammarW using $\sqsubseteq_{\leqq_L}$ by recalling from
Example~\ref{example:cfg} that $\vect{b}=\tuple{\{b\}, \{a\}}$ and 
$\Fn_{\cG }(\tuple{X_0,X_1}) =\tuple{X_0X_1 \cup X_1X_0, \varnothing}$:
\begin{align*}
\vect{Y}^{(0)} &= \vect{\varnothing}\\
\vect{Y}^{(1)} &= \vect{b} = \tuple{\{b\}, \{a\}} \\
\vect{Y}^{(2)} &= \vect{b} \cup \Fn_{\cG}(\vect{Y}^{(1)}) 
= \tuple{\{b\}, \{a\}} \cup \tuple{\{ba,ab\},\varnothing}= \tuple{\{ba,ab,b\}, \{a\}}\\ 
\vect{Y}^{(3)} &= \vect{b} \cup \Fn_{\cG}(\vect{Y}^{(2)}) = \tuple{\{b\}, \{a\}} \cup  \tuple{\{baa,aba,ba,aab,ab \},\varnothing} = \tuple{\{baa,aba,ba,aab,ab,b\},\{a\}}
\end{align*}
It turns out that $\tuple{\{baa,aba,ba,aab,ab,b\},\{a\}} \sqsubseteq_{\leqq_L} \tuple{\{ba,ab,b\}, \{a\}}$ 
because $a \leqq_L baa$,  $a \leqq_L aba$,  $a \leqq_L aab$ hold, so that 
$\Kleene (\sqsubseteq_{\leqq_{L}}, \lambda \vect{X}\ldotp \vect{b} \cup \Fn_{\cG}(\vect{X}), \vect{\varnothing})$ stops with $\vect{Y}^{(3)}$ and 
outputs $\vect{Y}= \tuple{\{ba,ab,b\}, \{a\}}$.
Since $ab\in \vect{Y}_0$ but \(ab \notin \lang{\cA}\), Algorithm~\AlgGrammarW correctly 
concludes that \(\lang{\cG} \subseteq \lang{\cA}\) does not hold. \eox
\end{example}

Similarly to Section~\ref{subsec:state-qos}, next we consider a state-based quasiorder that can be used with Algorithm~\AlgGrammarW. 
First, given a FA \(\cA = \tuple{Q, \delta, I, F, \Sigma}\) we define the state-based equivalent of the context of a word \(w \in \Sigma^*\) as follows:
\[\ctx_{\cA}(w) \ud \{(q,q') \in Q\times Q \mid q \stackrel{w}{\leadsto} q' \} \enspace .\]
Then, the quasiorder \(\leq_{\cA}\) on $\Sigma^*$ 
induced by $\cA$ is defined as follows: for all $u,v\in \Sigma^*$,
\begin{equation}\label{eqn:state-qo:CFG}
u \leq_{\cA} v  \udr \ctx_{\cA}(u) \subseteq \ctx_{\cA}(v)\enspace .
\end{equation}
The following result is the analogue of Lemma~\ref{lemma:LAconsistent} and 
shows that \(\mathord{\leq_{\cA}}\) is a \(\lang{\cA}\)-consistent well-quasiorder and, therefore, it can be used with Algorithm~\AlgGrammarW to solve a language inclusion \(\lang{\cG} \subseteq \lang{\cA}\).

\begin{lemma}\label{lemma:LAconsistent:CFG}
The relation \(\mathord{\leq_{\cA}}\) is a decidable \(\lang{\cA}\)-consistent wqo.
\end{lemma}
\begin{proof}
For every \(u \in \Sigma^*\), \(\ctx_{\cA}(u)\) is a finite and computable set, so that \(\mathord{\leq_{\cA}}\) is a decidable wqo. 
Next, we show that \(\mathord{\leq_{\cA}}\) is \(\lang{A}\)-consistent according to Definition~\ref{def:LConsistent}~(\ref{eq:LConsistentPrecise})-(\ref{eq:LConsistentmonotonic}). 

\noindent
\eqref{eq:LConsistentPrecise} By picking \(u\in \lang{\cA}\) and \(v\notin \lang{\cA}\) we have that \(\ctx_{\cA}(u)\) contains a pair \((q_i, q_f)\) with \(q_i \in I\) and \(q_f \in F\) while \(\ctx_{\cA}(v)\) does not, hence \(u \nleq_{\cA} v\).

\noindent
\eqref{eq:LConsistentmonotonic} Let us check that $\leq_{\cA}$ is monotonic.
Observe that $\ctx_{\cA}: \tuple{\Sigma^*,\leq_{\cA}} \ra \tuple{\wp(Q^2),\subseteq}$ is monotonic.
Therefore, for all $x_1,x_2\in \Sigma^*$ and $a,b\in \Sigma$, 
\begin{align*}
x_1 \leq_{\cA} x_2 & \Ra  \quad\text{~[by definition of \(\leq_{\cA}\)]} \\
\ctx_{\cA}(x_1) \subseteq \ctx_{\cA}(x_2) & \Ra  \quad\text{~[as $\ctx_\cA$ is monotonic]} \\
\ctx_{\cA}(ax_1 b) \subseteq \ctx_{\cA}(a x_2 b) & \Ra \quad \text{~[by definition of \(\leq_{\cA}\)]} \\
ax_1b \leq_{\cA} ax_2b & \enspace . \tag*{\qedhere}
\end{align*}
\end{proof} 

For the Myhill wqo
$\leqq_{\lang{\cA}}$, it turns out that for all $u,v\in \Sigma^*$,
\begin{align*}
& u \leqq_{\lang{\cA}} v  \Lra \ctx_{\lang{\cA}}(u) \subseteq \ctx_{\lang{\cA}}(v) \Lra \\[-2pt]
&\{(x, y) \mid x \in W_{I,q} \land y \in W_{q', F} \land q \stackrel{u}{\leadsto} q'\} \subseteq \{(x, y) \mid x \in W_{I,q} \land y \in W_{q', F} \land q \stackrel{v}{\leadsto} q'\} \enspace .
\end{align*}
Therefore, \(u \leq_{\cA} v \Ra u \leqq_{\lang{\cA}} v\) and, consequently, \( \rho_{\leqq_{\lang{\cA}}}(\wp(\Sigma^*)) \subseteq \rho_{\leq^l_{\cA_2}}(\wp(\Sigma^*))\) holds.

\begin{example}
Let us illustrate the use of the state-based quasiorder $\leq_{\cA}$ to solve the language inclusion \(\lang{\cG} \subseteq \lang{\cA}\) of Example~\ref{example:CFGIncA}.
Here, among others, we have the following contexts:
\begin{align*}
\ctx_{\cA}(\epsilon) & = \{(q_1, q_1), (q_2, q_2), (q_3,q_3)\} & \ctx_{\cA}(a) & = \{(q_1, q_2), (q_2, q_3), (q_3,q_3)\} \\
\ctx_{\cA}(b) &= \{(q_1,q_3), (q_2,q_2), (q_3, q_3)\} & \ctx_{\cA}(ba) & = \{(q_1,q_3), (q_2,q_3), (q_3,q_3)\}
\end{align*}
Moreover, \(\ctx_{\cA}(ba) = \ctx_{\cA}(baa) = \ctx_{\cA}(aab) = \ctx_{\cA}(aba)\).
Recall from Example~\ref{example:CFGIncA} that for the Myhill quasiorder we have that \(a \leqq_{\lang{\cA}} ba\), while for the state-based quasiorder \(a \nleq_{\cA} ba\).
The Kleene iterates computed by Algorithm \AlgGrammarW when using 
\(\sqsubseteq_{\mathord{\leq_{\cA}}}\) are exactly the same of Example~\ref{example:CFGIncA}. 
Here, it turns out that \AlgGrammarW outputs $\vect{Y}^{(2)} = \tuple{\{ba,ab,b\}, \{a\}}$ because
$\vect{Y}^{(3)} = \tuple{\{baa,aba,ba,aab,ab,b\},\{a\}} \sqsubseteq_{\leqq_L} \tuple{\{ba,ab,b\}, \{a\}}=\vect{Y}^{(2)}$ holds: in fact, we have that
$ba \leq_{\cA} baa$,  $ba \leq_{\cA} aba$,  $ba \leq_{\cA} aab$ hold. 
Since $ab\in \vect{Y}^{(2)}_0$ but \(ab \notin \lang{\cA}\), Algorithm~\AlgGrammarW derives that \(\lang{\cG} \not\subseteq \lang{\cA}\). \eox
\end{example}

\subsection{An Antichain Inclusion Algorithm for CFGs}\label{sec:ACGrammar}
We can easily formulate an equivalent of 
Theorem~\ref{theorem:EffectiveAlgorithm} for context-free languages, therefore 
defining an algorithm for solving \(\lang{\cG} \subseteq L_2\) 
by computing on an abstract domain as defined by a Galois connection. 

\begin{theorem}\label{theorem:EffectiveAlgorithmCFG}
Let \(\cG=\tuple{\cV,\Sigma,P}\) be a CFG in CNF and let \(L_2\in \wp(\Sigma^*)\).
Let  \( \tuple{D,\leq_D}\) be a poset and
\(\tuple{\wp(\Sigma^*),\subseteq} \galois{\alpha}{\gamma}\tuple{D,\sqsubseteq}\) be a GC.
Assume that the following properties hold:
\begin{compactenum}[\upshape(i\upshape)]
\item \(L_2\in\gamma(D)\) and for every \( a \in \Sigma\), \(X \in \wp(\Sigma^*)\), \(\gamma(\alpha(a X)) = \gamma(\alpha(a \gamma(\alpha(X))))\) and \(\gamma(\alpha(Xa)) = \gamma(\alpha(\gamma(\alpha(X))a))\).\label{theorem:EffectiveAlgorithmCFG:prop:rho}
\item \((D,\leq_D,\sqcup,\bot_D)\) is an effective domain, meaning that: \((D,\leq_D,\sqcup,\bot_D)\) is an ACC join-semilattice with bottom $\bot_D$, 
every element of \(D\) has a finite representation, the binary relation 
\(\leq_D\) is decidable and the binary lub \(\sqcup\) is computable.\label{theorem:EffectiveAlgorithmCFG:prop:absdecidable}
\item There is an algorithm, say \(\Fn^{\sharp}(\vect{X}^\sharp)\), which computes \(\alpha\comp \Fn_{\cG}\comp \gamma \).
\label{theorem:EffectiveAlgorithmCFG:prop:abspre}
\item There is an algorithm, say \(\base^\sharp\), which computes \(\alpha(\vect{b})\).\label{theorem:EffectiveAlgorithmCFG:prop:abseps}
\item There is an algorithm, say \(\absincl\), which decides 
\(\vect{X}^\sharp \leq_D \alpha(\vectarg{L_2}{X_0})\), for all \(\vect{X}^\sharp\in \alpha(\wp(\Sigma^*))^{|\cV|}\).
\label{theorem:EffectiveAlgorithmCFG:prop:absincl}
\end{compactenum}
Then, 

\medskip
\(\tuple{Y_i^\sharp}_{i \in [0,n]} := \Kleene (\leq_D,\lambda \vect{X}^\sharp\ldotp\base^\sharp \sqcup \Fn^{\sharp}(\vect{X}^\sharp), \vect{\bot_D})\)\emph{;}

\emph{\textbf{return}} \(\absincl(\tuple{Y_i^\sharp}_{i \in [0,n]})\)\emph{;}

\medskip
\noindent
is a decision algorithm for \(\lang{\cG} \subseteq L_2\).
\end{theorem}
\begin{proof}
Analogous to the proof of Theorem~\ref{theorem:EffectiveAlgorithm}.
\end{proof}

Similarly to what is done in Section~\ref{sec:usingGC}, 
in order to solve 
an inclusion problem $\cL(\cG) \subseteq \cL(\cA)$, where 
$\cA = \tuple{Q,\delta,I,F,\Sigma}$ is a FA, 
we leverage Theorem~\ref{theorem:EffectiveAlgorithmCFG}
to systematically design a ``state-based'' algorithm that computes Kleene iterates
on the antichain poset \(\tuple{\AC_{\tuple{\wp(Q\times Q),\subseteq}},\sqsubseteq}\) viewed as an abstraction of \(\tuple{\wp(\Sigma^*), \subseteq}\).
Here, the abstraction and concretization maps
\(\alpha\colon \wp(\Sigma^*) \ra \AC_{\tuple{\wp(Q\times Q),\subseteq}}\) and
\(\gamma\colon \AC_{\tuple{\wp(Q\times Q),\subseteq}}\ra\wp(\Sigma^*)\) and the  function
\({\Fn}_{\cG}^{\cA}: (\AC_{\tuple{\wp(Q\times Q),\subseteq}})^{|\cV|} \ra (\AC_{\tuple{\wp(Q\times Q),\subseteq}})^{|\cV|}\) are defined 
as follows:
\begin{align*}
  &\alpha(X)\ud \minor{\{\ctx_{\cA}(u) \in \wp(Q\times Q) \mid u \in X\}}\; ,
  \qquad 
  \gamma(Y)  \ud \{v \in \Sigma^* \mid \exists y \in Y, y \subseteq \ctx_{\cA}(v)\}\; , \\
  &\Fn_{\cG }^{\cA}(\tuple{X_i}_{i\in[0,n]})  \ud \langle \minor{\{X_j \comp X_k \in \wp(Q\times Q) \mid X_i {\to} X_j X_k \in P\}} \rangle_{i \in [0,n]} \;,
\end{align*}
where $\minor{X}$ is the unique minor set w.r.t.\ subset inclusion 
of some $X\subseteq \wp(Q\times Q)$ and 
\(X \comp Y \ud \{(q,q') \in Q\times Q \mid (q,q'') \in X,\, (q'',q') \in Y\}\) denotes
the standard composition of two 
relations  $X,Y\subseteq Q\times Q$. By
the analogue of Lemma~\ref{lemma:rhoisgammaalpha} (the proof follows the same pattern
and is therefore omitted), it turns out that:
\begin{compactenum}[\upshape(\rm a\upshape)]
\item \(\tuple{\wp(\Sigma^*),\subseteq}\galois{\alpha}{\gamma}\tuple{\AC_{\tuple{\wp(Q\times Q),\subseteq}},\sqsubseteq}\) is a GC,\label{lemma:rhoisgammaalpha:GCCFG}
\item \(\gamma \comp \alpha = \rho_{\leqslant_{\cA}}\),\label{lemma:rhoisgammaalpha:rhoCFG}
\item \(\Fn_{\cG}^{\cA} = \alpha \comp \Fn_{\cG} \comp \gamma\).\label{lemma:rhoisgammaalpha:preCFG}
\end{compactenum}

Thus,  the GC \(\tuple{\wp(\Sigma^*),\subseteq}\galois{\alpha}{\gamma}\tuple{\AC_{\tuple{\wp(Q\times Q),\subseteq}},\sqsubseteq}\) and the abstract 
function \(\Fn_{\cG}^{\cA}\) satisfy the hypotheses (i)-(iv) of
Theorem~\ref{theorem:EffectiveAlgorithmCFG}. Here, the inclusion check 
\(\vect{X}^\sharp \leq_D \alpha(\vectarg{\cL(\cA)}{X_0})\) boils down to verify 
that for the start component $Y_0$ of the output $\tuple{Y_i}_{i\in[0,n]}$ of 
$\Kleene (\sqsubseteq, \lambda \vect{X}^\sharp\ldotp\alpha(\vect{b}) \sqcup \Fn_{\cG}^{\cA}(\vect{X}^\sharp), \vect{\varnothing})$, for all $R\in Y_0$, $R$ does not contain
a pair $(q_i,q_f)\in I\times F$. We therefore derive the following state-based
algorithm \AlgGrammarA ($\mathtt{S}$ stands for state) that decides 
an inclusion \(L(\cG) \subseteq L(\cA)\) on the abstract domain of antichains 
$\AC_{\tuple{\wp(Q\times Q),\subseteq}}$.

\RemoveAlgoNumber
\begin{algorithm}[!ht]
\SetAlgorithmName{\AlgGrammarA}{}

\caption{State-based algorithm for \(L(\cG) \subseteq L(\cA)\)}\label{alg:CFGIncA}

\KwData{CFG \(\cG = \tuple{\cV,\Sigma,P}\) and FA \(\cA = \tuple{Q,\delta,I,F,\Sigma}\)}

\(\tuple{Y_i}_{i\in[0,n]} := \Kleene (\lambda \vect{X}^\sharp\ldotp\alpha(\vect{b}) \sqcup \Fn_{\cG}^{\cA}(\vect{X}^\sharp), \vect{\varnothing})\)\;

\ForAll{\(R \in Y_0\)}{
  \lIf{\(R \cap (I \times F) = \varnothing\)}{\Return \textit{false}}
}
\Return \textit{true}\;
\end{algorithm}

\begin{theorem}\label{theorem:statesQuasiorderAlgorithmCFG}
 The algorithm \AlgGrammarA decides the inclusion problem \(L(\cG) \subseteq L(\cA)\).
\end{theorem}
\begin{proof}
The proof follows the same pattern of the proof of Theorem~\ref{theorem:statesQuasiorderAlgorithm}. We just focus
on the inclusion check  at lines 2-4, which is slightly different from 
the check at lines 2-5 of Algorithm \AlgRegularA. Let $L_2=\cL(\cA)$. 
Since \(\alpha(\vectarg{L_2}{X_0})=\tuple{\alpha(\nullable{i \eqqm 0}{L_2}{\Sigma^*})}_{i \in [0,n]}\), for all $\vect{Y}\in \alpha(\wp(\Sigma^*))^{|\cV|}$ the relation \(\vect{Y} \sqsubseteq \alpha(\vectarg{L_2}{X_0})\) trivially holds 
for all components \(Y_i\) with \(i \neq 0\).
For $Y_0$, it is enough to prove that
\(Y_0 \sqsubseteq \alpha(L_2) \Lra \forall R \in Y_q, \; R \cap (I \times F) \neq \varnothing\): 
\begin{align*}
Y_0 \sqsubseteq \alpha(L_2) & \Lra \quad\text{~[since \(Y_0 = \alpha(U)\) for some \(U \in \wp(\Sigma^*)\)]} \\
\alpha(U) \sqsubseteq \alpha(L_2) & \Lra \quad\text{~[by GC]} \\
U \subseteq \gamma(\alpha(L_2)) & \Lra \quad\text{~[by $\gamma(\alpha(L_2))=L_2$]} \\
U \subseteq L_2 & \Lra \quad\text{~[by definition of \(\ctx_{\cA}(u)\)]} \\
\forall u \in U, \ctx_{\cA}(u) \cap (I \times F) \neq \varnothing & \Lra \quad\text{~[since \(Y_0 = \alpha(U) = \lfloor \{ \ctx_{\cA}(u) \mid u\in U\} \rfloor \)]} \\
\forall R \in Y_0, R \cap I \neq \varnothing &\enspace .
\end{align*}
Hence, Theorem~\ref{theorem:EffectiveAlgorithmCFG} entails that 
Algorithm \AlgGrammarA decides \(\lang{\cG} \subseteq \lang{\cA}\). 
\end{proof}

The resulting algorithm \AlgGrammarA shares some features with two previous related works.
On the one hand, it is related to the work of \citet{Hofmann2014} which defines an abstract interpretation-based language inclusion decision procedure similar to ours.  
Even though Hofmann and Chen's algorithm and ours both manipulate sets of pairs of states of an automaton, 
their abstraction is based on equivalence relations and not quasiorders. 
Since quasiorders are strictly more general than equivalences our framework can be instantiated to 
a larger class of abstractions, most importantly coarser ones. 
Finally, it is worth pointing out 
that \citet{Hofmann2014} approach aims at including languages of finite and also infinite words.

A second related work is that of \citet{Holk2015} who define an antichain-based algorithm manipulating sets of pairs of states. 
However, they tackle the inclusion problem \(\lang{\cG} \subseteq \lang{\cA}\), where \(\cG\) is a grammar and \(\cA\) and automaton, by rephrasing it as a data flow analysis problem over a relational domain.
In this scenario, the solution of the problem requires the computation of a least fixpoint on the relational domain, followed by an inclusion check between sets of relations.
Then, they use the ``antichain principle'' to improve the performance of the fixpoint computation and, finally, they move from manipulating relations to manipulating pairs of states.
As a result, \citet{Holk2015} devise an antichain algorithm for checking the inclusion \(\lang{\cG} \subseteq \lang{\cA}\).

By contrast to these two approaches, our design 
technique is direct and systematic, since the algorithm \AlgGrammarA is derived 
from the known Myhill quasiorder.
We believe that our approach reveals the relationship between the original antichain algorithm by \citet{DBLP:conf/cav/WulfDHR06} for regular languages and the one by \citet{Holk2015} for context-free languages, which is the relation between our algorithms~\AlgRegularA and~\AlgGrammarA.
Specifically, we have shown that these two algorithms are conceptually identical and just differ in the well-quasiorder used to define the abstract domain where
computations take place.

\section{An Equivalent Greatest Fixpoint Algorithm}%
\label{sec:greatest_fixpoint_based_algorithm}

Let us assume that 
\(g \colon C\ra C\) is a monotonic function on a complete lattice $\tuple{C,\leq,\vee,\wedge}$ which admits 
its unique right-adjoint \(\widetilde{g} \colon C\ra C\), 
i.e., $\forall c,c'\in C,\, g(c)\leq c' \Lra c\leq \widetilde{g}(c')$ holds.
Then, \citet[Theorem~4]{cou00} shows that the following equivalence holds:
for all \(c,c'\in C\), 
\begin{equation}\label{eqn:duality}
\lfp(\lambda x\ldotp c \vee g(x)) \leq c' \;\Lra\;
c\leq \gfp(\lambda y\ldotp c' \wedge \widetilde{g}(y)) \enspace .
\end{equation}
This property has been used in \cite{cou00} to derive equivalent least/greatest fixpoint-based invariance proof methods for programs. 
In the following, we use \eqref{eqn:duality} to derive an algorithm for deciding the inclusion \(\lang{\cA_1}\subseteq \lang{\cA_2}\), which relies on the computation of a greatest fixpoint rather than a least fixpoint.
This can be achieved by exploiting the following simple observation, which defines an adjunction between concatenation and quotients
of sets of words.
\begin{lemma}\label{lemma:adjointbinary} 
For all \(X,Y \in \wp(\Sigma^*)\) and \(w\in \Sigma^*\),  \(wY \subseteq Z \Lra Y \subseteq w^{-1}Z\) and \(Yw \subseteq Z \Lra Y \subseteq Zw^{-1}\).
\end{lemma}
\begin{proof}
	By definition, for all \(u\in \Sigma^*\), \(u \in w^{-1}Z\) if{}f \( wu \in Z\). 
Hence, \(Y\subseteq w^{-1}Z \Lra \forall u\in Y,\: wu \in Z \Lra wY\subseteq Z\). 
 Symmetrically, \(Yw\subseteq Z\) \(\Lra\) \(Y\subseteq Zw^{-1}\) holds. %
\end{proof}

Given a FA \(\cA = \tuple{Q,\delta,I,F,\Sigma}\), we define the function $\widetilde{\Pre}_\cA:\wp(\Sigma^*)^{|Q|} \ra \wp(\Sigma^*)^{|Q|}$ on $Q$-indexed vectors of
sets of words as follows:
\[
	\widetilde{\Pre}_\cA(\tuple{X_q}_{q\in Q})	\ud \langle {\textstyle\bigcap_{a\in \Sigma, q'\in \delta(q,a)}}\; a^{-1} X_q
\rangle_{q'\in Q} \enspace ,
\]
where, as usual, \(\bigcap \varnothing = \Sigma^*\). It turns out that $\widetilde{\Pre}_\cA$ is the usual weakest liberal precondition which is 
right-adjoint
of $\Pre_\cA$.
\begin{lemma}\label{lemma:FnAdjoint}
For all \(\vect{X},\vect{Y}\in \wp(\Sigma^*)^{|Q|}\), \(\Pre_{\cA}(\vect{X})\subseteq \vect{Y}\Lra \vect{X}\subseteq \widetilde{\Pre}_{\cA}(\vect{Y})\).
\end{lemma}
\begin{proof}\renewcommand{\qedsymbol}{}
	\begin{align*}
		\Pre_{\cA}(\tuple{X_q}_{q\in Q}) \subseteq \tuple{Y_q}_{q\in Q} &\Lra 
		\quad\text{~[by definition of $\Pre_{\cA}$]} \\
		\forall q\in Q, {\textstyle \bigcup_{q\ggoes{a}{q'}}} a X_{q'} \subseteq Y_q &\Lra
		\quad\text{~[by set theory]}\\
		\forall q,{q'}\in Q, q\ggoes{a} q' \Ra  a X_{q'} \subseteq Y_q &\Lra
		\quad\text{~[by Lemma~\ref{lemma:adjointbinary}]}\\
		\forall q,{q'}\in Q, q\ggoes{a} q' \Ra X_{q'} \subseteq a^{-1} Y_q &\Lra
		\quad\text{~[by set theory]}\\
		\forall {q'}\in Q, X_{q'} \subseteq {\textstyle\bigcap_{q\ggoes{a} q'}} a^{-1} Y_q&\Lra
		\quad\text{~[by definition of $\widetilde{\Pre}_{\cA}$]} \\
   \tuple{X_q}_{q\in Q} \subseteq \widetilde{\Pre}_{\cA}(\tuple{Y_q}_{q\in Q}) \tag*{$\Box$}
	\end{align*}
\end{proof}

Hence, from equivalences~\eqref{eq:lfp} and~\eqref{eqn:duality} we obtain that for all 
FAs $\cA_1$ and $L_2\in \wp(\Sigma^*)$:
\begin{equation}
\lang{\cA_1} \subseteq L_2 \:\Lra\:
\vectarg{\epsilon}{F_1} \subseteq \gfp(\lambda \vect{X}\ldotp \vectarg{L_2}{I_1} \cap \widetilde{\Pre}_{\cA_1}(\vect{X})) \enspace .  \label{eq:inclgfplfp}
\end{equation}

The following algorithm \AlgRegularGfp decides the inclusion \(\lang{\cA_1} \subseteq L_2\) when $L_2$ is regular by implementing the greatest fixpoint
computation in equivalence~\eqref{eq:inclgfplfp}. 

\begin{figure}[H]
\RemoveAlgoNumber
\begin{algorithm}[H]
\SetAlgorithmName{\AlgRegularGfp}{}
\SetSideCommentRight
\caption{Greatest fixpoint algorithm for \(\lang{\cA_1}\subseteq L_2\)}\label{alg:RegIncGfp}

\KwData{FA \(\cA_1=\tuple{Q_1,\delta_1,I_1,F_1,\Sigma}\); regular language \(L_2\).}

\(\tuple{Y_q}_{q\in Q} := \Kleene (\supseteq, \lambda \vect{X}\ldotp\vectarg{L_2}{I_1} \cap \widetilde{\Pre}_{\cA_1}(\vect{X}), \vect{{\Sigma^*}})\)\;

\ForAll{\(q \in F_1\)}{
	\lIf{\(\epsilon \notin Y_q\)}{\Return \textit{false}}
}
\Return \textit{true}\;
\end{algorithm}
\end{figure}

\noindent
The intuition behind Algorithm~\AlgRegularGfp is that
\[\lang{\cA_1} \subseteq L_2 \Lra \forall w \in \lang{\cA_1},\: (\epsilon \in w^{-1}L_2 \Lra \epsilon \in {\textstyle\bigcap_{w \in \lang{\cA_1}}} w^{-1}L_2) \enspace .\]
Therefore, \AlgRegularGfp computes the set $\textstyle\bigcap\{ w^{-1}L_2 \mid w \in \lang{\cA_1}\}$.
by using the automaton \(\cA_1\) and by 
considering prefixes of \(\lang{\cA_1}\) of increasing lengths. This means that
after \(n\) iterations of \(\Kleene\), the algorithm \AlgRegularGfp has computed 
\[
\textstyle\bigcap\{ w^{-1}L_2 \mid wu\in \lang{\cA_1},\, |w| \leq n,\, q_0 \in I_1,\, q_0 \stackrel{w}{\leadsto} q\}
\]
for every state \(q \in Q_1\).
The regularity of \(L_2\) and the property of regular languages of being closed under intersections and quotients entail that each Kleene 
iterate of $\Kleene (\supseteq, \lambda \vect{X}\ldotp\vectarg{L_2}{I_1} \cap \widetilde{\Pre}_{\cA_1}(\vect{X}), \vect{{\Sigma^*}})$
is a (computable) regular language. 
To the best of our knowledge, this gfp-based 
language inclusion  algorithm \AlgRegularGfp has never been described in the literature before.

Next, we discharge the fundamental assumption guaranteeing 
the correctness of this algorithm \AlgRegularGfp: the Kleene iterates computed by \AlgRegularGfp are finitely many.
In order to do that, we consider an abstract version of the greatest fixpoint computation exploiting 
a closure operator which ensures that the abstract Kleene iterates are finitely many. 
This closure operator $\rho_{\leq_{\cA_2}}$ will be defined by using an ordering relation $\leq_{\cA_2}$ 
induced by a FA $\cA_2$ such that 
\(L_2=\lang{\cA_2}\) and will be shown to be 
\emph{forward complete} for the function \(\lambda \vect{X}\ldotp \vectarg{L_2}{I_1} \cap \widetilde{\Pre}_{\cA_1}(\vect{X})\) 
used by \AlgRegularGfp.
\\
\indent
Forward completeness of abstract interpretations \cite{gq01}, also called
exactness \cite[Definition~2.15]{mine17},  is different 
from and orthogonal to backward completeness introduced in Section~\ref{sec:inclusion_checking_by_complete_abstractions}
and crucially used throughout Sections~\ref{sec:an_algorithmic_framework_for_language_inclusion_based_on_complete_abstractions}--\ref{sec:context_free_languages}.
In particular, a remarkable consequence 
of exploiting a forward complete abstraction is 
that the Kleene iterates of the concrete and abstract greatest fixpoint computations coincide.  
The intuition here is that this forward complete closure $\rho_{\leq_{\cA_2}}$ allows us to establish that all the Kleene iterates of \(\vect{X}\ldotp \vectarg{L_2}{I_1} \cap \widetilde{\Pre}_{\cA_1}(\vect{X})\) belong to the image of the closure $\rho_{\leq_{\cA_2}}$, more precisely that every Kleene iterate is a language which is upward closed for \(\leq_{\cA_2}\).
Interestingly, a similar phenomenon occurs in well-structured transition systems~\cite{ACJT96,Finkel2001}.
\\
\indent
Let us now describe in detail this abstraction. 
A closure \(\rho\in\uco(C)\) on a concrete domain $C$ is forward complete for a monotonic function \(f:C\ra C\) if \(\rho f \rho = f \rho\) holds. 
The intuition here is that forward completeness means that no loss of precision
is accumulated when the output of a computation of $f\rho$ is approximated by $\rho$, or, equivalently, the concrete function $f$ maps abstract elements 
of $\rho$ into abstract elements of $\rho$. 
Dually to the case of backward completeness, forward completeness implies that \(\gfp(f)=\gfp(f\rho) = \gfp(\rho f \rho)\) holds, when these greatest fixpoints exist (this is the case, e.g., when $C$ is a complete lattice).  
When the function 
\(f\colon C\ra C\) admits the right-adjoint \(\widetilde{f}\colon C\ra C\), i.e.,
$f(c) \leq c' \Lra c \leq \widetilde{f}(c')$ holds, 
it turns out that forward and backward completeness are related by the following duality~\cite[Corollary~1]{gq01}:
\begin{equation}\label{lemma:forwardbackwardtransfer}
	\text{$\rho$ is backward complete for \(f\) if{}f \(\rho\) is forward complete for \(\widetilde{f}\)}.
\end{equation}
Thus, by \eqref{lemma:forwardbackwardtransfer}, in the following result instead of 
assuming the hypotheses implying that a closure $\rho$ is forward complete for the right-adjoint $\widetilde{\Pre}_{\cA_1}$ we
state some hypotheses which guarantee that $\rho$ is backward complete for its left-adjoint, which, by Lemma~\ref{lemma:FnAdjoint}, is ${\Pre}_{\cA_1}$. 

\begin{theorem}\label{theorem:dualalgorithm}
	Let \(\cA_1 = \tuple{Q_1,\delta_1,I_1,F_1,\Sigma}\) be a FA , \(L_2\) be a regular language and
	\(\rho \in \uco(\wp(\Sigma^*))\). Let us assume that:  
\begin{compactenum}[\upshape(\rm 1\upshape)]
	\item \(\rho(L_2) = L_2\);
	\item \(\rho\) is backward complete for \(\lambda X\ldotp a X\) for all \(a\in \Sigma\).
\end{compactenum}
Then, \(\lang{\cA_1}\subseteq L_2\) if{}f \(\vectarg{\epsilon}{F_1} \subseteq \gfp(\lambda \vect{X}.\rho (\vectarg{L_2}{I_1} \cap \widetilde{\Pre}_{\cA_1}(\rho(\vect{X}))))\).
Moreover, the Kleene iterates of 
$\lambda \vect{X}\ldotp \rho(\vectarg{L_2}{I_1} \cap \widetilde{\Pre}_{\cA_1}(\rho(\vect{X})))$ and $\lambda \vect{X}\ldotp \vectarg{L_2}{I_1} \cap \widetilde{\Pre}_{\cA_1}(\vect{X})$ from the initial value $\vect{{\Sigma^*}}$ coincide in lockstep.
\end{theorem}
\begin{proof}
	Theorem~\ref{theorem:backComplete} shows that if \(\rho\) is backward complete for \(\lambda X\ldotp a X\) for every \(a\in\Sigma\) then it is backward complete for \(\Pre_{{\cA_1}}\). 
	Thus, by \eqref{lemma:forwardbackwardtransfer}, \(\rho\) is forward complete for \(\widetilde{\Pre}_{\cA_1}\).
	Then, it turns out that \(\rho\) is forward complete for \(\lambda \vect{X}\ldotp \vectarg{L_2}{I_1} \cap \widetilde{\Pre}_{\cA_1} (\vect{X})\), because: 
	\begin{align*}
	\rho (\vectarg{L_2}{I_1} \cap  \widetilde{\Pre}_{\cA_1} (\rho(\vect{X}))) & = 
	\quad\text{~[by forward completeness for $\widetilde{\Pre}_{\cA_1}$ and \(\rho(L_2)=L_2\)]}\\
	\rho (\rho( \vectarg{L_2}{I_1}) \cap  \rho(\widetilde{\Pre}_{\cA_1} (\rho(\vect{X})))) & = 
	\quad\text{~[by \eqref{equation:lubAndGlb}]}\\
	\rho (\vectarg{L_2}{I_1}) \cap  \rho (\widetilde{\Pre}_{\cA_1} (\rho(\vect{X}))) & = 
	\quad\text{~[by forward completeness for $\widetilde{\Pre}_{\cA_1}$ and \(\rho (L_2)=L_2\)]}\\
	\vectarg{L_2}{I_1} \cap  \widetilde{\Pre}_{\cA_1} (\rho (\vect{X}))&\enspace .
	\end{align*}
Since, by forward completeness,	
$\gfp(\lambda \vect{X}\ldotp \vectarg{L_2}{I_1} \cap \widetilde{\Pre}_{\cA_1}(\vect{X})) = \gfp(\lambda\vect{X}\ldotp \rho(\vectarg{L_2}{I_1} \cap \widetilde{\Pre}_{\cA_1}(\rho(\vect{X}))))$, by equivalence \eqref{eq:inclgfplfp}, we conclude 
that \(\lang{\cA_1}\subseteq L_2\) if{}f \(\vectarg{\epsilon}{F_1} \subseteq \gfp(\lambda\vect{X}\ldotp \rho(\vectarg{L_2}{I_1} \cap \widetilde{\Pre}_{\cA_1}(\rho(\vect{X}))))\).

\noindent
Finally, we observe that the Kleene iterates of  \(\lambda \vect{X}\ldotp \vectarg{L_2}{I_1} \cap  \widetilde{\Pre}_{\cA_1} (\vect{X})\) and $\lambda \vect{X}\ldotp \rho(\vectarg{L_2}{I_1} \cap \widetilde{\Pre}_{\cA_1}(\rho(\vect{X})))$ starting from $\vect{{\Sigma^*}}$ coincide in lockstep since 
$\rho (\vectarg{L_2}{I_1} \cap  \widetilde{\Pre}_{\cA_1} (\rho(\vect{X}))) =
\vectarg{L_2}{I_1} \cap  \widetilde{\Pre}_{\cA_1} (\rho (\vect{X}))$ and 
\(\rho(\vectarg{L_2}{I_1})=\vectarg{L_2}{I_1}\). %
\end{proof}

We can now establish that the Kleene iterates of
$\Kleene (\supseteq, \lambda \vect{X}\ldotp \vectarg{L_2}{I_1} \cap \widetilde{\Pre}_{\cA_1}(\vect{X}) , \vect{{\Sigma^*}})$ are finitely many. 
Let \(L_2=\lang{\cA_2}\), for some FA \(\cA_2\), and consider the corresponding left
state-based quasiorder \(\mathord{\leq_{\cA_2}^{l}}\) on $\Sigma^*$ as defined by~\eqref{eqn:state-qo}.
By Lemma~\ref{lemma:LAconsistent}, \(\mathord{\leq_{\cA_2}^{l}}\) is a left \(L_2\)-consistent wqo.  
Furthermore, since \(Q_2\) is finite we have that both \(\mathord{\leq_{\cA_2}^{l}}\) and \((\mathord{\leq_{\cA_2}^{l}})^{-1}\) are wqos, so that, in turn, \( \tuple{\rho_{\leq_{\cA_2}^{l}}(\wp(\Sigma^*)),\subseteq}\) is a poset which is both ACC and DCC.  
In particular, the definition of \(\mathord{\leq_{\cA_2}^{l}}\) implies that every chain in \( \tuple{\rho_{\leq_{\cA_2}^{l}}(\wp(\Sigma^*)),\subseteq}\) has at most \(2^{|Q_2|}\) elements, so that
if we compute \(2^{|Q_2|}\) Kleene iterates then we surely converge to the greatest fixpoint.
Moreover, as a consequence of the DCC we have that
$\Kleene (\supseteq, \lambda \vect{X}\ldotp\rho_{\leq_{\cA_2}}(\vectarg{L_2}{I_1} \cap \widetilde{\Pre}_{\cA_1}(\rho_{\leq_{\cA_2}}(\vect{X}))), \vect{{\Sigma^*}})$
always terminates, 
thus implying that
$\Kleene (\supseteq, \lambda \vect{X}\ldotp \vectarg{L_2}{I_1} \cap \widetilde{\Pre}_{\cA_1}(\vect{X}) , \vect{{\Sigma^*}})$
terminates as well, because their Kleene iterates go in lockstep as stated by Theorem~\ref{theorem:dualalgorithm}. We have therefore shown the correctness of \AlgRegularGfp.

\begin{corollary}\label{coroGFP}
The algorithm \AlgRegularGfp decides the inclusion \(\lang{\cA_1} \subseteq L_2\)
\end{corollary}

\begin{example}
Let us illustrate the  greatest fixpoint algorithm \AlgRegularGfp on the inclusion check 
$L(\cB)\subseteq L(\cA)$ where $\cA$ is the FA in Fig.~\ref{fig:A} and 
$\cB$ is the following FA:
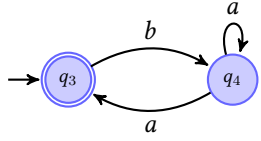
\begin{center}
	\begin{tikzpicture}[->,>=stealth',shorten >=1pt,auto,node distance=1.5cm,thick,initial text=]
	\tikzstyle{every state}=[scale=0.75,fill=blue!20,draw=blue!60,text=black]
	
	\node[state,initial,accepting] (1) {$q_3$};
	\node[state] (2) [right=of 1] {$q_4$};
	
	\path (1) edge[bend left] node {$b$} (2)
	      (2) edge[bend left] node {$a$} (1)
	      (2) edge[loop above] node {$a$} (2)
	          ;
	\end{tikzpicture}
	 \end{center}
By Corollary~\ref{coroGFP}, the Kleene iterates of $\lambda \vect{Y}\ldotp \vectarg{L(\cA)}{\{q_3\}} \cap \widetilde{\Pre}_{\cB}(\vect{Y})\) 
are guaranteed to converge in finitely many steps. We have that 
\[
\vectarg{L(\cA)}{\{q_3\}} \cap \widetilde{\Pre}_{\cB}(\tuple{Y_3,Y_4})
=\tuple{L(\cA) \cap a^{-1}Y_4,\: b^{-1}Y_3 \cap a^{-1}Y_4}\enspace .
\]
Then, the Kleene iterates are as follows (we automatically 
checked them by 
the FAdo tool \cite{Almeida2009}):
\begin{align*}
	Y_3^{(0)}	&= \Sigma^* & Y_4^{(0)} &= \Sigma^*\\
	Y_3^{(1)}	&= L(\cA)\cap a^{-1}\Sigma^* = L(\cA) & Y_4^{(1)} &= b^{-1}\Sigma^* \cap a^{-1}\Sigma^* = \Sigma^*\\
	Y_3^{(2)}	&= L(\cA) \cap a^{-1}\Sigma^* = L(\cA)  & Y_4^{(2)} &= b^{-1} L(\cA) \cap a^{-1} \Sigma^* = b^{-1} L(\cA) = (b^* a)^+\\
	Y_3^{(3)}	&= L(\cA) \cap a^{-1}(b^* a)^+ = L(\cA)  & Y_4^{(3)} &= b^{-1} L(\cA) \cap a^{-1} (b^* a)^+ = (b^* a)^+
\end{align*}
Thus, $\Kleene$ outputs  the vector 
\(\tuple{Y_3,Y_4}=
\tuple{L(\cA),(b^* a)^+}\). 
Since \(\epsilon\in L(\cA)\), \AlgRegularGfp concludes that \( L(\cB)\subseteq L(\cA)\) holds.
\eox
\end{example}

Finally, it is worth citing that Fiedor et al.~\shortcite{fiedor2019nested} put forward an algorithm for deciding WS1S formulae which relies on the same lfp computation used in \AlgRegularA.
Then, they derive a dual gfp computation by relying on Park's duality~\cite{park1969fixpoint}: \(\lfp (\lambda X \ldotp f(X)) = (\gfp (\lambda X \ldotp (f(X^c))^c))^c\).
Their approach differs from ours since we use the equivalence~\eqref{eqn:duality} to compute a gfp, different from the lfp, which still allows us to decide the inclusion problem.
Furthermore, their algorithm decides whether a given automaton accepts \(\epsilon\) and it is not clear how their algorithm could be extended for deciding language inclusion.

\section{Future Work}%
\label{sec:conclusions}

We believe that this work only scratched the surface of the use of well-quasiorders on words for solving language inclusion problems.
In particular, our approach based on complete abstract interpretations allowed us to systematically derive well-known algorithms 
, such as the antichain algorithms by De Wulf et al.~\shortcite{DBLP:conf/cav/WulfDHR06}, as well as novel algorithms, such as \AlgRegularGfp, for deciding the inclusion of regular languages.

Future directions include leveraging well-quasiorders for infinite words \cite{ogawa_well-quasi-orders_2004} to shed new light on the inclusion problem between \(\omega\)-languages.  
Our results could also be extended to inclusion of tree languages by relying on the extensions of Myhill-Nerode theorems for tree languages \cite{Kozen92onthe}.
Another interesting topic for future work is the enhancement of quasiorders using simulation relations.
Even though we already showed in this paper that simulations can be used to refine our language inclusion algorithms, we are not on par with the thoughtful use of simulation relations made by Abdulla et al.~\shortcite{Abdulla2010} and Bonchi and Pous~\shortcite{DBLP:conf/popl/BonchiP13}.   
Finally, let us mention that the correspondence between least and greatest fixpoint-based inclusion checks assuming complete abstractions was studied by Bonchi et al.~\shortcite{Bonchi2018} with the aim of formally connecting sound up-to techniques and complete abstract interpretations.
Further possible developments include the study of our abstract interpretation-based algorithms for language inclusion from the viewpoint of sound up-to techniques.  

%

%
%
%
%
\citestyle{acmauthoryear}


\begin{thebibliography}{38}

%
%
%
%
%
%
%
%
%
%
%
%
%
%
%
%

\ifx \showCODEN    \undefined \def \showCODEN     #1{\unskip}     \fi
\ifx \showDOI      \undefined \def \showDOI       #1{#1}\fi
\ifx \showISBNx    \undefined \def \showISBNx     #1{\unskip}     \fi
\ifx \showISBNxiii \undefined \def \showISBNxiii  #1{\unskip}     \fi
\ifx \showISSN     \undefined \def \showISSN      #1{\unskip}     \fi
\ifx \showLCCN     \undefined \def \showLCCN      #1{\unskip}     \fi
\ifx \shownote     \undefined \def \shownote      #1{#1}          \fi
\ifx \showarticletitle \undefined \def \showarticletitle #1{#1}   \fi
\ifx \showURL      \undefined \def \showURL       {\relax}        \fi
%
%
\providecommand\bibfield[2]{#2}
\providecommand\bibinfo[2]{#2}
\providecommand\natexlab[1]{#1}
\providecommand\showeprint[2][]{arXiv:#2}

\bibitem[\protect\citeauthoryear{Abdulla, Cerans, Jonsson, and Tsay}{Abdulla
  et~al\mbox{.}}{1996}]%
        {ACJT96}
\bibfield{author}{\bibinfo{person}{Parosh~Aziz Abdulla},
  \bibinfo{person}{Karlis Cerans}, \bibinfo{person}{Bengt Jonsson}, {and}
  \bibinfo{person}{Yih-Kuen Tsay}.} \bibinfo{year}{1996}\natexlab{}.
\newblock \showarticletitle{General decidability theorems for infinite-state
  systems}. In \bibinfo{booktitle}{\emph{Proc.\ of the 11th Annual IEEE Symp.
  on Logic in Computer Science (LICS'96)}}. \bibinfo{publisher}{IEEE Computer
  Society}, \bibinfo{address}{Washington, DC, USA}, \bibinfo{pages}{313--321}.
\newblock


\bibitem[\protect\citeauthoryear{Abdulla, Chen, Hol{\'{\i}}k, Mayr, and
  Vojnar}{Abdulla et~al\mbox{.}}{2010}]%
        {Abdulla2010}
\bibfield{author}{\bibinfo{person}{Parosh~Aziz Abdulla},
  \bibinfo{person}{Yu-Fang Chen}, \bibinfo{person}{Luk{\'{a}}{\v{s}}
  Hol{\'{\i}}k}, \bibinfo{person}{Richard Mayr}, {and}
  \bibinfo{person}{Tom{\'{a}}{\v{s}} Vojnar}.} \bibinfo{year}{2010}\natexlab{}.
\newblock \showarticletitle{When Simulation Meets Antichains}.
\newblock In \bibinfo{booktitle}{\emph{Proceedings of the 16th International
  Conference on Tools and Algorithms for the Construction and Analysis of
  Systems (TACAS'10)}}. \bibinfo{publisher}{Springer Berlin Heidelberg},
  \bibinfo{pages}{158--174}.
\newblock
\urldef\tempurl%
\url{https://doi.org/10.1007/978-3-642-12002-2_14}
\showURL{%
\tempurl}


\bibitem[\protect\citeauthoryear{Almeida, Almeida, Alves, Moreira, and
  Reis}{Almeida et~al\mbox{.}}{2009}]%
        {Almeida2009}
\bibfield{author}{\bibinfo{person}{Andr{\'{e}} Almeida}, \bibinfo{person}{Marco
  Almeida}, \bibinfo{person}{Jos{\'{e}} Alves}, \bibinfo{person}{Nelma
  Moreira}, {and} \bibinfo{person}{Rog{\'{e}}rio Reis}.}
  \bibinfo{year}{2009}\natexlab{}.
\newblock \showarticletitle{{FAdo} and {GUItar}: Tools for Automata
  Manipulation and Visualization}.
\newblock In \bibinfo{booktitle}{\emph{Implementation and Application of
  Automata}}. \bibinfo{publisher}{Springer Berlin Heidelberg},
  \bibinfo{pages}{65--74}.
\newblock
\urldef\tempurl%
\url{https://doi.org/10.1007/978-3-642-02979-0_10}
\showDOI{\tempurl}


\bibitem[\protect\citeauthoryear{Baier and Katoen}{Baier and Katoen}{2008}]%
        {baier08}
\bibfield{author}{\bibinfo{person}{Christel Baier} {and}
  \bibinfo{person}{Joost-Pieter Katoen}.} \bibinfo{year}{2008}\natexlab{}.
\newblock \bibinfo{booktitle}{\emph{Principles of Model Checking}}.
\newblock \bibinfo{publisher}{The MIT Press}.
\newblock
\showISBNx{026202649X, 9780262026499}


\bibitem[\protect\citeauthoryear{Bauer and Eickel}{Bauer and Eickel}{1976}]%
        {bauer76}
\bibfield{author}{\bibinfo{person}{Friedrich~L. Bauer} {and}
  \bibinfo{person}{J\"{u}rgen Eickel}.} \bibinfo{year}{1976}\natexlab{}.
\newblock \bibinfo{booktitle}{\emph{Compiler Construction, An Advanced Course,
  2nd Ed.}}
\newblock \bibinfo{publisher}{Springer-Verlag}, \bibinfo{address}{Berlin,
  Heidelberg}.
\newblock
\showISBNx{3-540-07542-9}


\bibitem[\protect\citeauthoryear{Bonchi, Ganty, Giacobazzi, and
  Pavlovic}{Bonchi et~al\mbox{.}}{2018}]%
        {Bonchi2018}
\bibfield{author}{\bibinfo{person}{Filippo Bonchi}, \bibinfo{person}{Pierre
  Ganty}, \bibinfo{person}{Roberto Giacobazzi}, {and} \bibinfo{person}{Dusko
  Pavlovic}.} \bibinfo{year}{2018}\natexlab{}.
\newblock \showarticletitle{Sound up-to techniques and Complete abstract
  domains}. In \bibinfo{booktitle}{\emph{Proceedings of the 33rd Annual
  {ACM}/{IEEE} Symposium on Logic in Computer Science (LICS'18)}}.
  \bibinfo{publisher}{{ACM} Press}.
\newblock
\urldef\tempurl%
\url{https://doi.org/10.1145/3209108.3209169}
\showDOI{\tempurl}


\bibitem[\protect\citeauthoryear{Bonchi and Pous}{Bonchi and Pous}{2013}]%
        {DBLP:conf/popl/BonchiP13}
\bibfield{author}{\bibinfo{person}{Filippo Bonchi} {and}
  \bibinfo{person}{Damien Pous}.} \bibinfo{year}{2013}\natexlab{}.
\newblock \showarticletitle{Checking {N}{F}{A} Equivalence with Bisimulations
  Up to Congruence}. In \bibinfo{booktitle}{\emph{Proceedings of the 40th
  Annual ACM SIGPLAN-SIGACT Symposium on Principles of Programming Languages
  (POPL'13)}}. \bibinfo{publisher}{ACM Press}, \bibinfo{pages}{457--468}.
\newblock
\showISBNx{978-1-4503-1832-7}
\urldef\tempurl%
\url{https://doi.org/10.1145/2429069.2429124}
\showDOI{\tempurl}


\bibitem[\protect\citeauthoryear{Chomsky}{Chomsky}{1959}]%
        {DBLP:journals/iandc/Chomsky59a}
\bibfield{author}{\bibinfo{person}{Noam Chomsky}.}
  \bibinfo{year}{1959}\natexlab{}.
\newblock \showarticletitle{On Certain Formal Properties of Grammars}.
\newblock \bibinfo{journal}{\emph{Information and Control}}
  \bibinfo{volume}{2}, \bibinfo{number}{2} (\bibinfo{year}{1959}),
  \bibinfo{pages}{137--167}.
\newblock


\bibitem[\protect\citeauthoryear{Clarke, Henzinger, Veith, and Bloem}{Clarke
  et~al\mbox{.}}{2018}]%
        {clarke18}
\bibfield{author}{\bibinfo{person}{Edmund~M. Clarke},
  \bibinfo{person}{Thomas~A. Henzinger}, \bibinfo{person}{Helmut Veith}, {and}
  \bibinfo{person}{Roderick Bloem}.} \bibinfo{year}{2018}\natexlab{}.
\newblock \bibinfo{booktitle}{\emph{Handbook of Model Checking}
  (\bibinfo{edition}{1st} ed.)}.
\newblock \bibinfo{publisher}{Springer Publishing Company, Incorporated}.
\newblock
\showISBNx{3319105744, 9783319105741}


\bibitem[\protect\citeauthoryear{Cousot}{Cousot}{2000}]%
        {cou00}
\bibfield{author}{\bibinfo{person}{Patrick Cousot}.}
  \bibinfo{year}{2000}\natexlab{}.
\newblock \showarticletitle{Partial Completeness of Abstract Fixpoint
  Checking}. In \bibinfo{booktitle}{\emph{Proceedings of the 4th International
  Symposium on Abstraction, Reformulation, and Approximation (SARA'02)}}.
  \bibinfo{publisher}{Springer-Verlag}, \bibinfo{pages}{1--25}.
\newblock
\showISBNx{3-540-67839-5}
\urldef\tempurl%
\url{https://doi.org/10.1007/3-540-44914-0_1}
\showURL{%
\tempurl}


\bibitem[\protect\citeauthoryear{Cousot and Cousot}{Cousot and Cousot}{1977}]%
        {CC77}
\bibfield{author}{\bibinfo{person}{Patrick Cousot} {and}
  \bibinfo{person}{Radhia Cousot}.} \bibinfo{year}{1977}\natexlab{}.
\newblock \showarticletitle{Abstract interpretation: a unified lattice model
  for static analysis of programs by construction or approximation of
  fixpoints}. In \bibinfo{booktitle}{\emph{Proceedings of the 4th ACM
  SIGACT-SIGPLAN Symposium on Principles of Programming Languages (POPL'77)}}.
  \bibinfo{publisher}{ACM Press}, \bibinfo{pages}{238--252}.
\newblock
\urldef\tempurl%
\url{http://doi.acm.org/10.1145/512950.512973}
\showURL{%
\tempurl}


\bibitem[\protect\citeauthoryear{Cousot and Cousot}{Cousot and Cousot}{1979}]%
        {CC79}
\bibfield{author}{\bibinfo{person}{Patrick Cousot} {and}
  \bibinfo{person}{Radhia Cousot}.} \bibinfo{year}{1979}\natexlab{}.
\newblock \showarticletitle{Systematic design of program analysis frameworks}.
  In \bibinfo{booktitle}{\emph{Proceedings of the 6th ACM SIGACT-SIGPLAN
  Symposium on Principles of Programming Languages (POPL'79)}}.
  \bibinfo{publisher}{ACM}, \bibinfo{address}{New York, NY, USA},
  \bibinfo{pages}{269--282}.
\newblock
\urldef\tempurl%
\url{https://doi.org/10.1145/567752.567778}
\showDOI{\tempurl}


\bibitem[\protect\citeauthoryear{de~Luca and Varricchio}{de~Luca and
  Varricchio}{1994}]%
        {deLuca1994}
\bibfield{author}{\bibinfo{person}{Aldo de Luca} {and} \bibinfo{person}{Stefano
  Varricchio}.} \bibinfo{year}{1994}\natexlab{}.
\newblock \showarticletitle{Well quasi-orders and regular languages}.
\newblock \bibinfo{journal}{\emph{Acta Informatica}} \bibinfo{volume}{31},
  \bibinfo{number}{6} (\bibinfo{year}{1994}), \bibinfo{pages}{539--557}.
\newblock
\showISSN{1432-0525}
\urldef\tempurl%
\url{https://doi.org/10.1007/BF01213206}
\showDOI{\tempurl}


\bibitem[\protect\citeauthoryear{de~Luca and Varricchio}{de~Luca and
  Varricchio}{2011}]%
        {deluca2011}
\bibfield{author}{\bibinfo{person}{Aldo de Luca} {and} \bibinfo{person}{Stefano
  Varricchio}.} \bibinfo{year}{2011}\natexlab{}.
\newblock \bibinfo{booktitle}{\emph{Finiteness and Regularity in Semigroups and
  Formal Languages}}.
\newblock \bibinfo{publisher}{Springer}.
\newblock
\showISBNx{3642641504, 9783642641503}
\urldef\tempurl%
\url{https://doi.org/10.1007/978-3-642-59849-4}
\showDOI{\tempurl}


\bibitem[\protect\citeauthoryear{{De Wulf}, Doyen, Henzinger, and Raskin}{{De
  Wulf} et~al\mbox{.}}{2006}]%
        {DBLP:conf/cav/WulfDHR06}
\bibfield{author}{\bibinfo{person}{Martin {De Wulf}}, \bibinfo{person}{Laurent
  Doyen}, \bibinfo{person}{Thomas~A. Henzinger}, {and}
  \bibinfo{person}{Jean{-}Fran{\c{c}}ois Raskin}.}
  \bibinfo{year}{2006}\natexlab{}.
\newblock \showarticletitle{Antichains: A New Algorithm for Checking
  Universality of Finite Automata}. In \bibinfo{booktitle}{\emph{Proceedings of
  the 18th International Conference on Computer Aided Verification (CAV'06)}}.
  \bibinfo{publisher}{Springer-Verlag}, \bibinfo{pages}{17--30}.
\newblock
\showISBNx{3-540-37406-X, 978-3-540-37406-0}
\urldef\tempurl%
\url{http://dx.doi.org/10.1007/11817963_5}
\showURL{%
\tempurl}


\bibitem[\protect\citeauthoryear{Ehrenfeucht, Haussler, and
  Rozenberg}{Ehrenfeucht et~al\mbox{.}}{1983}]%
        {ehrenfeucht_regularity_1983}
\bibfield{author}{\bibinfo{person}{Andrzej Ehrenfeucht}, \bibinfo{person}{David
  Haussler}, {and} \bibinfo{person}{Grzegorz Rozenberg}.}
  \bibinfo{year}{1983}\natexlab{}.
\newblock \showarticletitle{On regularity of context-free languages}.
\newblock \bibinfo{journal}{\emph{Theoretical Computer Science}}
  \bibinfo{volume}{27}, \bibinfo{number}{3} (\bibinfo{year}{1983}),
  \bibinfo{pages}{311--332}.
\newblock
\showISSN{0304-3975}
\urldef\tempurl%
\url{https://doi.org/10.1016/0304-3975(82)90124-4}
\showDOI{\tempurl}


\bibitem[\protect\citeauthoryear{Fiedor, Hol{\'\i}k, Leng{\'a}l, and
  Vojnar}{Fiedor et~al\mbox{.}}{2019}]%
        {fiedor2019nested}
\bibfield{author}{\bibinfo{person}{Tom{\'a}{\v{s}} Fiedor},
  \bibinfo{person}{Luk{\'a}{\v{s}} Hol{\'\i}k}, \bibinfo{person}{Ond{\v{r}}ej
  Leng{\'a}l}, {and} \bibinfo{person}{Tom{\'a}{\v{s}} Vojnar}.}
  \bibinfo{year}{2019}\natexlab{}.
\newblock \showarticletitle{Nested antichains for {W}{S}{1}{S}}.
\newblock \bibinfo{journal}{\emph{Acta Informatica}} \bibinfo{volume}{56},
  \bibinfo{number}{3} (\bibinfo{year}{2019}), \bibinfo{pages}{205--228}.
\newblock


\bibitem[\protect\citeauthoryear{Finkel and Schnoebelen}{Finkel and
  Schnoebelen}{2001}]%
        {Finkel2001}
\bibfield{author}{\bibinfo{person}{Alain Finkel} {and}
  \bibinfo{person}{Philippe Schnoebelen}.} \bibinfo{year}{2001}\natexlab{}.
\newblock \showarticletitle{Well-structured transition systems everywhere!}
\newblock \bibinfo{journal}{\emph{Theoretical Computer Science}}
  \bibinfo{volume}{256}, \bibinfo{number}{1-2} (\bibinfo{year}{2001}),
  \bibinfo{pages}{63--92}.
\newblock
\urldef\tempurl%
\url{https://doi.org/10.1016/s0304-3975(00)00102-x}
\showDOI{\tempurl}


\bibitem[\protect\citeauthoryear{Ganty, Ranzato, and Valero}{Ganty
  et~al\mbox{.}}{2019}]%
        {grv-sas2019}
\bibfield{author}{\bibinfo{person}{Pierre Ganty}, \bibinfo{person}{Francesco
  Ranzato}, {and} \bibinfo{person}{Pedro Valero}.}
  \bibinfo{year}{2019}\natexlab{}.
\newblock \showarticletitle{Language Inclusion Algorithms as Complete Abstract
  Interpretations}. In \bibinfo{booktitle}{\emph{Proc.\ of the 26th
  International Static Analysis Symposium (SAS'19), LNCS vol.~11822}},
  \bibfield{editor}{\bibinfo{person}{Bor-Yuh~Evan Chang}} (Ed.).
  \bibinfo{publisher}{Springer}, \bibinfo{pages}{140--161}.
\newblock
\showISBNx{978-3-030-32304-2}


\bibitem[\protect\citeauthoryear{Giacobazzi and Quintarelli}{Giacobazzi and
  Quintarelli}{2001}]%
        {gq01}
\bibfield{author}{\bibinfo{person}{Roberto Giacobazzi} {and}
  \bibinfo{person}{Elisa Quintarelli}.} \bibinfo{year}{2001}\natexlab{}.
\newblock \showarticletitle{Incompleteness, Counterexamples, and Refinements in
  Abstract Model-Checking}. In \bibinfo{booktitle}{\emph{Proceedings of the 8th
  Static Analysis Symposium (SAS'01), LNCS vol.~2126}}.
  \bibinfo{publisher}{Springer}, \bibinfo{pages}{356--373}.
\newblock
\showISBNx{978-3-540-47764-8}
\urldef\tempurl%
\url{https://doi.org/10.1007/3-540-47764-0_20}
\showURL{%
\tempurl}


\bibitem[\protect\citeauthoryear{Giacobazzi, Ranzato, and Scozzari}{Giacobazzi
  et~al\mbox{.}}{2000}]%
        {GiacobazziRS00}
\bibfield{author}{\bibinfo{person}{Roberto Giacobazzi},
  \bibinfo{person}{Francesco Ranzato}, {and} \bibinfo{person}{Francesca
  Scozzari}.} \bibinfo{year}{2000}\natexlab{}.
\newblock \showarticletitle{Making Abstract Interpretations Complete}.
\newblock \bibinfo{journal}{\emph{J. ACM}} \bibinfo{volume}{47},
  \bibinfo{number}{2} (\bibinfo{year}{2000}), \bibinfo{pages}{361--416}.
\newblock
\showISSN{0004-5411}
\urldef\tempurl%
\url{https://doi.org/10.1145/333979.333989}
\showDOI{\tempurl}


\bibitem[\protect\citeauthoryear{Ginsburg and Rice}{Ginsburg and Rice}{1962}]%
        {ginsburg}
\bibfield{author}{\bibinfo{person}{Seymour Ginsburg} {and}
  \bibinfo{person}{H.~Gordon Rice}.} \bibinfo{year}{1962}\natexlab{}.
\newblock \showarticletitle{Two Families of Languages Related to ALGOL}.
\newblock \bibinfo{journal}{\emph{J. ACM}} \bibinfo{volume}{9},
  \bibinfo{number}{3} (\bibinfo{date}{July} \bibinfo{year}{1962}),
  \bibinfo{pages}{350--371}.
\newblock
\showISSN{0004-5411}
\urldef\tempurl%
\url{https://doi.org/10.1145/321127.321132}
\showDOI{\tempurl}


\bibitem[\protect\citeauthoryear{Hofman, Mayr, and Totzke}{Hofman
  et~al\mbox{.}}{2013}]%
        {Hofman:2013:DWS:2591370.2591405}
\bibfield{author}{\bibinfo{person}{Piotr Hofman}, \bibinfo{person}{Richard
  Mayr}, {and} \bibinfo{person}{Patrick Totzke}.}
  \bibinfo{year}{2013}\natexlab{}.
\newblock \showarticletitle{Decidability of Weak Simulation on One-Counter
  Nets}. In \bibinfo{booktitle}{\emph{Proceedings of the 28th Annual ACM/IEEE
  Symposium on Logic in Computer Science}} \emph{(\bibinfo{series}{LICS'13})}.
  \bibinfo{publisher}{IEEE Computer Society}, \bibinfo{pages}{203--212}.
\newblock
\showISBNx{978-0-7695-5020-6}
\urldef\tempurl%
\url{https://doi.org/10.1109/LICS.2013.26}
\showDOI{\tempurl}


\bibitem[\protect\citeauthoryear{Hofman and Totzke}{Hofman and Totzke}{2018}]%
        {hofman_trace_2018}
\bibfield{author}{\bibinfo{person}{Piotr Hofman} {and} \bibinfo{person}{Patrick
  Totzke}.} \bibinfo{year}{2018}\natexlab{}.
\newblock \showarticletitle{Trace inclusion for one-counter nets revisited}.
\newblock \bibinfo{journal}{\emph{Theoretical Computer Science}}
  \bibinfo{volume}{735} (\bibinfo{date}{July} \bibinfo{year}{2018}),
  \bibinfo{pages}{50--63}.
\newblock
\showISSN{03043975}
\urldef\tempurl%
\url{https://doi.org/10.1016/j.tcs.2017.05.009}
\showDOI{\tempurl}


\bibitem[\protect\citeauthoryear{Hofmann and Chen}{Hofmann and Chen}{2014}]%
        {Hofmann2014}
\bibfield{author}{\bibinfo{person}{Martin Hofmann} {and} \bibinfo{person}{Wei
  Chen}.} \bibinfo{year}{2014}\natexlab{}.
\newblock \showarticletitle{Abstract interpretation from {B}\"{u}chi automata}.
  In \bibinfo{booktitle}{\emph{Proceedings of the Joint Meeting of the
  Twenty-Third {EACSL} Annual Conference on Computer Science Logic ({CSL}'14)
  and the Twenty-Ninth Annual {ACM}/{IEEE} Symposium on Logic in Computer
  Science ({LICS}'14)}}. \bibinfo{publisher}{{ACM} Press}.
\newblock
\urldef\tempurl%
\url{https://doi.org/10.1145/2603088.2603127}
\showDOI{\tempurl}


\bibitem[\protect\citeauthoryear{Hol{\'{\i}}k and Meyer}{Hol{\'{\i}}k and
  Meyer}{2015}]%
        {Holk2015}
\bibfield{author}{\bibinfo{person}{Luk{\'{a}}{\v{s}} Hol{\'{\i}}k} {and}
  \bibinfo{person}{Roland Meyer}.} \bibinfo{year}{2015}\natexlab{}.
\newblock \showarticletitle{Antichains for the Verification of Recursive
  Programs}.
\newblock In \bibinfo{booktitle}{\emph{Networked Systems}}.
  \bibinfo{publisher}{Springer International Publishing},
  \bibinfo{pages}{322--336}.
\newblock
\urldef\tempurl%
\url{https://doi.org/10.1007/978-3-319-26850-7\_22}
\showDOI{\tempurl}


\bibitem[\protect\citeauthoryear{Hopcroft and Ullman}{Hopcroft and
  Ullman}{1979}]%
        {HU79}
\bibfield{author}{\bibinfo{person}{John~E. Hopcroft} {and}
  \bibinfo{person}{Jeff~D. Ullman}.} \bibinfo{year}{1979}\natexlab{}.
\newblock \bibinfo{booktitle}{\emph{Introduction to Automata Theory, Languages,
  and Computation}}.
\newblock \bibinfo{publisher}{Addison-Wesley Publishing Company}.
\newblock


\bibitem[\protect\citeauthoryear{Hunt, Rosenkrantz, and Szymanski}{Hunt
  et~al\mbox{.}}{1976}]%
        {hunt}
\bibfield{author}{\bibinfo{person}{Harry~B. Hunt}, \bibinfo{person}{Daniel~J.
  Rosenkrantz}, {and} \bibinfo{person}{Thomas~G. Szymanski}.}
  \bibinfo{year}{1976}\natexlab{}.
\newblock \showarticletitle{On the equivalence, containment, and covering
  problems for the regular and context-free languages}.
\newblock \bibinfo{journal}{\emph{J. Comput. System Sci.}}
  \bibinfo{volume}{12}, \bibinfo{number}{2} (\bibinfo{year}{1976}),
  \bibinfo{pages}{222 -- 268}.
\newblock
\showISSN{0022-0000}
\urldef\tempurl%
\url{https://doi.org/10.1016/S0022-0000(76)80038-4}
\showDOI{\tempurl}


\bibitem[\protect\citeauthoryear{Jan{\v c}ar, Esparza, and Moller}{Jan{\v c}ar
  et~al\mbox{.}}{1999}]%
        {JANCAR1999476}
\bibfield{author}{\bibinfo{person}{Petr Jan{\v c}ar}, \bibinfo{person}{Javier
  Esparza}, {and} \bibinfo{person}{Faron Moller}.}
  \bibinfo{year}{1999}\natexlab{}.
\newblock \showarticletitle{Petri Nets and Regular Processes}.
\newblock \bibinfo{journal}{\emph{J. Comput. System Sci.}}
  \bibinfo{volume}{59}, \bibinfo{number}{3} (\bibinfo{year}{1999}),
  \bibinfo{pages}{476--503}.
\newblock
\showISSN{0022-0000}
\urldef\tempurl%
\url{https://doi.org/10.1006/jcss.1999.1643}
\showDOI{\tempurl}


\bibitem[\protect\citeauthoryear{Kozen}{Kozen}{1992}]%
        {Kozen92onthe}
\bibfield{author}{\bibinfo{person}{Dexter Kozen}.}
  \bibinfo{year}{1992}\natexlab{}.
\newblock \showarticletitle{On the {M}yhill-{N}erode Theorem for Trees}.
\newblock \bibinfo{journal}{\emph{Bulletin of the EATCS}}  \bibinfo{volume}{47}
  (\bibinfo{year}{1992}), \bibinfo{pages}{170--173}.
\newblock


\bibitem[\protect\citeauthoryear{Min{\'{e}}}{Min{\'{e}}}{2017}]%
        {mine17}
\bibfield{author}{\bibinfo{person}{Antoine Min{\'{e}}}.}
  \bibinfo{year}{2017}\natexlab{}.
\newblock \showarticletitle{Tutorial on Static Inference of Numeric Invariants
  by Abstract Interpretation}.
\newblock \bibinfo{journal}{\emph{Foundations and Trends in Programming
  Languages}} \bibinfo{volume}{4}, \bibinfo{number}{3-4}
  (\bibinfo{year}{2017}), \bibinfo{pages}{120--372}.
\newblock
\urldef\tempurl%
\url{https://doi.org/10.1561/2500000034}
\showDOI{\tempurl}


\bibitem[\protect\citeauthoryear{Ogawa}{Ogawa}{2004}]%
        {ogawa_well-quasi-orders_2004}
\bibfield{author}{\bibinfo{person}{Mizuhito Ogawa}.}
  \bibinfo{year}{2004}\natexlab{}.
\newblock \showarticletitle{Well-quasi-orders and regular $\omega$-languages}.
\newblock \bibinfo{journal}{\emph{Theoretical Computer Science}}
  \bibinfo{volume}{324}, \bibinfo{number}{1} (\bibinfo{year}{2004}),
  \bibinfo{pages}{55--60}.
\newblock
\showISSN{0304-3975}
\urldef\tempurl%
\url{https://doi.org/10.1016/j.tcs.2004.03.052}
\showDOI{\tempurl}


\bibitem[\protect\citeauthoryear{Park}{Park}{1969}]%
        {park1969fixpoint}
\bibfield{author}{\bibinfo{person}{David Park}.}
  \bibinfo{year}{1969}\natexlab{}.
\newblock \showarticletitle{Fixpoint induction and proofs of program
  properties}.
\newblock \bibinfo{journal}{\emph{Machine Intelligence}}  \bibinfo{volume}{5}
  (\bibinfo{year}{1969}).
\newblock


\bibitem[\protect\citeauthoryear{Ranzato}{Ranzato}{2013}]%
        {Ranzato13}
\bibfield{author}{\bibinfo{person}{Francesco Ranzato}.}
  \bibinfo{year}{2013}\natexlab{}.
\newblock \showarticletitle{Complete Abstractions Everywhere}. In
  \bibinfo{booktitle}{\emph{Proceedings of the 14th International Conference on
  Verification, Model Checking, and Abstract Interpretation}}
  \emph{(\bibinfo{series}{VMCAI'13})}, Vol.~\bibinfo{volume}{7737}.
  \bibinfo{publisher}{LNCS Springer}, \bibinfo{pages}{15--26}.
\newblock
\urldef\tempurl%
\url{https://doi.org/10.1007/978-3-642-35873-9\_3}
\showDOI{\tempurl}


\bibitem[\protect\citeauthoryear{Rival and Yi}{Rival and Yi}{2020}]%
        {rival-yi}
\bibfield{author}{\bibinfo{person}{Xavier Rival} {and}
  \bibinfo{person}{Kwangkeun Yi}.} \bibinfo{year}{2020}\natexlab{}.
\newblock \bibinfo{booktitle}{\emph{Introduction to Static Analysis: An
  Abstract Interpretation Perspective}}.
\newblock \bibinfo{publisher}{The {M}{I}{T} {P}ress}.
\newblock


\bibitem[\protect\citeauthoryear{Sakarovitch}{Sakarovitch}{2009}]%
        {Sakarovitch}
\bibfield{author}{\bibinfo{person}{Jacques Sakarovitch}.}
  \bibinfo{year}{2009}\natexlab{}.
\newblock \bibinfo{booktitle}{\emph{Elements of Automata Theory}}.
\newblock \bibinfo{publisher}{Cambridge University Press}.
\newblock
\showISBNx{978-0-521-84425-3}
\urldef\tempurl%
\url{https://doi.org/10.1017/CBO9781139195218}
\showDOI{\tempurl}


\bibitem[\protect\citeauthoryear{Sch{\"{u}}tzenberger}{Sch{\"{u}}tzenberger}{1963}]%
        {Schutzenberger63}
\bibfield{author}{\bibinfo{person}{Marcel~Paul Sch{\"{u}}tzenberger}.}
  \bibinfo{year}{1963}\natexlab{}.
\newblock \showarticletitle{On Context-Free Languages and Push-Down Automata}.
\newblock \bibinfo{journal}{\emph{Information and Control}}
  \bibinfo{volume}{6}, \bibinfo{number}{3} (\bibinfo{year}{1963}),
  \bibinfo{pages}{246--264}.
\newblock
\urldef\tempurl%
\url{https://doi.org/10.1016/S0019-9958(63)90306-1}
\showDOI{\tempurl}


\bibitem[\protect\citeauthoryear{Waite and Goos}{Waite and Goos}{1984}]%
        {waite84}
\bibfield{author}{\bibinfo{person}{William~M. Waite} {and}
  \bibinfo{person}{Gerhard Goos}.} \bibinfo{year}{1984}\natexlab{}.
\newblock \bibinfo{booktitle}{\emph{Compiler Construction}}.
\newblock \bibinfo{publisher}{Springer-Verlag}, \bibinfo{address}{New York,
  USA}.
\newblock
\showISBNx{0-387-90821-8}


\end{thebibliography}
\end{document}